\documentclass[a4paper,11pt,reqno]{amsart}

\usepackage{xcolor}
\usepackage[hidelinks]{hyperref}
\usepackage{amsmath,amssymb}
\usepackage{mathrsfs}  
\usepackage{graphicx}
\usepackage{multicol}
\usepackage{enumerate}
\usepackage{bbold}

\newtheorem{theorem}{Theorem}[section]
\newtheorem{lemma}[theorem]{Lemma}
\newtheorem{proposition}[theorem]{Proposition}

\theoremstyle{definition}
\newtheorem{remark}[theorem]{Remark}

\numberwithin{equation}{section}

\newcommand{\ran}{\operatorname{ran }}
\DeclareMathOperator{\sgn}{sgn}

\newcommand{\CO}{\mathbb C}

\newcommand{\RE}{\mathbb R}
\newcommand{\rr}{\mathbb R}

\newcommand{\Gb}{\mathbb G}
\newcommand{\Mb}{\mathbb M}

\newcommand{\BB}{\mathcal B}

\newcommand{\NN}{N}
\newcommand{\QQ}{\mathcal Q}
\newcommand{\PP}{\mathcal P}

\newcommand{\Tau}{\mathcal{T}}
\renewcommand{\SS}{\mathcal S}
\newcommand{\Fou}{\mathscr{F}}
\newcommand{\Bou}{\mathscr{B}}
\newcommand{\Hil}{\mathscr{H}}
\newcommand{\Lscr}{\mathscr{L}}
\newcommand{\Lrm}{\mathrm{L}}
\newcommand{\ve}{\varepsilon}

\renewcommand{\d}{\mathrm{d}}
\newcommand{\x}{{\bf  x}}
\renewcommand{\k}{{\bf k}}

\newcommand{\dk}{\mathrm{d}{\bf k}}

\newcommand{\brevegg}{\breve{g}}
\newcommand{\bfg}{{g}}
\newcommand{\mm}{{ m}}
\newcommand{\Qeff}{\mathrm{Q^{eff}_\ve}}

\newcommand{\eff}{\mathrm{eff}}

\renewcommand{\Im}{\operatorname{Im}}

\newcommand{\Ai}{\operatorname{Ai}}

\newcommand{\supp}{\operatorname{supp}}
\newcommand{\dist}{\operatorname{dist}}

\newcommand{\bof}{\text{\rm b/f}}
\newcommand{\bos}{\text{\rm b}}
\newcommand{\fer}{\text{\rm f}}
\newcommand{\bosfer}{{\natural}}
\newcommand{\bosferpm}{{(+)}}
\newcommand{\alve}{\ve}
\newcommand{\be}{\begin{equation}}
\newcommand{\ee}{\end{equation}}

\newcommand{\eigenvalueH}{E}
\newcommand{\eigenvalueL}{\mathscr{E}}

\title[]{The Born-Oppenheimer approximation for a 1D 2+1 particle system with zero-range  interactions}
\author{Claudio Cacciapuoti}
\email{claudio.cacciapuoti@uninsubria.it}%
\address{Universit\`a dell'Insubria, Dipartimento di Scienza e Alta Tecnologia, Sezione di Matematica, Via Valleggio 11, 22100 Como, Italy, EU}

\author{Andrea Posilicano}
\email{andrea.posilicano@uninsubria.it}%
\address{Universit\`a dell'Insubria, Dipartimento di Scienza e Alta Tecnologia, Sezione di Matematica, Via Valleggio 11, 22100 Como, Italy, EU}

\author{Hamidreza Saberbaghi}
\email{hamidreza.saberbaghi@uninsubria.it}%
\address{Universit\`a dell'Insubria, Dipartimento di Scienza e Alta Tecnologia, Sezione di Matematica, Via Valleggio 11, 22100 Como, Italy, EU}

\keywords{Born-Oppenheimer approximation, mathematical methods in quantum mechanics, three-body, point interactions, zero-range Interactions, one-dimensional.}
\subjclass{81Q10, 81Q15, 81Q20, 70F07, 46N50.}

\begin{document}

\begin{abstract}
We study the self-adjoint Hamiltonian 
that models the quantum dynamics of a one-dimensional (1D) three-body system consisting of a light particle interacting with two heavy ones through a zero-range force.  For an attractive interaction we determine the behavior of   the eigenvalues below the essential spectrum in the regime $\varepsilon\ll 1$, where $\varepsilon$ is proportional to the square root of the mass ratio. We show that the $n$-th eigenvalue behaves as 
$E_{n}(\varepsilon)=-\alpha^{2}+|\sigma_{n}|\,\alpha^{2}\varepsilon^{2/3}+O(\varepsilon)$,   
where $\alpha$ is a negative constant that explicitly relates to the physical parameters and  $\sigma_{n}$ is either the $n$-th extremum or the $n$-th zero of the Airy function Ai, depending on the kind (respectively, bosons or fermions) of the two heavy particles. Additionally, we prove that the essential spectrum coincides with the half-line $[-\frac{\alpha^2}{4+\ve^{2}}\,,+\infty)$. 
\end{abstract}

\maketitle

\section{Introduction}

We consider two heavy particles (both bosons or fermions) of mass $M$ and one light particle of mass $m$  in one  dimension. Let $x_1$ and $x_2$ denote the coordinates of the heavy particles, and $x_3$ denote the coordinate of the light particle.   The Hilbert space of the system is $L^2_{\bos}(\RE^3)$ or $L^2_{\fer}(\RE^3)$, i.e., the subspace of square integrable functions that are either symmetric (if the heavy particles are bosons) or antisymmetric (if the heavy particles are  fermions) under the exchange of the coordinates $x_1$ and $x_2$:
\[
L^2_{\bof}(\RE^3) := \{\Psi\in L^2(\RE^3):\; \Psi(x_1,x_2,x_3) =\bosferpm_{\bof}\, \Psi(x_2,x_1,x_3) \},
\]
where $\bosferpm_{\bos} := +$ and $\bosferpm_{\fer} := -\,$. 
 
We assume that the heavy particles do not interact between themselves and that they interact with the light particle through a contact interaction. We introduce, heuristically, the  Hamiltonian 
\[
H^{\bof}_{2+1} := -\frac{\hbar^2}{2M}\, \partial^2_{x_1} -\frac{\hbar^2}{2M} \,\partial^2_{x_2} -\frac{\hbar^2}{2m} \,\partial^2_{x_3} + \beta \delta(x_3 - x_1) + \beta\delta (x_3-x_2)\,.
\]
Here $\beta$ is real valued and represents the interaction parameter, and $\delta(x_3 - x_1)$, $\delta (x_3-x_2)$ denote the Dirac delta distributions supported on the  coincidence planes of the heavy particles with the light particle: 
\[
\Pi_1:= \{(x_1,x_2,x_3)\in \RE^3: x_3 = x_1\} \qquad \text{and} \qquad  \Pi_2:= \{(x_1,x_2,x_3)\in \RE^3|\; x_3 = x_2\}.
\]
To proceed, we pass to units in which $\hbar =1$ and define the Jacobi coordinates
\[
x_{cm} = \frac{M(x_1+x_2)+m x_3}{M_{tot}} \,,\qquad x= x_1 - x_2 \,,\qquad  y = x_3 - \frac{x_1+x_2}{2} ; 
\]
\[
M_{tot} = 2M+m\,,\qquad \mu = \frac{2Mm}{M_{tot}}\,.
\]
In Jacobi coordinates, the Hamiltonian $H^{\bof}_{2+1}$ acts in the Hilbert space of square-integrable functions $\Psi$ such that $\Psi(x_{cm},x,y) = \bosferpm_{\bof}\Psi(x_{cm},-x,y)$ and takes the form (we abuse the notation and continue to denote the Hamiltonian by $H^{\bof}_{2+1}$)
\[
H^{\bof}_{2+1} = -\frac{1}{2M_{tot}}\,\partial^2_{x_{cm}}-\frac{1}{M}\,\partial^2_{x}-\frac{1}{2\mu}\,\partial^2_{y} + \beta\delta(y-x/2)+ \beta\delta(y+x/2)\, .
\]
Neglecting the coordinates of the center of mass and keeping the previous notation, we consider the Hilbert space
\[
L^2_{\bof}(\RE^2) := \{\psi\in L^2(\RE^2):\; \psi(x,y) =\bosferpm_{\bof} \,\psi(-x,y) \}
\]
and the coincidence lines
\begin{equation}\label{Pi12}
\Pi_1= \{(x,y) \in \RE^2:\; y = x/2\} \qquad \text{and}\qquad  \Pi_2=  \{(x,y) \in \RE^2:\; y = - x/2\} .
\end{equation}
Additionally, since we consider the regime $m/M\ll 1$, we define the rescaled Hamiltonian 
\[
H^{\bof}_{\ve} = -\ve^2 \partial^2_{x}-\partial^2_{y} + \alpha \delta(y-x/2)+ \alpha\delta(y+x/2)
\]
where the small parameter $\ve\ll 1$ is defined by  
\[
\ve^2 = \frac{2\mu}{M}
\]  
and $\alpha$ is  kept fixed (independent of $\ve$); we remark that $H^{\bof}_\ve$ is to be considered as an operator in the Hilbert space 
$ L^2_{\bof}(\RE^2)$. 
At a heuristic level, the Hamiltonian $H^{\bof}_\ve$ is given by  $H^{\bof}_\ve = 2 \mu  \left(H^{\bof}_{2+1} -K_{cm}\right)$, where $K_{cm}= -\frac{1}{2M_{tot}}\,\partial^2_{x_{cm}}$  represents the kinetic energy of the center of mass and $\alpha$ relates to $\beta$ through the formula   $\alpha = 2\mu\beta$.

The operator   $H^{\bof}_\ve$ can be  characterized as the self-adjoint, bounded from below, operator associated to the closed and bounded from below quadratic form 
\[
\BB^{\bof}_{\ve}: H^1(\RE^2)\cap  L^2_{\bof}(\RE^2)\to\RE 
\]
\[
\BB^{\bof}_{\ve}(\psi)= \int_{\RE^2} \ve^2\left|\partial_x \psi(x,y)\right|^2+\left| \partial_y \psi(x,y)\right|^2  \d\x +2 \alpha
 \int_{\RE} \left|\psi(s,s/2) \right|^2   \d s, 
\]
see Proposition \ref{p:quadratic}.

We point out that functions in  the domain of $H^{\bof}_\ve$ have to satisfy a  boundary condition on the coincidence planes $\Pi_1$ and $\Pi_2$ (as defined in Eq. \eqref{Pi12}); more precisely,  they  are regular outside  $\Pi_1$ and $\Pi_2$ , continuous on $\Pi_1$ and $\Pi_2$, and 
\begin{equation}\label{explicitBC}
\Big(- \frac{ \ve^2}{2} \partial_x  \psi   +  \partial_y \psi \Big)(x,(x/2)^+) - 
\Big(- \frac{ \ve^2}{2} \partial_x  \psi   +  \partial_y \psi  \Big)(x,(x/2)^-) 
 = \alpha  \psi(x,x/2) \; \text{for a.e. $x\in\RE$},
\end{equation}
$(x/2)^\pm$ denoting the right (respectively, left) limit. A  similar condition holds true on the coincidence plane $\Pi_2$ and is  obtained as a consequence of the bosonic or fermionic symmetry. The left hand side in the previous equation represents (up to a normalization factor) the jump of the normal derivative relative to $ \ve^2 \partial^2_{x}+\partial^2_{y}$ across $\Pi_1$, see also Remark \ref{r:bc}.

Our main result is the following: 

 \begin{theorem}\label{th:main}
 Let $\alpha  \in \rr$ and $\ve>0$.  Then, 
\begin{equation}\label{spectrum1}
    \begin{aligned}
        & \sigma(H^{\bof}_\ve)=\sigma_{ess}(H^{\bof}_\ve) =[0,+\infty)  &&\text{if}\quad  \alpha \geq 0, \\
        & \sigma(H^{\bof}_\ve)\subseteq [-\alpha^{2},+\infty),\qquad   \sigma_{ess}(H^{\bof}_\ve)=\left[-\frac{\alpha^2}{4+\ve^{2}}\,,+\infty\right) 
         \qquad &&\text{if} \quad \alpha < 0.
    \end{aligned}
\end{equation}
Moreover, if  $\alpha<0$,  for any  fixed  integer $n\ge0$ there exists $\ve>0$ sufficiently small such that $H^{\bof}_\ve$ has at least $(n+1)$  simple isolated eigenvalues 
\[
-\alpha^{2}<\eigenvalueH^{\bof}_{\ve,0}<\eigenvalueH^{\bof}_{\ve,1}<\dots<\eigenvalueH^{\bof}_{\ve,n}<-\frac{\alpha^2}{4+\ve^{2}}\,,\]
such that 
\[
            \eigenvalueH^{\bof}_{\ve, k}  = -\alpha^2 +  s_{k}^{\bof}\alpha^2\ve^{2/3} + O (\ve) \, , \qquad \text{for all }k=0,\dots, n,
\]
    where 
    \[ s_{k}^{\bos}=-\sigma_{2k}\,,\qquad s_{k}^{\fer}=-\sigma_{2k+1}\,,\]
    and the interlacing negative numbers $\sigma_{k}$ 
    \[
    \dots<
    \sigma_{2k+1}<\sigma_{2k}<\sigma_{2k-1}<\dots<\sigma_{2}<\sigma_{1}<\sigma_{0}<0
    \,,
    \] 
    are either the  extrema or the zeros of the Airy function $\Ai$, i.e., $\Ai'(\sigma_{2k})=0$ or $ \Ai(\sigma_{2k+1})=0$.
\end{theorem}

The first part of Theorem~\ref{th:main} (Eq.~\eqref{spectrum1}) characterizes the bottom of the spectrum and the essential spectrum for arbitrary values of the mass ratio~$\varepsilon$. The second part establishes the validity of the Born--Oppenheimer approximation and provides the asymptotic behavior of the isolated eigenvalues in the small mass ratio limit.  

The Born-Oppenheimer approximation, introduced in the early years of quantum theory \cite{BO27}, was developed as a method for deriving the molecular structure from quantum mechanical principles. We do not attempt to provide an overview of the extensive mathematical literature on this topic. For an accessible introduction and references to both classical and recent developments, we refer the reader to the reviews \cite{HJ07, jecko14}.

In the context of zero-range interactions, the study of the quantum systems of three or more non relativistic bosons interacting via contact interactions in dimension three is plagued by the so-called ultraviolet catastrophe \cite{MF62a, MF62b, thomas}. This singular behavior also appears in systems consisting of two or more bosons or fermions interacting with a particle of a different nature (see, e.g., \cite{FT24a}, and \cite{Becker2018,CDFMT12,CDFMT15,CFT15, Yoshitomi17} for the fermionic case). Overcoming this issue typically requires the introduction of nonlocal modifications and/or effective three-body interaction terms, see, e.g., \cite{BCFTahp23, FT23,  FT24b, FT24a, FST24, Saberbaghi2025} and references therein. This problem  affects also the study of  a many particle system in three spatial dimensions within the Born-Oppenheimer framework, see, e.g., \cite{FST24}. 

However, these difficulties  do not arise whenever one considers particle systems in one spatial dimension, see, e.g.,  \cite{BCFTjmp18, GHL20}. Nevertheless, even in this case, the standard mathematical procedure to prove the validity of the Born-Oppenheimer approximation (as given, e.g., in \cite{CDS}) does not work, for this reason we use the general scheme developed in \cite{Krejcirik-etal-m2018}.

Let us remark that the eigenvalue expansion given in Theorem ~\ref{th:main} conforms with the one obtained (for the bosonic case) by Akbas, Turgut in  \cite{AT1}; their approach is more aligned with theoretical physics literature and does not present a rigorous proof of the validity of the Born-Oppenheimer approximation. In the same spirit a first insight into the two-dimensional case is provided in Ref.   \cite{AT2}.

Related to our work is the study of the spectrum of the  Laplacian in dimension two with   $\delta$-interactions supported on two crossing lines, see, e.g.,  \cite{ExnerNemcova03}, or on an almost straight line, see, e.g.,     \cite{ExnerIchinose01}; even though,  by the nature of the problem, in these works there is no account of the fermionic or bosonic symmetry. We mention also the work \cite{PankrashkinVogel22}, where the case of $\delta'$-interactions supported on two crossing lines is considered. 


After completing  our work, Nicholas Raymond drew our attention  to Ref. \cite{DucheneRaymond14}, addressing the analysis of the  discrete spectrum for the Laplacian in dimension two with a  $\delta$-interaction supported on a broken line, similar to  \cite{ExnerIchinose01} and \cite{ExnerNemcova03}.  The analysis in \cite{DucheneRaymond14} also makes use of a dimensional reduction argument and is closely related to the study of the Born-Oppenheimer approximation; by using bounds on the quadratic form, it results in an asymptotic expansion for the eigenvalues similar to the one given in Theorem~\ref{th:main}. 

\subsection*{Outline of the proof of Theorem \ref{th:main}}
We first define the Hamiltonian $H^{\bof}_\ve$ by means of standard tools from the theory of self-adjoint extensions of symmetric operators  and provide a formula for its  resolvent. This is done in Theorem \ref{th:resolvent-main}, within the approach developed in \cite{Pjfa01}. This allows a precise characterization of its essential spectrum which is also relevant for the proof of the second part of Theorem  \ref{th:main}.

The Born-Oppenheimer approximation is based on the idea that the dynamics of the system factorizes in a fast dynamics, relative to the light particle, and a slow one, describing the evolution of the heavy particles subsystem.  We fix the relative position $x$ of the heavy particles and study the spectrum of the light particle Hamiltonian $h_x$ associated to the quadratic form 
\[
b_{x}: H^1(\RE)  \to\RE
\]
\[
b_x (u) = \int_\RE |u'(y)|^2\, \d y + \alpha |u(x/2)|^2+\alpha |u(-x/2)|^2\,.
\]
 Note that $h_x$ is the Hamiltonian in dimension one with two delta interactions centered in $y=x/2$ and $y=-x/2$, see \cite{AGHKH05}. Functions in the domain of $h_x$ are regular outside the points $y=\pm x/2$, continuous in $y = \pm x/2$,  and satisfy the boundary conditions 
\[
u'((x/2)^+) -u'((x/2)^-) =\alpha u(x/2)
\qquad\text{and}\qquad 
u'((-x/2)^+)-u'((-x/2)^-)=\alpha u(-x/2).
\]
For $\alpha<0$, $h_x$ has non empty discrete spectrum and the lowest eigenvalue, denoted by $-\lambda_0(x)$, and the corresponding normalized eigenfunction, $ \psi^{BO}_{x}$, can be explicitly computed. 


To extract an effective contribution from the light particle component, we use the projection on the eigenfunction $\psi^{BO}_{x}$ (more precisely, the direct integral of a family of projections). This procedure produces an effective Hamiltonian for the heavy particles subsystem. We follow closely the  very versatile and general approach developed in \cite{Krejcirik-etal-m2018}. By doing so we reproduce and rewrite, in a slightly different way, some key estimates from that paper, giving simpler expressions for the bounds (see Propositions \ref{p:QwidehatQ} and \ref{p:diff}, and Lemma \ref{l:resolvent-difference} compared with \cite[Theorems 1.1 and 2.5, and Proposition 2.6]{Krejcirik-etal-m2018}).

To relate the Hamiltonians  $H^{\bof}_\ve$ and $h_x$, we observe that given any function $\phi\in H^1(\RE^2)$, 
for a.e. $x\in \RE$, its $x$-section $\phi_x(y):= \phi(x,y)$ belongs to $H^1(\RE)$. Moreover,  there holds 
\[
\BB^{\bof}_{\ve}( \phi)= \int_{\RE^2} \ve^2 |\partial_x \phi(x,y)|^2 \d\x  +\int_\RE  b_x(\phi_x)\d x \qquad \forall \phi \in H_{\bof}^1(\RE^2).    
\]
As a result, for any arbitrary mass ratio $\ve$ and $\alpha <0$ we infer  $\sigma(H^{\bof}_\ve)\subset [-\alpha^2,+\infty)$, see Proposition \ref{p:spectrumHveLB}.


To proceed further, we  notice that  $\psi^{BO}_x$  can be regarded as a function of two variables, denoted by  $\psi^{BO}$, with the obvious identification $\psi^{BO}(x,y)\equiv \psi^{BO}_x(y)$;  we introduce the orthogonal  projection 
\[
\mathcal{P}:L^{2}(\RE^{2})\to L^{2}(\RE^{2})\,,\qquad
\mathcal{P} \phi(x,y) := 
\psi^{BO}(x,y) f_\phi(x) \, ,
\]
 where   
 \[
    f_\phi(x) := \int_{\RE} {\psi^{BO}(x,y)} \phi(x,y)\, \d y  \,;
    \]
    additionally we set $\mathcal{P}^\perp  :=1 - \mathcal P$ and point out that $\mathcal P$ leaves invariant both $ L^{2}_{\bos}(\RE^{2})$ and  $L^{2}_{\fer}(\RE^{2})$.
    
It turns out  that  the quadratic form 
\[
 D(\widehat{\BB}^{\bof}_\ve): =   H_{\bof}^1(\RE^2) \qquad \widehat{\BB}^{\bof}_\ve(\phi)  : = 
 {\BB}^{\bof}_\ve (\PP\phi)+ {\BB}^{\bof}_\ve(\PP^\perp\phi)
\]
is well-defined, closed and bounded from below (see Remarks \ref{remL2} and \ref{remH1}), and so,  it defines a self-adjoint operator $\widehat H^{\bof}_\ve$.

 We provide estimates regarding the relations between the resolvent sets of  $\widehat H^{\bof}_\ve$ and $H^{\bof}_\ve$ and the difference between their resolvents (see Lemma \ref{l:resolvent-difference} for the detailed statements).  Note that for technical reasons, in Sections \ref{s:reduction}, \ref{s:5} and \ref{s:6} we prefer to work with positive definite quantities. To this end, we shift quadratic forms and operators by $\alpha^2$. Specifically, all operators denoted by  $H$ and $\Lscr$ -- regardless of superscripts or subscripts -- satisfy the relation  $\Lscr = H +\alpha^2$. An analogous convention holds for quadratic forms:  $\QQ = \BB+\alpha^2$ and $q = b+\alpha^2$. Translating the results back to the original setting is straightforward.

\par
Next we need to identify the effective Hamiltonian for the heavy particles subsystem. The operator $\widehat H^{\bof}_\ve$ can be written as a direct sum with the following spectrum:
\[\widehat H^{\bof}_\ve = \widehat H^{\bof}_{\ve,\PP}\oplus \widehat H^{\bof}_{\ve,\PP^\perp}\,,\qquad
\sigma(\widehat H^{\bof}_{\ve,\PP})  \subseteq[-\alpha^2,+\infty)\,,\qquad
\sigma(\widehat H^{\bof}_{\ve,\PP^\perp}) \subseteq[-\alpha^2/4,+\infty)\,.
\] 
For the study of the eigenvalues at the bottom of the spectrum (near $-\alpha^2$), the most relevant operator is $\widehat H^{\bof}_{\ve,\PP}$. We observe that, by means of a unitary map, the Hamiltonian  $\widehat H^{\bof}_{\ve,\PP}$ can be reduced to an effective  one dimensional operator on the heavy particles subsystem 
\[      D(H^{\eff\, b/f }_\ve) =H^{2}(\RE)\cap L^{2}_{\bof}(\RE)\qquad  
       H^{\eff\, b/f}_\ve  
       = -\ve^2 \frac{d^2}{dx^2}-\lambda_0 + \ve^2 \,R,
  \]
where $-\lambda_0(x)$ is the lowest eigenvalue of $h_x$ and $ \ve^2 R(x) := \ve^2 \int_\rr |\partial_x \psi^{BO}(x,y)|^2 dy$ is a perturbative potential term (see Remark \ref{r:UBO}). We remark that unlike the smooth potential case, $\ran(\PP| D( H^{\bof}_{\ve}))  \nsubseteq D( H^{\bof}_{\ve})$, because $\mathcal{P} \phi$ does not satisfy the boundary condition \eqref{explicitBC}. For this reason, contrarily to what is done in \cite{Krejcirik-etal-m2018},  it is not possible to identify $\widehat H^{\bof}_{\ve,\PP}$ with the compression $\PP H^{\bof}_{\ve}\PP$, which is not well defined. 

Finally, to conclude the proof of the second part of Theorem \ref{th:main}, we prove that for any  fixed  integer $n\geq0$ there exists $\ve>0$ sufficiently small such that $H^{\eff\, \bos/\fer}_{\ve}$ has at least $(n+1)$  simple isolated eigenvalues which satisfy 
\[
            \eigenvalueH^{\eff\, \bos/\fer}_{\ve, k} = -\alpha^2 + s^{\bos/\fer }_{k}\alpha^2 \ve^{2/3} + O (\ve)  \, , \qquad \text{for all }k=0,\dots, n,
\]
where $s^{\bof}_{k}$ are the same numbers as in Theorem \ref{th:main}. This result follows immediately from Theorem \ref{th:eigenvalues} shifting the spectrum by the constant $\alpha^2$. The proof of Theorem \ref{th:eigenvalues} is based on the seminal paper \cite{Sim}. We remark that with  respect to the case studied in \cite{Sim}, where the potential is smooth, in our analysis the potential term $-\lambda_0$ is only piecewise smooth, it is continuous, but not differentiable at  $x=0$. Additionally,  $-\lambda_0$ is linear around $x=0$, rather than quadratic, as in the smooth case. For this reason the eigenvalue expansion involves zeroes and extrema of the Airy function, a result already pointed out in \cite{AT1} and \cite{DucheneRaymond14}.  
\\

The paper  is organized as follows. In Section \ref{s:Hamiltonian}, we characterize the Hamiltonian $H^{\bof}_\ve$ and its essential spectrum. In Section \ref{s:light}, we study the Hamiltonian of the light particle, denoted by $h_x$, and we give a complete description of its spectrum. In Section \ref{s:reduction}, we carry out the dimensional reduction following the approach of \cite{Krejcirik-etal-m2018}. In Section \ref{s:5}, we analyze the effective Hamiltonian $\Lscr^{\eff,\bof}_\ve := H^{\eff,\bof}_\ve + \alpha^2$ and its eigenvalues below its essential spectrum. Finally, Section \ref{s:6} contains the proof of Theorem \ref{th:main}.

\subsection*{Notation}
\begin{itemize}
\item[-] $\x = (x,y) \in \RE^2$.
\item[-] $\k = (k,p) \in \RE^2$.
\item[-] We denote by $\Fou f$ or $\hat f $ the Fourier transform of $f$, defined as:
\[
\hat \psi (\k) := \frac{1}{2\pi} \int_{\RE^2} e^{- i\k \cdot\x} \psi(\x) \d\x ,
\]
or 
\[
\hat  \xi (p) := \frac{1}{\sqrt{2\pi}} \int_{\RE} e^{- ipy } \xi(y)  \d y.
\]
The inverse Fourier transform is denoted by   $\Fou^{-1} f$ or $\check f$. 
\item[-] The $L^2(\RE^n)$-norm and scalar product are denoted simply by $\|\cdot\|$ and  $\langle\cdot,\cdot\rangle$ respectively; norm and scalar product in different Hilbert spaces are denoted by an appropriate  subscript.
\item[-] The symbol $\bosfer$ denotes either $\bos$ (bosonic) or $\fer$ (fermionic), with the corresponding value $\bosferpm_\bosfer$ defined as $\bosferpm_\bos = +$ and $\bosferpm_\fer = -$.
\item[-] $H^\nu(\RE^n)$,  $\nu\in\RE $, denotes the usual Sobolev space of order $\nu$;
$H^\nu_{\bosfer}(\RE^2) := H^\nu(\RE^2) \cap L^2_{\bosfer}( \RE^2)$, $\nu>0$.
\item[-] $H^{\nu-}:=\cap_{a <\nu}H^{a}$.
\item[-] $\langle\cdot,\cdot\rangle_{\mp\nu,\pm\nu}$ denotes the (anti-linear with respect to the first variable) $H^{\pm\nu}(\RE^{n})$-$H^{\mp\nu}(\RE^{n})$ duality pairing extending the  scalar product in $L^{2}(\RE^{n})$.
\item[-]  $\SS(\RE^{n})$ denotes the space of Schwartz functions.
\item[-] $\Bou(X,Y)$ denotes the Banach space of bounded linear operators between the two Hilbert spaces $X$ and $Y$; we use the shorthand notation $\Bou(X) = \Bou(X,X)$. The operator norm is denoted by $\|\cdot\|$.
\item[-] $D(L)$ denotes the domain of the linear operator $L$; $\ker(L)$, $\ran(L)$ denote its kernel and range respectively.
\item[-] $\rho(L)$, $\sigma(L)$ denote the resolvent and the spectra of $L$. 
\item[-] $\sigma_{ac}(L)$, $\sigma_{sc}(L)$, $\sigma_{ess}(L)$, $\sigma_{p}(L)$, $\sigma_{d}(L)$ denote the absolutely continuous, singular continuous, essential, point and discrete spectra of $L$.
\item[-] $L|V$ denotes the restriction of the linear operator $L$ to the linear subspace $V\subset D(L)$.
\item[-] If $\mathcal Q$ is a sesquilinear form in a Hilbert space $X$, we use the same notation for the associated quadratic form $\mathcal Q(\psi) = \mathcal Q(\psi,\psi)$.  
\item[-] $C^{\infty}_{0}(\RE^n)$ denotes the set of smooth and compactly supported functions from $\RE^n$ to $\CO$. 
\item[-] $\lambda_{+}:=\max\{0,\lambda\}$.
\end{itemize}

\section{ $H^\bosfer_\ve$ as a self-adjoint extension}\label{s:Hamiltonian}
In this Section we start working in the usual Hilbert space $L^2(\RE^2)$. 
Our first aim is to construct the resolvent of a one-parameter  family of self-adjoint extensions of the restriction of the free Hamiltonian 
\[
D(H^{0}_{\ve}) := H^2(\RE^2)\;,\qquad
H^{0}_{\ve} :=-\ve^2 \partial^2_{x}-\partial^2_{y} , 
\]
to the subspace of functions vanishing on the  contact subset $\Pi= \Pi_1\cup \Pi_2$.  This family models interacting Hamiltonians describing zero-range interactions  between the particle $1$ and $3$, and between the particle $2$ and $3$, without taking into account either bosonic or fermionic symmetry. Successively, in Section \ref{ss:2.2},  we compress such self-adjoint Hamiltonians onto the subspace 
\[
L^2_{\bosfer}(\RE^2) := \{\psi\in L^2(\RE^2):\; \psi(x,y) =(+)_{\bosfer} \,\psi(-x,y) \}\,,
\]
where 
$\bosfer={\bos}$ in the bosonic case,  $\bosfer={\fer}$ in the fermionic case, and $(+)_{\bos}=+$, $(+)_{\fer}=-\,$. 
This procedure gives  Hamiltonians modeling the same zero-range interactions with the additional constraint that 
particles $1$ and $2$ are either bosons or fermions. Such Hamiltonians correspond to sesquilinear forms in $L^{2}_{\bosfer}(\RE^{2})$, with domain $H^{1}_{\bosfer}(\RE^{2})=H^{1}(\RE^{2})\cap L^{2}_{\bosfer}(\RE^{2})$, of the kind 
\begin{equation*}
\BB^\bosfer_{\ve}(\varphi, \psi):= \int_{\RE^2} \ve^2 {\partial_x \overline\varphi(x,y)}\partial_x \psi(x,y) + {\partial_y \overline\varphi(x,y)} \partial_y \psi(x,y) \d\x +2 \alpha
 \int_{\RE}  {\overline\varphi(s,s/2)} \, \psi(s,s/2)  \d s.
\end{equation*}

\subsection{Building a resolvent.}
For any  $\phi \in \SS(\RE^2)$ we define $\tau_1$ (resp. $\tau_2$) as the trace of $\phi$ on $\Pi_1$ (resp. $\Pi_2$); explicitly,
\[
(\tau_1\phi)(s) := \phi(s, s/2) \,, \qquad (\tau_2\phi)(s) := \phi(-s,   s/2)\,, \qquad s\in\RE .
\]
The maps $\tau_1$ and $\tau_2$ have unique extensions to bounded and surjective linear operators (which we denote by the same symbols)  
\[
\tau_1: H^\nu(\RE^2) \to H^{\nu-1/2}(\RE) \,, \qquad
\tau_2:H^\nu(\RE^2)\to H^{\nu-1/2}(\RE) 
\]
for any $\nu>1/2$. Note that  $C_0^\infty(\RE ) \subset\ran (\tau_j)$ and $C_0^\infty(\RE^2\backslash\Pi_{j} ) \subset\ker (\tau_j)$ are both dense in $L^2(\RE^2)$, $j=1,2$.

We define  the bounded linear operator 
\[
\Tau : H^\nu(\RE^2) \to H^{\nu-1/2}(\RE) \oplus H^{\nu-1/2}(\RE)\,,\qquad\Tau \phi := \tau_1\phi\oplus\tau_2\phi\,,\qquad \nu >1/2 \,.
\]
Note that $C_0^\infty(\RE\backslash\{0\})\times C_0^\infty(\RE\backslash\{0\}) \subset \ran(\Tau)$ is dense in  $L^2(\RE)\oplus L^2(\RE)$ and $C_0^\infty(\RE^2\backslash\Pi ) \subset \ker (\Tau)$ is dense in  $L^2(\RE^2)$.
\begin{remark}\label{r:Tau} The linear operator $
\Tau$ is  not surjective whenever $\nu>1$; indeed, 
\[
\ran(\Tau)\subseteq \big\{\xi_1\oplus\xi_2 \in H^{\nu-1/2}(\RE)\oplus H^{\nu-1/2}(\RE): \,\xi_1(0) = \xi_2(0)\big\}\subsetneqq H^{\nu-1/2}(\RE)\oplus H^{\nu-1/2}(\RE)\,.
\]
Therefore, our successive construction of a self-adjoint extension of the symmetric operator $H^{0}_{\ve}|\ker(\Tau)$ does not embed into the framework of standard boundary triples theory. In particular, this prevents the use of the spectral results given, e.g., in \cite[Theorem 3.3]{BGP}. We provide analogous results adapted to our framework in Lemma \ref{l:spectrum} below.
\end{remark}

Recalling that the resolvent set $\rho(H^{0}_{\ve})$ of the self-adjoint operator $H^{0}_{\ve}$, is $\CO\backslash [0,+\infty)$, for any $z \in\CO\backslash [0,+\infty)$ we define the resolvent operator 
\[
R^{0}_{\ve}(z): = (H^{0}_{\ve} - z)^{-1}\,.
\] Obviously 
\[
R^{0}_{\ve}(z) : H^{\nu}(\RE^2) \to H^{\nu+2}(\RE^2)\,,\qquad \nu\ge0\,,
\]
is a continuous bijection for all $z\in\CO\backslash [0,+\infty)$ and it extends to a continuous bijection (which we denote by the same symbol),
\[
R^{0}_{\ve}(z) : H^{\nu}(\RE^2) \to H^{\nu+2}(\RE^2)\,,\qquad \nu<0\,.
\]
For any $z\in\CO\backslash [0,+\infty)$ we define the  bounded operator 
\begin{equation*}
\breve \Gb_{\ve}(z) : H^\nu(\RE^2) \to H^{\nu+3/2}(\RE) \oplus  H^{\nu+3/2}(\RE) \,,\qquad \breve \Gb_{\ve}(z) := \Tau R^{0}_{\ve}(z)\,,\qquad \nu>-3/2\,.
\end{equation*}
One has 
\[
\breve \Gb_{\ve}(z) \psi=   \breve G_{1,\ve} (z)\psi\oplus\breve G_{2,\ve}(z)\psi\, ; \qquad \breve G_{j,\ve}(z) :  H^\nu(\RE^2) \to H^{\nu+3/2}(\RE)\,, \qquad  \breve G_{j,\ve}(z) = \tau_j R^{0}_{\ve}(z)\,,  \quad j =1,2.
\]
Note  that $\ran(\breve G_{j,\ve}(z)) = H^{\nu+3/2}(\RE)$; however, by Remark \ref{r:Tau}, $\ran(\breve \Gb_{\ve}(z))\subsetneqq H^{\nu+3/2}(\RE)\oplus H^{\nu+3/2}(\RE)$ whenever $\nu>-1$.

We define the bounded operator 
\begin{equation*}
 \Gb_{\ve}(z) : = \breve \Gb_{\ve}(\bar z)^* : H^{\nu}(\RE) \oplus  H^{\nu}(\RE)  \to  H^{\nu+3/2}(\RE^2)\,,\qquad \nu<0\,,
\end{equation*}
where the adjoint is defined in terms of the $H^{-\nu}(\RE^{d})$-$H^{\nu}(\RE^{d})$  duality (taken to be anti-linear with respect to the first variable) which extends the $L^{2}(\RE^{d})$ scalar product.
We  remark that $\Gb_{\ve}(z)$ is also represented as 
\begin{equation*}
\Gb_{\ve}(z) (\xi_{1}\oplus\xi_{2})= G_{1,\ve}(z)\xi_{1}+G_{2,\ve}(z)\xi_{2}; \qquad G_{j,\ve}(z) : H^\nu(\RE)\to  H^{\nu+3/2}(\RE^2),  \qquad  G_{j,\ve}(z) = \breve G_{j,\ve}(\bar z)^*\,.
\end{equation*}
In particular, for all  $z\in\CO\backslash [0,+\infty)$ and for all $\nu<3/2$ there holds 
\begin{equation*}
G_{j,\ve}(z) \in\Bou( L^2(\RE), H^{\nu}(\RE^2))\,,\qquad \Gb_{\ve}(z) \in\Bou( L^2(\RE)\oplus L^2(\RE), H^{\nu}(\RE^2))\,,
\end{equation*}
and so
\begin{equation}\label{ran}
\ran(G_{j,\ve}(z)|L^2(\RE))\subseteq H^{3/2-}(\RE^2)\,,\qquad \ran(\Gb_{\ve}(z)| L^2(\RE)\oplus L^2(\RE))\subseteq H^{3/2-}(\RE^2)\,,
\end{equation}
where $H^{3/2-}(\RE^2):=\cap_{\nu<3/2}H^{\nu}(\RE^2)$.

\begin{remark}\label{r:GcapDH}Since $\ker (\Tau)$ is dense in $L^2(\RE^2)$, there follows that 
\begin{equation}\label{empty_inter}
\ran(\Gb_\ve(z)|L^{2}(\RE) \oplus  L^{2}(\RE) ) \cap D(H^{0}_{\ve}) = \{0\},
\end{equation}
see \cite[Rem. 2.8, 2.9 and Theorem 2.1]{Pjfa01}. Indeed, suppose that \eqref{empty_inter} is false. Then there exist $Q \in (L^{2}(\RE)\backslash\{0\}) \oplus  (L^{2}(\RE)\backslash\{0\}) $ and $\phi\in L^2(\RE^2)\backslash\{0\}$,  such that $R^{0}_{\ve}(z)  \phi = \Gb_\ve(z) Q$. Hence, for all  $\psi\in L^2(\RE^2)$ one would have 
\[
\langle\phi, R^{0}_{\ve}(\bar z)  \psi \rangle_{L^2(\RE^2)} = \langle Q, \Tau R^{0}_{\ve}(\bar z) \psi \rangle_{ L^{2}(\RE) \oplus  L^{2}(\RE) }. 
\]
Since, $R^{0}_{\ve}(\bar z)$ maps $L^2(\RE^2)$  onto $H^2(\RE^2)$ and the set $\{f\in H^2(\RE^2): \, \Tau f = 0\}$ is dense in $L^2(\RE^2)$, there follows $\phi =0$, leading to a contradiction. 
\end{remark}

\begin{remark}\label{r:op-in-fou-transf}We note the following expressions for  the integral kernels of the relevant operators in Fourier transform:
\[
R^{0}_{\ve}(\x,\x';z)  =\frac{1}{(2\pi)^2}  \int_{\RE^2} \frac{ e^{i\k\cdot (\x -\x') }}{\ve^2 k^2 +p^2  - z}   \dk ;
\]
\[
\breve G_{1,\ve}( s,\x'; z)  =   \frac{1}{(2\pi)^2}  \int_{\RE^2} \frac{ e^{i (k + \frac{p}2) s} e^{- i\k\cdot \x' }}{\ve^2 k^2 +p^2  - z}   \dk ;
\qquad 
\breve G_{2,\ve}( s,\x'; z)  =   \frac{1}{(2\pi)^2}  \int_{\RE^2} \frac{ e^{i (-k + \frac{p}2) s} e^{- i\k \cdot \x' }}{\ve^2 k^2 +p^2  - z}   \dk ;
\]
\[
G_{1,\ve}( \x, s'; z)  =   \frac{1}{(2\pi)^2}  \int_{\RE^2} \frac{e^{i\k \cdot \x }  e^{-i (k + \frac{p}2) s'} }{\ve^2 k^2 +p^2  - z}   \dk ; 
\qquad 
G_{2, \ve}( \x, s'; z)  =   \frac{1}{(2\pi)^2}  \int_{\RE^2} \frac{e^{i\k \cdot \x }  e^{-i (-k + \frac{p}2) s'} }{\ve^2 k^2 +p^2  - z}   \dk . 
\]
\end{remark}

By the properties of the operators $\tau_1$, $\tau_2$ and $\Tau$, and by the mapping properties of the operators $\Gb_{\ve}(z)$ and $G_{j,\ve}(z)$ there follows that the following operators are well defined and bounded for any $z\in\CO\backslash [0,+\infty)$: 
\[
M_{\ell j,\ve}(z) :  H^{\nu}(\RE) \to  H^{\nu+1}(\RE)\;,\qquad  M_{\ell j,\ve}(z)  :=  \tau_{\ell} G_{j,\ve}(z) \qquad \ell,j = 1,2;\qquad \nu>-1\, .
\]
and 
\[
\Mb_{\ve}(z) :  H^{\nu}(\RE)\oplus H^{\nu}(\RE) \to  H^{\nu+1}(\RE)\oplus H^{\nu+1}(\RE)\;,\qquad  \Mb_{\ve}(z)  := \Tau \Gb_{\ve}(z)\,,\qquad \nu>-1\, .
\]
In particular, $M_{\ell j,\ve}(z) \in \Bou(L^2(\RE))$  and $\Mb_{\ve}(z) \in \Bou(L^2(\RE)\oplus L^2(\RE))$. 

$\Mb_{\ve}(z) $ can be represented as the block operator matrix 
\[
\Mb_{\ve}(z)   =	\begin{bmatrix}
			M_{11,\ve}(z) & M_{12,\ve}(z)\\ 
			M_{21,\ve}(z) &M_{22,\ve}(z)
			\end{bmatrix}.
\]
\begin{lemma}\label{rem}
For any $z,w\in \CO\backslash[0,+\infty)$ there holds
\begin{equation}\label{resid2}
\breve \Gb_{\ve} (z)  - \breve \Gb_\ve (w) = (z-w) \breve \Gb_{\ve} (z) R^{0}_{\ve} (w) = (z-w)\breve \Gb_{\ve} (w)R^{0}_{\ve} (z)\,,
\end{equation}
\begin{equation}\label{resid3}
 \Gb_{\ve} (z)  -  \Gb_\ve (w) = (z-w)  R^{0}_{\ve} (w)  \Gb_{\ve} (z) = (z-w)  R^{0}_{\ve} (z)\Gb_{\ve} (w)\,,
\end{equation}
\begin{equation}\label{Mb_prop2}
 \Mb_{\ve} (z)  -  \Mb_\ve (w) = (z-w)  \breve \Gb_{\ve} (w)  \Gb_{\ve} (z) = (z-w)  \breve \Gb_{\ve}  (z)\Gb_{\ve} (w).
\end{equation}
\end{lemma} 
\begin{proof} The relations \eqref{resid2} follow by applying $\Tau$ to the  resolvent identity
\begin{equation}\label{resid1}
R^{0}_{\ve} (z)  - R^{0}_{\ve} (w) = (z-w) R^{0}_{\ve} (z)R^{0}_{\ve} (w) = (z-w)R^{0}_{\ve} (w)R^{0}_{\ve} (z)\,.
\end{equation}
The relations \eqref{resid3} follow by evaluating \eqref{resid2} in $\bar z$ and $\bar w$ and taking the adjoint.
Finally,  by applying $\Tau$ to \eqref{resid3}, one obtains \eqref{Mb_prop2}.
\end{proof}
From now on, the operators $\Gb_{\ve}  (z)$, $\breve \Gb_{\ve} (z)$ $\Mb_{\ve} (z)$ are to be intended as bounded operators acting from $L^{2}$-spaces to $L^{2}$-spaces and with the previously stated regularity properties whenever restricted to smaller subspaces.\par

We note the following expressions for  the integral kernels of the  operators $M_{\ell j,\ve}(z)$ in Fourier transform  (as usual we denote the integral kernels and the operators by the same symbol):
\[
M_{11,\ve}( s, s'; z)  =   \frac{1}{(2\pi)^2}  \int_{\RE^2} \frac{  e^{i (k + \frac{p}2) (s-s')} }{\ve^2 k^2 +p^2  - z}   \dk ; 
\qquad 
M_{12, \ve}( s, s'; z)  =   \frac{1}{(2\pi)^2}  \int_{\RE^2} \frac{e^{i (k + \frac{p}2) s}  e^{-i (-k + \frac{p}2) s'} }{\ve^2 k^2 +p^2  - z}   \dk ;
\]
\[
M_{21,\ve}( s, s'; z)  =   \frac{1}{(2\pi)^2}  \int_{\RE^2} \frac{e^{i (-k + \frac{p}2) s} e^{-i (k + \frac{p}2) s'} }{\ve^2 k^2 +p^2  - z}   \dk ; 
\qquad 
M_{22,\ve}( s, s'; z)  =   \frac{1}{(2\pi)^2}  \int_{\RE^2} \frac{ e^{i (-k + \frac{p}2) (s-s')} }{\ve^2 k^2 +p^2  - z}   \dk . 
\]
Indeed, by the changing variable $k \to -k$, it is easy to convince oneself that  $M_{11,\ve}( s, s'; z)  = M_{22,\ve}( s, s'; z)  $ and $M_{12, \ve}( s, s'; z)  =M_{21,\ve}( s, s'; z)$,  for this reason we introduce the notation 
\[
M_{d,\ve}( s, s'; z) := M_{11,\ve}( s, s'; z) = M_{22,\ve}( s, s'; z)  =   \frac{1}{(2\pi)^2}  \int_{\RE^2} \frac{  e^{i (k + \frac{p}2) (s-s')} }{\ve^2 k^2 +p^2  - z}   \dk ; 
\]
\[
M_{od,\ve} ( s, s'; z) := M_{12, \ve}( s, s'; z)  =M_{21,\ve}( s, s'; z)    =   \frac{1}{(2\pi)^2}  \int_{\RE^2} \frac{e^{i (k + \frac{p}2) s}  e^{-i (-k + \frac{p}2) s'} }{\ve^2 k^2 +p^2  - z}   \dk .
\]
The suffixes ``$d$'' and ``$od$'' stand for ``{\it diagonal}'' and ``{\it off-diagonal}'' respectively. With this notation one has 
\[
\Mb_{\ve}(z)   =	\begin{bmatrix}
			M_{d,\ve}(z) & M_{od,\ve}(z)\\ 
			M_{od,\ve}(z) &M_{d,\ve}(z)
			\end{bmatrix}.
\]
We point out the following identities, $z\in\CO\backslash[0,+\infty)$:  
\begin{equation}\label{Md_eq1}
\langle\xi, M_{d,\ve}(z) \xi\rangle =  \frac{1}{2\pi}  \int_{\RE^2} \frac{  |\hat \xi (k+p/2)|^2}{\ve^2 k^2 +p^2  -z}   \dk  =\langle M_{d,\ve}(\bar z)\xi,\xi\rangle \,,
\end{equation}
\begin{equation}\label{Mod_eq1}
\langle\eta, M_{od,\ve}(z) \xi\rangle =   \frac{1}{2\pi}  \int_{\RE^2} \frac{  \overline{\hat \eta (k+p/2)} \hat \xi (-k+p/2) }{\ve^2 k^2 +p^2  -z }   \dk 
=\langle M_{od,\ve}(\bar z)\eta, \xi\rangle . 
\end{equation}
Since for all $\Xi,\widetilde \Xi \in L^2(\RE)\oplus L^2(\RE)$ there holds 
\begin{equation}\label{3.13a}
\langle\widetilde \Xi, \Mb_{\ve}(z)\Xi\rangle_{L^2(\RE^2)\oplus L^2(\RE)} = \langle\tilde \xi_1, M_{d,\ve}(z) \xi_1\rangle+\langle\tilde \xi_2, M_{d,\ve}(z) \xi_2\rangle + 
\langle\tilde \xi_2,  M_{od,\ve}(z)\xi_1\rangle+ \langle\tilde \xi_1,  M_{od,\ve}(z)\xi_2\rangle, 
\end{equation}
we deduce that 
\begin{equation}\label{Mb_prop1}
\Mb_{\ve}(z) = \Mb_{\ve}(\overline z)^* \qquad z\in\CO\backslash [0,+\infty)
\end{equation}
(we remark that this property follows by construction since $\Mb_{\ve}(z)  =  \Tau (\Tau R^{0}_{\ve}(\overline z))^* =   \Tau  R^{0}_{\ve} (z) \Tau^*$).

Note that the bounded operator 
\[
\frac1{\alpha }+\Mb_{ \ve}(z) :  L^2(\RE)\oplus L^2(\RE) \to  L^2(\RE)\oplus L^2(\RE) , \qquad 
\qquad \alpha \in  \RE\backslash\{0\} ,\; z\in \CO\backslash[0,+\infty)
\]
enjoys the same properties \eqref{Mb_prop2} and \eqref{Mb_prop1} as $\Mb_{\ve}(z)$.

In the forthcoming Proposition \ref{p:invertibility} we will prove that for $\lambda$ large enough $\frac1{\alpha }+\Mb_{ \ve}(-\lambda)$ is invertible with bounded inverse. In the proof we will use a well known result  recalled in the following remark.
\begin{remark}\label{r:invertibility} Let $\Hil$ be a  Hilbert space and $T:D(T)\subset 
\Hil \to \Hil$ be a self-adjoint operator. If there exists a positive constant $c$ such that 
\begin{equation}\label{coercivity0}
\left|\langle u, T u \rangle_{\Hil }\right| \geq c \|u\|^2_{\Hil }\qquad \forall u \in D(T), 
\end{equation}
then, $T$ is injective, surjective, and has (self-adjoint) inverse bounded by $1/c$. To see that this indeed the case, note that, by the Cauchy-Schwarz inequality, Eq. \eqref{coercivity0} implies 
\begin{equation}\label{coercivity2}
\|Tu\|_{\Hil } \geq c \|u\|_{\Hil} \qquad \forall u \in D(T) .
\end{equation}
So $T$  is injective, which in turn implies $\overline{\ran(T)} =\Hil$. Hence, for any $v\in\Hil $, there exists a sequence $\{u_n\} \in D(T) $ such that $\|T u_n  - v\|_{ \Hil } \to 0$. By the lower bound \eqref{coercivity2}, there follows that $\{u_n\}$ is a Cauchy sequence in $\Hil $ and so  it converges to $u_v\in \Hil $. Since $T$ is closed,  $u_v \in D(T)$ and $Tu_v = v$. Hence, $T$ is also surjective, and has inverse bounded by $1/c$.
\end{remark}

\begin{proposition}\label{p:invertibility}
Let $\alpha   \in \RE\backslash\{0\}$ and $\ve>0$, then there exists $\lambda_{\alpha,\ve}>0$  such that for all $\lambda > \lambda_{\alpha,\ve}$ the  operator $\frac1{\alpha }+\Mb_{ \ve}(-\lambda)$ is invertible in $L^2(\RE)$ with bounded inverse.  
\end{proposition}
\begin{proof} 
By Remark \ref{r:invertibility}, it is enough to prove that there exists $\lambda_{\alpha,\ve}$ such that  for all $\lambda>\lambda_{\alpha,\ve}$ there holds 
\begin{equation*}
\left|\left\langle\Xi, \left(\frac1{\alpha }+\Mb_{ \ve}(-\lambda)\right)\Xi\right\rangle\right| \geq c_\lambda \|\Xi\|^2 \qquad \forall  \Xi\in L^2(\RE^2)\oplus L^2(\RE),
\end{equation*}
for some positive constant $c_\lambda$.
 
We point out the following identities (see Eqs. \eqref{Md_eq1} and \eqref{Mod_eq1}): 
\begin{equation*}
\langle\xi, M_{d,\ve}(-\lambda) \xi\rangle = \int_\RE |\hat \xi(\nu)|^2 \int_\RE \frac1\pi \frac{1}{\ve^2 k^2+4(\nu-k)^2+\lambda}\d k \,\d\nu  =\int_{\RE} \frac{|\hat \xi (\nu)|^2}{\sqrt{4\ve^2 \nu^2+(4+\ve^2)\lambda}}\d\nu;
\end{equation*}
\begin{equation*}
\langle\eta, M_{od,\ve}(-\lambda) \xi\rangle = 
\frac{2}{\pi} \int_{\RE}\int_{\RE} \frac{ \overline{\hat \eta (\nu)} \hat \xi (\nu')  }{(4+\ve^2)( \nu^2+{\nu'}^2)+2(4-\ve^2)\nu\nu' + 4\lambda}\d\nu\, \d\nu' .
\end{equation*}
Hence, in Fourier transform, $M_{d,\ve}(-\lambda)$ is the multiplication operator for 
\begin{equation}\label{Md-Fou}
\hat M_{d,\ve}(\nu;-\lambda)  =  \frac{1}{\sqrt{4\ve^2 \nu^2+(4+\ve^2)\lambda}}
\end{equation}
and $ M_{od,\ve}(-\lambda)$ is the operator with integral kernel 
\begin{equation*}
 \hat M_{od,\ve}(\nu, \nu';-\lambda) = 
\frac{2}{\pi} \frac{ 1}{(4+\ve^2)( \nu^2+{\nu'}^2)+2(4-\ve^2)\nu\nu' + 4\lambda} .
\end{equation*}
Both $M_{d,\ve}(-\lambda)$ and $ M_{od,\ve}(-\lambda)$ are bounded. Indeed,
\[
\|M_{d,\ve}(-\lambda)\| \leq  \frac{1}{\sqrt{(4+\ve^2)\lambda}};
\]
and  $M_{od,\ve}(-\lambda)$ is a Hilbert-Schmidt operator, hence it is also compact and bounded. To see that this is indeed the case notice that $(4+\ve^2)( \nu^2+{\nu'}^2)+2(4-\ve^2)\nu\nu'  \geq 2 \min(\ve^2,4)( \nu^2+{\nu'}^2) $, hence    
    \[
    \int_{\RE^2} |\hat M_{od,\ve}(\nu,\nu';-\lambda)|^2 \, \d\nu \d\nu' \leq \frac{4}{\pi^2} \left( \int_\RE \frac{1}{2\min(\ve^2,4) \nu^2 +4\lambda} \d\nu \right)^2 = \frac{1}{2\min(\ve^2,4) \lambda}.
    \]
So that 
\[
\|M_{od,\ve}(-\lambda)\| \leq \|M_{od,\ve}(-\lambda)\|_{HS}
\leq \sqrt{\frac{1}{2\min(\ve^2,4) \lambda}}.
\]
By Eq. \eqref{3.13a} there follows that   there exists $C_\ve>0$ such that    $\big|(\Xi, \Mb_{\ve}(-\lambda)\Xi)_{L^2(\RE^2)\oplus L^2(\RE)} \big| \leq C_\ve \|\Xi\|^2/\sqrt\lambda$. The latter bound gives 
\[
\left|\left\langle\Xi, \left(\frac1\alpha   + \Mb_{\ve}(-\lambda)\right)\Xi\right\rangle \right|\geq \left(\frac1{|\alpha |}  - \frac{C_\ve}{\sqrt{\lambda}}\right) \|\Xi\|^2 \geq c_\lambda \|\Xi\|^2 
\]
for $\lambda>\lambda_{\alpha,\ve}  =(C_\ve\alpha )^2 $, which concludes the proof. 
\end{proof}
Proposition \ref{p:invertibility} leads to the following (for notational simplicity, here and below we avoid to explicitly indicate the $\alpha$-dependence; notice that the case $\alpha =0$ gives the free Hamiltonian)

\begin{theorem}  \label{th:resolvent1}Let $\alpha  \in\RE$, $\ve>0$ and $z\in \CO\backslash[0,+\infty)$. Then the linear operator in $L^{2}(\RE^{2})$ defined by 
\[
D(\widetilde H_{\alve}) := \big\{ \phi \in H^{3/2-}(\RE^2): \phi+\alpha \, \Gb_{\ve}(z)\Tau\phi \in H^{2}(\RE^2) \big\}, 
\]
  \begin{equation}\label{act}
 (\widetilde H_{\alve} -z) \phi   := (H^{0}_{\ve} -z ) (\phi+\alpha \, \Gb_{\ve}(z)\Tau\phi )
  \end{equation}
is a $z$-independent, bounded-from-below, self-adjoint extension of the symmetric operator $H^{0}_{\ve}|\ker(\Tau)$. \par 
Furthermore, the operator $1+\alpha \,  \Mb_{\ve}(z): L^2(\RE) \oplus L^2(\RE)\to L^2(\RE)\oplus L^2(\RE)$ has a bounded inverse for all $z\in   
\rho(\widetilde H_{\alve})\cap \CO\backslash[0,+\infty)$ and the resolvent of   $\widetilde H_{\alve}$ is given by 
\[
\widetilde R_{\alve}(z) = R^{0}_{\ve}(z) - \alpha \,\Gb_{\ve}(z)\big(1+\alpha \,\Mb_{\ve}(z)\big)^{-1} \breve \Gb_{\ve}(z) \qquad  z\in \rho(\widetilde H_{\alve})\cap \CO\backslash[0,+\infty)\,.
\] 
\end{theorem}
\begin{proof} The case $\alpha =0$ is trivial. Let $\alpha \not=0$. By \cite[Theorem 2.1]{Pjfa01},  the bounded linear operator 
\[
\widetilde R_{\alve}(-\lambda) :L^2(\RE^2) \to L^2(\RE^2)  \qquad \lambda > \lambda_{\alpha,\ve}
\]
\begin{equation}\label{krein}
\widetilde R_{\alve}(-\lambda) := R^{0}_{\ve}(-\lambda) - \Gb_{\ve}(-\lambda)\left(\frac1\alpha   + \Mb_{ \ve}(-\lambda)\right)^{-1} \breve \Gb_{\ve}(-\lambda) 
\end{equation} 
is the resolvent of the self-adjoint extension $\widetilde H_{\alve}\ge -\lambda_{\alpha,\ve}$ of $H^{0}_{\ve}|\ker(\Tau)$  defined by 
\[
\widetilde D := \left\{ \phi \in L^{2}(\RE^2):\; \phi = \phi_{-\lambda} - \Gb_{\ve}(-\lambda)\left( \frac1\alpha   + \Mb_{\ve}(-\lambda)\right)^{-1}  \Tau  \phi_{-\lambda  }, \;  \phi_{-\lambda} \in D( H^{0}_{\ve})\right\};
\]
\begin{equation}\label{widetildeHthetave}
\widetilde H_{\alve}:\widetilde D\subseteq L^{2}(\RE^2)\to L^{2}(\RE^2)\,,\qquad (\widetilde H_{\alve} + \lambda   )\phi  := (H^{0}_{\ve} +\lambda)\phi_{-\lambda}. 
\end{equation}

The definition of $\widetilde H_{\alve}$ is independent of $\lambda$; for any fixed $\lambda$, the decomposition of $\phi$ in $D(\widetilde H_{\alve}) $ is unique. For the reader's convenience, we sketch the proof of \cite[Theorem 2.1]{Pjfa01}, referring to \cite{Pjfa01} for the details. By the properties \eqref{resid1}, \eqref{resid2}, \eqref{resid3}, and \eqref{Mb_prop2} (see \cite[Pag. 115]{Pjfa01}) there follows that $\widetilde R_{\alve}(-\lambda)$ satisfies the relation 
\begin{equation}\label{resid}
\widetilde R_{\alve}(-\lambda) - \widetilde R_{\alve}(-\mu) = (\mu - \lambda ) \widetilde R_{\alve}(-\lambda) \widetilde R_{\alve}(-\mu) \qquad \lambda,\mu > \lambda_{\alpha,\ve}. 
\end{equation}
Hence, it is a pseudo-resolvent. $\widetilde R_{\alve}(-\lambda)$ is the resolvent of a (closed) operator if and only if it is injective (see \cite[Theorem 4.10]{Stone}).  $\widetilde R_{\alve}(-\lambda)$ is indeed injective, since $\widetilde R_{\alve}(-\lambda) \psi =0$ would imply 
\[
R^{0}_{\ve}(-\lambda)\psi =  \Gb_{\ve}(-\lambda)Q 
\]
with $Q = \big(\frac1\alpha   + \Mb_{\ve}(-\lambda)\big)^{-1} \breve \Gb_{\ve}(-\lambda)\psi$, but this implies $\psi =0$ by Remark \ref{r:GcapDH}. By the self-adjointness of the operator  $\widetilde R_{\alve}(-\lambda)$ there follows that $\widetilde D := \ran\big(\widetilde R_{\alve}(-\lambda)\big)$ is dense,  and independent of $\lambda$ because of the resolvent identity \eqref{resid}; moreover, the closed  operator $\widetilde H_{\alve} = \widetilde R_{\alve}(-\lambda)^{-1}-\lambda$ is self-adjoint, because $\widetilde R_{\alve}(-\lambda)^* = \widetilde R_{\alve}(-\lambda)$. By Remark \ref{r:GcapDH} the decomposition of $\phi$ in $\widetilde D $ is unique. If $\phi \in D(H^{0}_{\ve}) \cap \ker (\Tau)$ one has $\phi = \phi_{-\lambda}$ and $\widetilde H_{\alve} \phi =  H_{0 , \ve}\phi$ by Eq. \eqref{widetildeHthetave}. 
\par
Then, by \cite[Theorem 2.19]{CFPrma18}, the resolvent formula \eqref{krein} extends to all $ z\in  \rho(\widetilde H_{\alve})\cap\CO\backslash[0,+\infty)$ and 
\[
\widetilde D = \big\{ \phi \in L^2(\RE^2):\; \phi = \phi_{z} - \Gb_{\ve}(z)\left( \frac1\alpha   + \Mb_{\ve}(z)\right)^{-1}  \Tau  \phi_{z}, \;  \phi_{z} \in D(H^{0}_{\ve})\big\}\,,
\]
\[
(\widetilde H_{\alve} -z  )\phi  = (H^{0}_{\ve} -z )\phi_{z}. 
\]
The definition of $\widetilde H_{\alve}$ is independent of $z$  and, for any fixed $z$, the decomposition of $\phi$ in $\widetilde D $ is unique. Notice that  $\widetilde D\subseteq H^{3/2-}(\RE^2)$ by the mapping properties of $\Gb_{\ve}(z)$ (see \eqref{ran}).\par
 Let us now show that $\widetilde D=D(\widetilde H_{\alve})$. Since, by 
\eqref{resid3}, $\ran(\Gb_{\ve}(z)-\Gb_{\ve}(w))\subseteq D(H^{0}_{\ve})=H^{2}(\RE^{2})$, it is enough to prove $\widetilde D=D(\widetilde H_{\alve})$ whenever the $z$ appearing in the definition of $D(\widetilde H_{\alve})$ belongs to $ \rho(\widetilde H_{\alve})\cap\CO\backslash[0,+\infty)$. Given $\phi\in \widetilde D$, one has $\Tau\phi=\Tau\phi_{z}-\alpha \,\Mb_{\ve}(z)\left( 1+\alpha  \,\Mb_{\ve}(z)\right)^{-1}  \Tau  \phi_{z}$, which entails $ \Tau\phi=\left( 1+\alpha \,\Mb_{\ve}(z)\right)^{-1}  \Tau  \phi_{z}$. Hence, $\phi_{z}=\phi+\alpha \, \Gb_{\ve}(z)\Tau\phi$; this gives $\widetilde D\subseteq D(\widetilde H_{\alve})$. 
Viceversa, given $\phi\in D(\widetilde H_{\alve})$ and defining $\phi_{z}\in H^{2}(\RE^{2})$ by $\phi_{z}:=\phi+\alpha  \,\Gb_{\ve}(z)\Tau\phi$, one has $\Tau\phi+\alpha \,\Mb_{\ve}(z)\Tau\phi=\Tau\phi_{z}$, which gives $  \Tau\phi=\left( 1+\alpha \, \Mb_{\ve}(z)\right)^{-1}  \Tau  \phi_{z}$. Hence, $\phi=\phi_{z}-\alpha \Gb_{\ve}(z)\left( 1+\alpha \, \Mb_{\ve}(z)\right)^{-1}  \Tau  \phi_{z}$; this entails $D(\widetilde H_{\alve})\subseteq \widetilde D$.
Therefore, $\widetilde D=D(\widetilde H_{\alve})$. Furthermore, the previous calculations also give
\be\label{eq}
(\widetilde H_{\alve} -z  )\phi=(H^{0}_{\ve} -z ) (\phi+\alpha \, \Gb_{\ve}(z)\Tau\phi )\,,\qquad z\in  \rho(\widetilde H_{\alve})\cap\CO\backslash[0,+\infty)\,.
\ee  
To conclude, let us show that \eqref{eq} holds true even whenever $z\in \CO\backslash[0,+\infty)$. Given any $w\in  \rho(\widetilde H_{\alve})\cap\CO\backslash[0,+\infty)$, one has, by \eqref{eq} and \eqref{resid3},
\begin{align*}
&(\widetilde H_{\alve} -z  )\phi=
(\widetilde H_{\alve} -w  )\phi+(w-z)\phi=
(H^{0}_{\ve} -w ) (\phi+\alpha \, \Gb_{\ve}(w)\Tau\phi )+(w-z)\phi\\
=&(H^{0}_{\ve} -w ) (\phi+\alpha \, \Gb_{\ve}(z)\Tau\phi )+
\alpha  (H^{0}_{\ve} -w ) (\Gb_{\ve}(w)-\Gb_{\ve}(z))\Tau\phi +(w-z)\phi\\
=&(H^{0}_{\ve} -w ) (\phi+\alpha \, \Gb_{\ve}(z)\Tau\phi )+(w-z)
\alpha \, \Gb_{\ve}(z)\Tau\phi +(w-z)\phi\\
=&(H^{0}_{\ve} -z ) (\phi+\alpha \, \Gb_{\ve}(z)\Tau\phi )\,.
\end{align*}
\end{proof}

\begin{remark}\label{temp} Introducing the adjoint 
$\Tau^{*}:H^{-\nu+1/2}(\RE)\oplus H^{-\nu+1/2}(\RE)\to H^{-\nu}(\RE^{2})$, $\nu>1/2$, which provides, whenever $\xi_{1}\oplus\xi_{2}\in L^{2}(\RE)\oplus L^{2}(\RE)$ the tempered distribution $\Tau^{*}\xi_{1}\oplus\xi_{2}\in \SS'(\RE^{2})$ acting on a test function $f\in\SS(\RE^{2})$ as
\[
(\Tau^{*}\xi_{1}\oplus\xi_{2})f=\int_{\RE}\xi_{1}(s)f(s,s/2)\,\d s+\int_{\RE}\xi_{2}(s)f(-s,s/2)\,\d s\,,
\]
one has $\Gb_{\ve}(z)=R^{0}_{ \ve}(z)\Tau^{*}$. The latter entails the distributional identity 
\be\label{d-id-1}
\big( -\ve^2 \partial^2_{x}-\partial^2_{y}-z\big)\Gb_{\ve}(z)\xi_{1}\oplus\xi_{2}=\Tau^{*}\xi_{1}\oplus\xi_{2}
\ee
and so, by \eqref{act}, for any $\phi \in D(\widetilde H_{\alve})$ one gets
\be\label{act2}
\widetilde H_{\alve}\phi=\big( -\ve^2 \partial^2_{x}-\partial^2_{y}\big)\phi+\alpha  \Tau^{*}\Tau\phi\,.
\ee
Notice that, unless $\phi\in\ker(\Tau)$, neither of the two tempered distributions on the right hand side of the latter equation is in $L^{2}(\RE^{2})$ but their sum is. Moreover, since $\supp(\Tau^{*}\xi_{1}\oplus\xi_{2})=\Pi$, 
 one has
 \[
 \widetilde H_{\alve}\phi(\x)=\big( -\ve^2 \partial^2_{x}-\partial^2_{y}\big)\phi(\x)\quad\text{for a.e. $\x$ in $\RE^{2}\backslash\Pi$.}
\]
\end{remark}

\begin{proposition} \label{sesqui} Let 
 \begin{equation*}
\widetilde\BB_{\alve}: H^1(\RE^2)\times H^1(\RE^2)\subseteq L^{2}(\RE^{2})\times L^{2}(\RE^{2})\to\CO
\end{equation*}
be the sesquilinear form
\begin{equation*}
\widetilde\BB_{\alve}(\varphi, \phi):= \ve^{2}\langle\partial_x \varphi,\partial_x \phi\rangle 
+ \langle\partial_y \varphi,\partial_y \phi\rangle+\alpha  \big(\langle\tau_{1}\varphi  ,  \tau_{1}\phi\rangle+\langle\tau_{2}\varphi  ,  \tau_{2}\phi\rangle\big)\,.
\end{equation*}
Then
\[
\widetilde\BB_{\alve}(\varphi , \phi) =  \langle\varphi  , \widetilde H_{\alve}\phi\rangle \qquad \forall \, \varphi \in H^{1}(\RE^{2}) ,\; \forall \phi\in D(\widetilde H_{\alve}).
\]
\end{proposition}
\begin{proof} For the sesquilinear form $\BB^{0}_{\ve}$ corresponding to $H^{0}_\ve$, there holds
\[
\BB^{0}_{\ve}(\varphi , \phi) =  \langle\varphi  , H^{0}_\ve\phi\rangle  \qquad \forall \, \varphi \in H^{1}(\RE^{2}) ,\; \forall \phi\in H^{2}(\RE^{2})
\]
and
\[
\BB^{0}_{\ve}(\varphi , \phi) 
=\langle(-\ve^{2}\partial^{2}_{x}-\partial^{2}_{y})\varphi  , \phi\rangle_{-1,+1}
=  \langle\varphi , (-\ve^{2}\partial^{2}_{x}-\partial^{2}_{y})\phi\rangle_{+1,-1} \qquad \forall \, \psi ,\, \phi\in H^{1}(\RE^{2})\,,
\]
where $\langle\cdot,\cdot\rangle_{\mp\nu,\pm\nu}$ denotes the (anti-linear with respect to the first variable) $H^{\pm\nu}(\RE^{2})$-$H^{\mp\nu}(\RE^{2})$ duality extending the  scalar product in $L^{2}(\RE^{2})$.\par
Then, defining $\phi_{z}:=\phi+\alpha   \Gb_{\ve}(z)\Tau\phi$,  by \eqref{d-id-1} and since both $\phi$ and  $\Gb_{\ve}(z)\Tau\phi$ belong to $H^{3/2-}(\RE^{2})\subset H^{1}(\RE^{2}) $, one gets
\begin{align*}
\langle\varphi  , \widetilde H_{\alve}\phi\rangle=&
\langle\varphi  , (\widetilde H_{\alve}-z)\phi\rangle  +z\langle\varphi  , \phi\rangle  =(\varphi  , (H^{0}_{\ve}-z)\phi_{z}\rangle  +z\langle\varphi  , \phi\rangle  \\=&(\BB^{0}_{\ve}-z)(\varphi  , \phi_{z})+z\langle\varphi  , \phi\rangle  
=\BB^{0}_{\ve}(\varphi  , \phi)+\alpha  (\BB^{0}_{\ve}-z)(\varphi  , \Gb_{\ve}(z)\Tau\phi)\\
=&\BB^{0}_{\ve}(\varphi  , \phi)+\alpha  \langle\varphi  ,  (-\ve^{2}\partial^{2}_{x}-\partial^{2}_{y}-z) \Gb_{\ve}(z)\Tau\phi\rangle_{+1,-1}\\
=&\BB^{0}_{\ve}(\varphi  , \phi)+\alpha  \langle\varphi  ,  \Tau^{*}\Tau\phi\rangle_{+1,-1}=
\BB^{0}_{\ve}(\varphi  , \phi)+\alpha  \big(\langle\tau_{1}\varphi  ,  \tau_{1}\phi\rangle +\langle\tau_{2}\varphi  ,  \tau_{2}\phi\rangle \big)\\
=&\widetilde\BB_{\alve}(\varphi  , \phi)\,.
\end{align*}
\end{proof}
\begin{remark}\label{r:bc} Since $G_{j,\ve}(z)$ corresponds to the single-layer operator for the elliptic operator $\ve^2 \partial^2_{x}+\partial^2_{y}+z$ relative to the line $\Pi_{j}$, by the jump and regularity properties of the single-layer potentials one gets
an alternative characterization of the self-adjointness domain of $\widetilde H_{\alve}$ in terms of boundary conditions: 
\[
D(\widetilde H_{\alve}) = 
\big\{\phi\in H^{3/2-}(\RE^2) :  \phi=\phi_{1}+\phi_{2}\,,\ \phi_{j}\in H^{2}(\RE^{2}\backslash\Pi_{j})\,,\  [\tau'_{j,\ve}]\phi_{j}=\alpha\tau_{j}\phi\,,\ j=1,2\big\}, 
\]
where $[\tau'_{j,\ve}]\phi_{j}$
denotes the jump across $\Pi_{j}$ of the normal derivative relative to $ \ve^2 \partial^2_{x}+\partial^2_{y}$. The  boundary condition on $\Pi_1$ is written  explicitly in Eq. \eqref{explicitBC}.
\end{remark}

\subsection{The compression of \texorpdfstring{$\widetilde H_{\ve,\alpha }$\label{ss:2.2}}{Hve} onto 
$L^{2}_{\bosfer}(\RE^{2})$} \label{ss:resolvent} In this section we introduce the bosonic and the fermionic symmetries. Recall that the symbol $\bosfer$ denotes either $\bos$ (bosonic) or $\fer$ (fermionic), with the corresponding value $\bosferpm_\bosfer$ defined as $\bosferpm_\bos = +$ and $\bosferpm_\fer = -$.\par
At first, notice that 
\be\label{inv}
R^{0}_{\ve}(z) : L^2_{\bosfer}(\RE^2) \to H^2_{\bosfer}(\RE^2) \qquad \forall z\in\CO\backslash [0,+\infty)\,.
\ee
We define the operator  
\begin{equation}\label{SS1}
S: H^{\nu}(\RE^2) \to  H^{\nu}(\RE^2) \,,\quad \nu\ge 0\,,\qquad
(S\psi)(x,y) := \psi(-x,y),
\end{equation}
and the orthogonal projector 
\begin{equation}\label{SS2}
\frac{1\bosferpm_\bosfer S}2: L^{2}(\RE^{2})\to L^{2}(\RE^{2})\,,\qquad \ran\left(\frac{1\bosferpm_\bosfer S}2\right)=L^{2}_{\bosfer}(\RE^{2})\,.
\end{equation}
$L^{2}_{\bosfer}(\RE^{2})$ enjoys the property  
\[
S\psi = \bosferpm_\bosfer\psi \qquad \forall \psi \in L^2_{\bosfer}(\RE^2). 
\]
One has 
\[
\tau_{1}S=\tau_{2}\,,\qquad \tau_{2}S=\tau_{1}\,,
\]
\[ 
 \breve G_{2,\ve}(z) =  \breve G_{1,\ve}(z) S \,,\qquad  G_{2,\ve}(z) = SG_{1,\ve}(z)\,.
\]
Noticing that
\[
\tau_{1} \psi=\bosferpm_\bosfer\tau_{2}\psi  \qquad \forall  \psi\in  H^2_{\bosfer}(\RE^2)\,,
\]
from now on we pose 
\[\tau\equiv\tau_{1}\]
and introduce the bounded operators
\[
\breve G_{\ve}(z):L^{2}_{\bosfer}(\RE^{2})\to L^{2}(\RE)\,,\qquad
\breve G_{\ve}(z):=\tau R_{\ve}^{0}(z)\equiv\breve G_{1,\ve}(z)|L^{2}_{\bosfer}(\RE^{2})\,,
\]
\[
G_{\ve}(z):L^{2}(\RE)\to L^{2}_{\bosfer}(\RE^{2})\,,\qquad G_{\ve}(z):=\breve G_{\ve}(\bar z)^{*}\equiv
\frac{1 \bosferpm_\bosfer S}2\,G_{1,\ve}(z)
\,,
\]
\[
M_{\ve}(z):L^{2}(\RE)\to L^{2}(\RE)\,,\qquad M_{\ve}(z):=\tau G_{\ve}(z)\,.
\]
\begin{remark} To avoid a too heavy indexing,  we avoided to use the symbol $\bosfer$ to distinguish the bosonic and fermionic versions of the operators $\breve G_{\ve}(z)$, $G_{\ve}(z)$ and $M_{\ve}(z)$. We trust that the reader will not be confused by this abuse of notation.
\end{remark} 
With such definitions one has
\begin{align*}
 M_{\ve}(z)=&\tau_{1}\,\frac{1\bosferpm_\bosfer S}2\,G_{1,\ve}(z)=
 \frac12\,\big(\tau_{1}G_{1,\ve}(z) \bosferpm_\bosfer \tau_{1}S G_{1,\ve}(z)\big)=
 \frac12\,\big(\tau_{1} G_{1,\ve}(z)\bosferpm_\bosfer \tau_{1} G_{2,\ve}(z)\big)\\
 =&
\frac12\,\big(M_{d,\ve}(z)\bosferpm_\bosfer M_{od,\ve}(z)\big)\,.
\end{align*}
Then, by such relations, one gets
\[
\breve \Gb_{\ve}(z)\psi=\breve G_{\ve}(z)\psi\oplus(\bosferpm_\bosfer\breve G_{\ve}(z)\psi)\,,\qquad \forall\psi\in L^{2}_{\bosfer}(\RE^{2})\,,
\]
\[
\Gb_{\ve}(z)\xi\oplus(\bosferpm_\bosfer\xi)=G_{1,\ve}(z)\xi\bosferpm_\bosfer G_{2,\ve}(z)\xi=(1\bosferpm_\bosfer S)G_{1,\ve}(z)\xi=2 \,G_{\ve}(z)\xi\,,
\]
\[
\Mb_{\ve}(z)\xi\oplus(\bosferpm_\bosfer \xi)=2\big(M_{\ve}(z)\xi\oplus (\bosferpm_\bosfer M_{\ve}(z)\xi)\big)
\,,
\]
\[
(1+\alpha \,\Mb_{\ve}(z))^{-1}\xi\oplus(\bosferpm_\bosfer \xi)=(1+2\alpha \,M_{\ve}(z))^{-1}\xi\oplus 
\big(\bosferpm_\bosfer(1+2\alpha \,M_{\ve}(z))^{-1}\xi\big)\,.
\]
Therefore, for all $\psi\in L^{2}_{\bosfer}(\RE^{2})$ and for all  $z\in   
\rho(\widetilde H_{\alve})\cap \CO\backslash[0,+\infty)$ one obtains
\begin{align}\label{t-res}
\widetilde R_{\alve}(z) \psi =& R^{0}_{\ve}(z)\psi   -\alpha \, \Gb_{\ve}(z)(1+\alpha \,\Mb_{\ve}(z))^{-1}\breve\Gb_{\ve}(z)\psi\nonumber\\
=&R^{0}_{\ve}(z)\psi   -2\alpha \,G_{\ve}(z)(1+2\alpha \,M_{\ve}(z))^{-1}\breve G_{\ve}(z)\psi\,.
\end{align}
Hence, by \eqref{inv},  $R^{\bosfer}_{\alve}(z):=\widetilde R_{\alve}(z)|L_{\bosfer}^{2}(\RE^2)$
is a pseudo resolvent in $L_{\bosfer}^{2}(\RE^2)$; furthermore, by the analogous properties for $\widetilde R_{\alve}(z)$, one has that  $R^{\bosfer}_{\alve}(z)$ is injective and $(R^{\bosfer}_{\alve}(z))^{*}=R^{\bosfer}_{\alve}(\bar z)$. Therefore, $R^{\bosfer}_{\alve}(z)$ is the resolvent of a self-adjoint operator $ \widehat H^{\bosfer}_{\alve}$ in $L^{2}_{\bosfer}(\RE^{2})$ which, by construction, is given by the compression of $\widetilde H_{\alve}$ onto $L^{2}_{\bosfer}(\RE^{2})$, i.e., 
\[
\widehat H^{\bosfer}_{\alve}=\frac{1\bosferpm_\bosfer S}2\,\widetilde H_{\alve}\frac{1\bosferpm_\bosfer S}2\equiv
\widetilde H_{\alve}|D(\widetilde H_{\alve})\cap L^{2}_{\bosfer}(\RE^{2})\,.
\]  
By \cite[Theorem 2.19]{CFPrma18}, the resolvent formula \eqref{t-res} extends to all $ z\in \CO\backslash[0,+\infty) \cap \rho(\widehat H^{\bosfer}_{\alve})$. \par
\begin{theorem}\label{th:resolvent-main}
Let $\alpha \in\RE$, $\ve>0$ and $z\in \CO\backslash[0,+\infty)$. Then the linear operator in $L^{2}_{\bosfer}(\RE^{2})$ defined by 
\[
D(H^\bosfer_\ve) := \big\{ \psi \in H_{\bosfer}^{3/2-}(\RE^2): \psi+2\alpha G_{\ve}(z)\tau\psi \in H_{\bosfer}^{2}(\RE^2) \big\}, 
\]
  \begin{equation*}
 (H^\bosfer_\ve -z) \psi   := (H^{0}_{\ve} -z ) (\psi+2\alpha G_{\ve}(z)\tau\psi )
  \end{equation*}
is a $z$-independent, bounded-from-below, self-adjoint extension of the symmetric operator $H^{0}_{\ve}|\ker(\tau|H^{2}_{\bosfer}(\RE^{2}))$. \par 
Furthermore, the operator $1  + 2\alpha M_{\ve}(z): L^2(\RE)\to L^2(\RE)$ has a bounded inverse for all $z\in   
\rho(H^\bosfer_\ve)\cap \CO\backslash[0,+\infty)$ and the resolvent of   $H^\bosfer_\ve$ is given by 
\be\label{krein-b}
R^{\bosfer}_{\ve}(z) = R^{0}_{\ve}(z) - 2\alpha G_{\ve}(z)\big(1  +2\alpha M_{\ve}(z)\big)^{-1} \breve G_{\ve}(z) \qquad  z\in \rho(H^\bosfer_\ve)\cap \CO\backslash[0,+\infty)\,.
\ee
\end{theorem}
By Remark \ref{temp}, one gets
\begin{remark}\label{r:temp-b} Introducing the adjoint 
$\tau^{*}:H^{-\nu+1/2}(\RE)\to H^{-\nu}(\RE^{2})$, $\nu>1/2$, which provides, whenever $\xi\in L^{2}(\RE)$, the tempered distribution $\tau^{*}\xi\in \SS'(\RE^{2})$ acting on a test function $f\in\SS(\RE^{2})$ as
\[
(\tau^{*}\xi)f=\int_{\RE}\xi(s)\,f(s,s/2)\,\d s\,,
\]
one has $G_{\ve}(z)=R^{0}_{ \ve}(z)\tau^{*}$. The latter entails the distributional identity 
\begin{equation*}
\big(-\ve^2 \partial^2_{x}-\partial^2_{y}-z\big)G_{\ve}(z)\xi=\tau^{*}\xi
\end{equation*}
and so, by \eqref{act2} and by $\Tau^{*}\xi\oplus\xi=2\tau^{*}\xi$, one gets
\[
 H^\bosfer_\ve\psi=\big( -\ve^2 \partial^2_{x}-\partial^2_{y}\big)\psi+2\alpha\,\tau^{*}\tau\psi\,.
\]
\end{remark}
By Proposition \ref{sesqui} there follows
\begin{proposition}\label{p:quadratic} Let
\begin{equation}\label{BBve1-b}
\BB^\bosfer_{\ve}: H^1_{\bosfer}(\RE^2)\times H^1_{\bosfer}(\RE^2)\subseteq L^{2}_{\bosfer}(\RE^{2})\times L^{2}_{\bosfer}(\RE^{2})\to\CO\
\end{equation}
be the sesquilinear form
\begin{equation}\label{BBve2-b}
\BB^\bosfer_{\ve}(\varphi, \psi):= \ve^{2}\langle\partial_x \varphi,\partial_x \psi\rangle  
+ \langle\partial_y \varphi,\partial_y \psi\rangle +2\alpha\langle\tau\varphi  ,  \tau\psi\rangle \,.
\end{equation}
Then
\[
\BB^\bosfer_{\ve}( \varphi , \psi) =  \langle \varphi , H^\bosfer_\ve\psi\rangle   \qquad \forall \,  \varphi \in H_{\bosfer}^{1}(\RE^{2}) ,\; \forall \psi\in D(H^\bosfer_\ve).
\]
\end{proposition}

\subsection{The Spectrum of \texorpdfstring{$H^\bosfer_\ve$}{Hve}\label{ss:spectrum}}
The next proposition allows to build functions in $D(H^\bosfer_\ve)$ given $\xi \in L^2(\RE)$. It will be helpful in the characterization of the spectrum of $H^\bosfer_\ve$, see Lemma \ref{l:spectrum} below. We premise the following result, whose proof is the same as that of Lemma \ref{rem}, since 
 $\breve G_{\ve}(z)=\tau R^{0}_{\ve}(z)$, $G_{\ve}(z)=(\tau R^{0}_{\ve}(\bar z))^{*}$ and $M_{\ve}(z)=\tau G_{\ve}(z)$.
\begin{lemma} For any $z,w\in \CO\backslash[0,+\infty)$, there holds
\begin{equation*}
\breve G_{\ve} (z)  - \breve G_\ve (w) = (z-w) \breve G_{\ve} (z) R^{0}_{\ve} (w) = (z-w)\breve G_{\ve} (w)R^{0}_{\ve} (z)\,,
\end{equation*}
\begin{equation}\label{resid3-b}
 G_{\ve} (z)  -  G_\ve (w) = (z-w)  R^{0}_{\ve} (w)  G_{\ve} (z) = (z-w)  R^{0}_{\ve} (z)G_{\ve} (w)\,,
\end{equation}
\begin{equation}\label{Mb_prop2-b}
 M_{\ve} (z)  -  M_\ve (w) = (z-w)  \breve G_{\ve} (w)  G_{\ve} (z) = (z-w)  \breve G_{\ve}  (z)G_{\ve} (w).
\end{equation}  
\end{lemma}
\begin{proposition}\label{p:decomposition} (i) For any $w,\,z\in  \CO\backslash[0,+\infty)$ and $\psi\in D(H^\bosfer_\ve)$, defining 
\[
\psi_{z}:=\psi+2\alpha G_{\ve}(z)\tau\psi\,,
\]
one has 
\[
(H^\bosfer_\ve -w ) \psi = (H^{0}_{\ve} -w)\psi_{z}-(z-w) 2\alpha G_{\ve}(z)\tau\psi\,.
\]
(ii) For any $w \in \CO\backslash[0,+\infty)$, $z\in  \rho(H^\bosfer_\ve)\cap \CO\backslash[0,+\infty)$ and $\xi \in  L^2(\RE)$, defining 
\[
\phi_{w,\xi}:=R^{\bosfer}_{\ve}(z)G_{\ve}(w)\xi\,,
\]
one has
\[
(H^\bosfer_\ve - w)\phi_{w,\xi}=  G_\ve (z) \left( 1  + 2\alpha M_{\ve}(z)\right)^{-1} \left( 1  +2\alpha M_{\ve}(w)\right) \xi\,.
\]
\end{proposition}
\begin{proof} (i) By the action  of $H^\bosfer_\ve$ on its domain, one has
\begin{align*}
(H^\bosfer_\ve -w ) \psi =& (H^{0}_\ve - z) (\psi+2\alpha G_{\ve}(z)\tau\psi)
 + (z-w) \psi\\
 =& (H^{0}_\ve -z) \psi_{z}
 + (z-w) (\psi_{z}-2\alpha G_{\ve}(z)\tau\psi)\\
 =&(H^{0}_{\ve} -w)\psi_{z}-(z-w) 2\alpha G_{\ve}(z)\tau\psi\,.
\end{align*}
(ii) By \eqref{krein-b}, \eqref{Mb_prop2-b} and \eqref{resid3-b}, one gets
\begin{align*}
&(H^\bosfer_\ve - w)\phi_{w,\xi}=(H^\bosfer_\ve - z)\phi_{w,\xi}+(z-w)\phi_{w,\xi}
=G_{\ve}(w)\xi+(z-w)\phi_{w,\xi}\\
=&G_{\ve}(w)\xi+(z-w)R^{0}_{\ve}(z)G_{\ve}(w)\xi - (z-w)2\alpha G_{\ve}(z)(1  +2\alpha M_{\ve}(z))^{-1} \breve G_{\ve}(z)G_{\ve}(w)\xi\\
=&G_{\ve}(z)\xi-  G_{\ve}(z)(1  +2\alpha M_{\ve}(z))^{-1} (1+2\alpha M_{\ve}(z)-1-2\alpha M_{\ve}(w))\xi\\
=&G_\ve (z) \left( 1  + 2\alpha M_{\ve}(z)\right)^{-1} \left( 1  +2\alpha M_{\ve}(w)\right) \xi\,.
\end{align*}
\end{proof}

\begin{lemma}\label{l:spectrum} Let $\lambda > 0$. Then,
\begin{enumerate}[(i)]
    \item $-\lambda \in \sigma_{p}(H^\bosfer_\ve) $ if and only if $0  \in \sigma_{p}(1  +2\alpha M_{\ve}(-\lambda))$;
    \item $-\lambda \in \sigma_{ess}(H^\bosfer_\ve) $ if and only if $0  \in \sigma_{ess}(1  +2\alpha M_\ve(-\lambda))$.
\end{enumerate}
\end{lemma}
\begin{proof}   
At first, let us notice that $-\lambda \in \rho (H^{0}_{\ve})$. 

(i) Let $\psi\in D(H^\bosfer_\ve)$ be an eigenvector of $H^\bosfer_\ve$ with eigenvalue $-\lambda$. 
By Proposition \ref{p:decomposition}(i),
\[0 =(H^\bosfer_\ve +\lambda ) \psi = (H^{0}_{\ve} +\lambda) (\psi+2\alpha G_{\ve}(z)\tau\psi)
 - (z+\lambda) 2\alpha G_{\ve}(z)\tau\psi\,,\qquad\forall z\in\CO\backslash[0,+\infty)\,.
\]
This gives
\[
 \psi+2\alpha G_{\ve}(-\lambda)\tau\psi =2\alpha(G_{\ve}(-\lambda)-G_{\ve}(z))\tau\psi - 2\alpha(z+\lambda)R^{0}_{\ve}(-\lambda)G_{\ve}(z)\tau\psi\,. 
\]
Applying $\tau$ to both sides and taking into account \eqref{Mb_prop2-b}, one gets
\[
(1+2\alpha M_{\ve}(-\lambda))\tau\psi=2\alpha  (M_{\ve}(-\lambda)-M_{\ve}(z))\tau\psi-2\alpha(z+\lambda)\breve G_{\ve}(-\lambda)G_{\ve}(z)\tau\psi=0\,.
\]

On the other hand, assume that there exists $\xi \in  L^2(\RE)$ such that $(1  + 2\alpha M_{\ve}(-\lambda))\xi =0$. Then, by Proposition \ref{p:decomposition}(ii), $\phi_{-\lambda,\xi}=R^{\bosfer}_{\ve}(z)G_{\ve}(-\lambda)\xi$, $z\in  \rho(H^\bosfer_\ve)\cap \CO\backslash[0,+\infty)$, is an eigenfunction with eigenvalue $-\lambda$.

(ii) Let $\{\psi_n\}$ be a singular Weyl  sequence for $( H^\bosfer_\ve\,,-\lambda)$, i.e., $\psi_n\in D( H^\bosfer_\ve)$,  $\|\psi_n\|= 1$, $\psi_n \rightharpoonup 0$, and $\| (H^\bosfer_\ve +\lambda )\psi_n \| \to 0$. Arguing as above, one gets, for any $z\in\CO\backslash[0,+\infty)$,
\begin{align*}
&\breve G_{\ve}(-\lambda)(H^\bosfer_\ve +\lambda ) \psi_{n}=
\tau R^{0}_{\ve}(-\lambda)\left((H^{0}_{\ve} +\lambda) (\psi_{n}+2\alpha G_{\ve}(z)\tau\psi_{n})+(z+\lambda) 2\alpha G_{\ve}(z)\tau\psi_{n}\right)
\\=&
\tau  (\psi_n+2\alpha G_{\ve}(z)\tau\psi_{n})
 + (z+\lambda) 2\alpha \breve G_{\ve}(-\lambda)G_{\ve}(z)\tau\psi_{n}
\\
=&(1+2\alpha M_{\ve}(-\lambda))\tau\psi_{n}-2\alpha(M_{\ve}(-\lambda)-M_{\ve}(z))\tau\psi_{n}+2\alpha(z+\lambda)\breve G_{\ve}(-\lambda)G_{\ve}(z)\tau\psi_{n}
\\
=&(1+2\alpha M_{\ve}(-\lambda))\tau\psi_{n} 
\end{align*}
and so $(1+2\alpha M_{\ve}(-\lambda))\tau\psi_{n}\to 0$. Let us now show that $\tau\psi_{n}\rightharpoonup 0$.\par
Defining $\psi_{n,z}:=\psi_{n}+2\alpha G_{\ve}(z)\tau\psi_{n}\in H_{\bosfer}^{2}(\RE^{2})$,  one has, for
any $f\in H^{-2}(\RE^{2} )$ and for any $z\in\CO\backslash[0,+\infty)$,
\begin{align*}
\langle(H^\bosfer_\ve+\lambda)\psi_{n},R^{0}_{\ve}(\bar z)f\rangle 
=&\langle(H^{0}_{\ve}-z)\psi_{n,z},R^{0}_{\ve}(\bar z)f\rangle +(\lambda-z)\langle\psi_{n}, R^{0}_{\ve}(\bar z)f\rangle \\
=&\langle\psi_{n,z},f\rangle_{+2,-2}+(\lambda-z)\langle\psi_{n}, R^{0}_{\ve}(\bar z)f\rangle \,.
\end{align*}
This gives $\psi_{n,z}\rightharpoonup 0$ in $H^{2}(\RE^{2})$ and so, by $\tau\in\Bou(H^{2}(\RE^{2}),L^{2}(\RE))$ and by $\tau\psi_{n}=(1+2\alpha M_{\ve}(z))^{-1}\tau\psi_{n,z}$, one gets $\tau\psi_{n}\rightharpoonup 0$ in $L^{2}(\RE)$. To get a singular Weyl sequence for 
$(1  + 2\alpha M_{\ve}(-\lambda),0)$ we need to normalize $\tau\psi_{n}$, so we need to show that $\|\tau\psi_{n}\|$ does not converge to zero. Let us assume that this is not the case. Then, by 
Proposition \ref{p:decomposition}(i), there follows  $\|(H^{0}_\ve+\lambda) \psi_{n,z}\|\to 0$, and, by $\|\psi_{n}\|=1$, one would have $\|\psi_{n,z}\| \to  1$, but this is impossible because, by taking as singular Weyl sequence $\psi_{n,z}/ \|\psi_{n,z}\|$, it would imply $-\lambda \in \sigma_{ess}(H^{0}_{\ve})$. 
\par
To conclude, we consider a singular Weyl sequence $\{\xi_n\}$, $\|\xi_n\|= 1$, $\xi_n \rightharpoonup 0$, and $\|(1  + 2\alpha M_{\ve}(-\lambda))\xi_n\|\to 0$. By Proposition \ref{p:decomposition}(ii), the sequence $\{\phi_{n}\}$, $\phi_{n}:=R^{\bosfer}_{\ve}(z)G_{\ve}(-\lambda)\xi_{n}$, $z\in \rho(H^\bosfer_\ve) \cap \CO\backslash[0,+\infty)$ is such that $\phi_n \rightharpoonup 0$ and $\|(H^\bosfer_\ve+\lambda)\phi_{n}\|\to 0$. We are left to prove that $\|\phi_{n}\|$ does not converge to zero. Let us assume that this is not the case. By $R^{0}_{\ve}(z)=R^{\bosfer}_{\ve}(z)+2\alpha G_{\ve}(z)\tau R^{0}_{\ve}(z)$, there follows $\phi_{n}+2\alpha G_{\ve}(z)\tau\phi_{n}=R^{0}_{\ve}(z)G_{\ve}(-\lambda)\xi_{n}$. Thus,
\[
(H^\bosfer_\ve+\lambda)\phi_{n}=(H^{0}_{\ve}-z)(\phi_{n}+2\alpha G_{\ve}(z)\tau\phi_{n})+(z+\lambda)\phi_{n}=G_{\ve}(-\lambda)\xi_{n}+(z+\lambda)\phi_{n}\,,
\]
and so $G_{\ve}(-\lambda)\xi_{n}\to 0$. This contradicts the bounded invertibility of $1+2\alpha M_{\ve}(z)$, which is equivalent to the existence of  $c_{z}>0$ such that $\|(1+2\alpha M_{\ve}(z))\xi_{n}\|\ge c_{z}\|\xi_{n}\|=c_{z}$. Indeed, by \eqref{Mb_prop2-b}, one obtains 
\begin{align*}
0<&\|(1+2\alpha M_{\ve}(z))\xi_{n}\|\le \|(1+2\alpha M_{\ve}(-\lambda))\xi_{n}\|+
2|\alpha|\,\|(M_{\ve}(z)-M_{\ve}(-\lambda))\xi_{n}\|\\
=&\|(1+2\alpha M_{\ve}(-\lambda))\xi_{n}\|+2|\alpha|\,|z+\lambda |
\,\|\breve G_{\ve}(z)G_{\ve}(-\lambda)\xi_{n}\|\,.
\end{align*}
\end{proof}

\begin{lemma}\label{l:esssp}
 For all $\ve>0$ and  $\alpha\in\RE$,    $[0, +\infty) \subseteq \sigma_{ess}( H^\bosfer_\ve)$.
\end{lemma}
\begin{proof}
 The idea is to construct a singular Weyl  sequence supported away from the coincidence lines, which also serves as a singular sequence for the free Hamiltonian. Let $\chi$ be a 
    $C_{0}^\infty(\RE^2)$, spherically symmetric function, such that $\chi(\x) =0$ for all $|\x|\geq 1$ and $\|\chi\| = 1$. 

    For $n\ge 1$,  define a sequence of functions:
    \begin{equation*}
        \chi_{n}(\x) := \frac{\sqrt{2}}{n} \,\chi\left(\frac{\x-\x_{n}}{n}\right), \qquad \x_{n}  =(3n,0) \, .
    \end{equation*}
Note that for any $n \geq 1$, the support of $\chi_{n}(\x)$  lies in an  open disk, which we denote   by $\mathcal D_{n}$, centered at $\x_{n}$ with radius $R=n$. In particular, $\chi_{n}(\x)$ vanishes on the coincidence lines $y = \pm x/2$. Next, for $n \geq 1$, we  define 
\[
        \psi_{n}(\x) := e^{-i\k\cdot \x} \chi_n (\x) \qquad \k = (\ve^{-1}\sqrt{\lambda/2},\sqrt{\lambda/2}) \, ,
\]
and 
 \begin{equation*}
        \psi^{\bosfer}_{n} :=  \frac{1\bosferpm_\bosfer S}{2} \,      \psi_{n}  ,
    \end{equation*}
with the orthogonal projection $\frac{1\bosferpm_\bosfer S}{2}$ defined in Eqs. \eqref{SS1} - \eqref{SS2}. 

Note that for any $n \geq 1$, the support of $\psi_{n}$  lies in the  open disk  $D_n$, while the support of $S\psi_{n}$  lies in an  open disk  centered at $\x_{n}=(-3n,0) $ with radius $R=n$. The supports of $\psi_n$ and $S\psi_n$ are disjoint and both functions vanish on the coincidence lines $y = \pm x/2$. By construction,  $\| \psi^{\bosfer}_{n}\| = 1$, furthermore,  $\psi_n^{\bosfer} \in D(H^\bosfer_\ve)\cap D( H^0_\ve )$ and $H^\bosfer_\ve \psi_n^{\bosfer} =  H^0_\ve \psi_n^{\bosfer}$, and we have (we use the notation $\nabla_{\!\ve} \phi= (\ve\partial_x\phi, \partial_y\phi $): 
    \begin{equation*}
        \begin{split}
            \lVert (H^\bosfer_\ve - \lambda) \psi^{\bosfer}_{n} \rVert
            &=\lVert ( H^0_\ve - \lambda) \psi^{\bosfer}_{n} \rVert \\  
            &\leq  \lVert ( H^0_\ve - \lambda) \psi_{n} \rVert= \left\lVert 2 (\nabla_{\!\ve} \chi_n ) \cdot (\nabla_{\!\ve} e^{-i\k\cdot \x }) + (H^0_\ve \chi_n) e^{-i\k\cdot\x} \right\rVert\\
            &\leq 2 \sqrt \lambda\, \lVert \nabla_{\!\ve} \chi_n \rVert + \lVert H^0_\ve\chi_n  \rVert 
            \\&
            \leq c_{\ve}\,\left(\frac{\sqrt\lambda}{n} + \frac{1}{n^2} \right)\, .
        \end{split}
    \end{equation*}
    In the last inequality, we used the fact that $\chi\in C_{0}^\infty(\RE^2) $, so its gradient and Laplacian are bounded. Therefore,
    \begin{equation*}
        \lim_{n \to \infty} \lVert (H^\bosfer_\ve - \lambda) \psi^{\bosfer}_{n} \rVert = 0 \, .
    \end{equation*}
Finally, for any $\phi\in L^{2}_{\bosfer}(\RE^{2})$ one has, by the Cauchy-Schwarz inequality and by the dominated convergence theorem, 
\[
\left|\langle\phi,\psi^{\bosfer}_{n}\rangle \right|^{2}= \left|\langle\phi,\psi_{n}\rangle \right|^{2} \le   \frac{2} {n^{2}}\,\text{Area}(D_{n}) \int_{\RE^{2}}\left|\phi(\x)\chi\left(\frac{\x-\x_{n}}{n}\right)\right|^{2}
\d\x\to 0 \,.
\]
Therefore $\psi_{n} \rightharpoonup 0$ and this concludes the proof.
\end{proof}
We conclude this section with a result that is a  part of Theorem \ref{th:main}.
\begin{theorem}\label{th:spectrumHve}For all $\ve>0$
\[
\sigma_{ess}(H^{\bosfer}_\ve) = \begin{cases}
\sigma(H^{\bosfer}_\ve)=[0,+\infty)&\alpha\geq   0\\
\left[-\frac{\alpha^2}{4+\ve^{2}}\,,+\infty\right)&\alpha<0\,.
\end{cases}
\]
\end{theorem}
\begin{proof} 
By Lemma \ref{l:esssp}, $[0,+\infty) \subseteq \sigma_{ess}(H^\bosfer_\ve)$ for any $\alpha\in \RE$. \par If $\alpha \geq 0$ the quadratic form associated to $H^\bosfer_\ve$ is positive definite, hence $\sigma(H^\bosfer_\ve)\subseteq [0,+\infty) \subseteq \sigma_{ess}(H^\bosfer_\ve)\subseteq\sigma(H^\bosfer_\ve)$ which concludes the proof in the case $\alpha\geq 0$.

    Let $ \alpha<0$ and $\lambda>0$. Since $M_{od,\ve}(-\lambda)$ is compact (see the proof of Proposition \ref{p:invertibility}), by the Weyl criterion (see, e.g., \cite[Lemma 3 in Section XIII.4]{ReSi}), there follows  
    \[
    \sigma_{ess}(1+ 2\alpha M_\ve(-\lambda)) = 
    \sigma_{ess}(1+ \alpha\,(M_{d,\ve}(-\lambda)\bosferpm_\bosfer M_{od,\ve}(-\lambda)))
    =\sigma_{ess}(1+ \alpha M_{d,\ve}(-\lambda)). 
    \]
    By \eqref{Md-Fou},   $1+\alpha M_{d,\ve}(-\lambda)$ is unitarily equivalent to the multiplication operator $\widehat M_{\lambda,\ve}$ corresponding to the function 
    \[f_{\lambda,\ve}(s)=1+\frac{\alpha}{\sqrt{4\ve^2s^2+(4+\ve^{2})\lambda}}\,.
    \] 
    Hence,
    \[
          \sigma_{ess}\big(1+2\alpha M_\ve(-\lambda)\big) =\sigma_{ess}(\widehat M_{\lambda,\ve})=\sigma(\widehat M_{\lambda,\ve})=\text{range}(f_{\lambda,\ve})=
   \left[1-\frac{|\alpha|}{\sqrt{(4+\ve^{2})\lambda}}\,,1\right]\,.
   \]     
        By Lemma \ref{l:spectrum} there follows that, given $\lambda>0$, $-\lambda \in \sigma_{ess}(H^\bosfer_\ve)$ if and only if $0\in \left[1-\frac{|\alpha|}{\sqrt{(4+\ve^{2})\lambda}}\,,1\right]$. This gives $\sigma_{ess}(H^\bosfer_\ve)=\left[-\frac{\alpha^2}{4+\ve^{2}}\,,0\right)\cup[0,+\infty)$. 
\end{proof}

\section{The Hamiltonian for the light particle\label{s:light}}
Following \cite[Section II.2.1]{AGHKH05} (see also \cite[Section 8.5]{Teta}), we introduce the self-adjoint Hamiltonian $h_{x}$ modeling a delta-interaction of a 1D quantum particle with two centers placed at points $-x/2$ and $x/2$. The operator $h_{x}$ depends on a real parameter 
$\alpha$ associated to  the strength of the interaction; for notational simplicity, we prefer not to explicitly indicate such a dependency. One has
\begin{equation}\label{hx-dom}
\begin{aligned}
D(h_x):= \big\{ & u \in H^2(\RE\backslash\{\pm x/2\})\cap H^1(\RE): [u'\,](\pm x/2)=\alpha u(\pm x/2)\big\}\,,
\end{aligned}
\end{equation}
\begin{equation}\label{hx}
h_x u(y)  = -u''(y) \qquad \text{for a.e. $y \in \RE\backslash \{\pm x/2\}$}\,,
\end{equation}
where $[u'](y)$ denotes the jump of the derivative $u'$ across $y$.\par 
Furthermore, the sesquilinear form associated to $h_{x}$ is given by  
\begin{equation}\label{sesbx1}
b_{x}: H^1(\RE) \times H^1(\RE)\subset L^{2}(\RE)\times
 L^{2}(\RE)\to\CO
 \end{equation}
\begin{equation}\label{sesbx2}
b_x (u,v) := \langle u', v'\rangle + \alpha \bar u(x/2)v(x/2)+\alpha \bar u(-x/2)v(-x/2)\,.
\end{equation}

\subsection{The spectrum of  \texorpdfstring{$h_x$}{hx}.} For any $z\in \CO\backslash [0,+\infty)$ we denote by $r^{0}(z)$ the resolvent of the self-adjoint Laplacian in $L^{2}(\RE)$ with domain $H^2(\RE)$; its integral kernel is 
\[
 r^{0}(y-y';z)  = i \,  \frac{e^{i\sqrt{z} \,| y - y'| }}{2\sqrt{z}} \qquad z\in \CO\backslash [0,+\infty)\,, \quad\Im \sqrt z>0\,. 
\]
Moreover, for any $z\in \CO\backslash [0,+\infty)$ we set: 
\[\brevegg_x(z):   L^2(\RE) \to \CO^{2}\,,\qquad
\brevegg_x(z) u :=\big((r^{0}(z)u) (x/2),  (r^{0}(z)u) (-x/2)\big)\,;
\]
\[
\bfg_x(z): \CO^{2}\to L^2(\RE)\,, \qquad \quad \bfg_x(z):= \brevegg_x(\bar z)^*\,,
\]
\[
(\bfg_x(z)(\zeta_{1},\zeta_{2}))(y)=r^{0}(y-x/2;z)\zeta_{1}+r^{0}(y+x/2;z)\zeta_{2}\,,
\]
and 
\begin{equation*}
\mm_x(z) :\CO^2 \to \CO^2 \qquad \mm_x(z) : = \frac{i}{2 \sqrt{z}}\,\begin{bmatrix}
 1&   e^{i\sqrt{z}\,|x|} \,\\ 
  e^{i\sqrt{z}\,|x|} & 1 
\end{bmatrix}.
\end{equation*}
For all  $z\in \CO\backslash [0,+\infty)$ and for all $\alpha\in\RE$ such that $\det(1+\alpha\,\mm_{x}(z))\not=0$,  the resolvent of $h_x$ is given by 
\begin{equation*}
    (h_x - z)^{-1}= r^{0}(z)-\alpha
\bfg_x(z)\big(1+\alpha\mm_x(z)\big)^{-1}\brevegg_x(z).
\end{equation*}

The values  $z=-\lambda < 0$ , such that  $\det(1+\alpha\mm_{x}(z)) =0$ equivalently, such that 
\begin{equation}\label{eigeneq2}
     \big(\alpha+2\sqrt{\lambda}\,\big)^{2} =\alpha^{2}  e^{-2\sqrt{\lambda}\,|x|}\,,
\end{equation}
correspond to the eigenvalues of $h_x$.
Obviously, \eqref{eigeneq2} has no solution for $\alpha\geq0$. For $\alpha <0$, Eq. 
\eqref{eigeneq2}  admits one or two solutions according to the values of $\alpha$ and $x$. We summarize the spectral properties of $h_{x}$ in the following
\begin{lemma}\label{l:spectrumhx}
    \item (i) $\sigma_{ac}(h_x) =\sigma_{ess}(h_x) = [0,+\infty)$; 
    \item (ii) $\sigma_{sc}(h_x) = \varnothing$;
    \item (iii)
     \[\sigma_{p}(h_x)=\sigma_{d}(h_{x}) =\begin{cases}\varnothing&\alpha \geq 0\\
    \{ -\lambda_0(x)\}& \alpha<0,\ |x|\leq 2/|\alpha|\\
    \{ -\lambda_0(x), -\lambda_1(x)\}&\alpha<0,\ |x|> 2/|\alpha|\,.\end{cases}\]
\end{lemma}
To better describe the behavior of the eigenvalues with respect to $x$, we introduce the  Lambert $W$-function,  defined as the solution of the equation 
\[W(x)e^{W(x)} = x\,.
\]
As consequence of such a definition, one gets the following
\begin{lemma}\label{l:lambda0} Let $\alpha<0$. The functions $\lambda_{0}$ and $\lambda_{1}$ are represented as 
\[
 \lambda_0:\RE\to (0,+\infty)\,,\qquad  \lambda_0(x)  = \left(\frac{ W\left(\frac{|\alpha||x|}{2} \,e^{-\frac{|\alpha||x|}{2}} \right)}{|x|} + 
   \frac{|\alpha|}{2}\right)^{\!\!2}\,,\]
   \[
 \lambda_1:\RE\backslash[-2/|\alpha|,2/|\alpha|]\to (0,+\infty)\,,\qquad   \lambda_1(x)  = \left(\frac{ W\left(-\frac{|\alpha||x|}{2}\, e^{-\frac{|\alpha| |x|}{2}} \right)}{|x|} + \frac{|\alpha|}{2}\right)^{\!\!2};
\]
\begin{enumerate}[(i)]
\item both $\lambda_{0}$ and $\lambda_{1}$ are even functions; 
\item $\alpha^2/4<\lambda_0(x)\leq\alpha^2$, $0<\lambda_1(x)<\alpha^2/4$;
\item $\lambda_0(0)=\alpha^2$, $ \lim_{x\to (\pm2/|\alpha|)_\pm}\lambda_1(x)= 0$, $\lim_{x\to\pm\infty} \lambda_0(x) =\lim_{x\to\pm\infty } \lambda_1(x) =\alpha^2/4$;
\item $\lambda_0$ is strictly decreasing on $(0,+\infty)$, $\lambda_1$ is strictly increasing on $(2/|\alpha|,+\infty)$;
\item $\lambda_0 \in   C^\infty(\RE\backslash\{0\})$, $\lambda_1 \in   C^\infty(\RE\backslash[-2/|\alpha|,2/|\alpha|\,])$.
\end{enumerate}
\end{lemma}

\begin{remark}Even though $\lambda_1(x) = 0$ is a solution of \eqref{eigeneq2}, zero is not an eigenvalue of the Hamiltonian $h_x$. To prove that this is indeed the case, assume on the contrary that there is eigenfunction $u_0$ such that $h_x u_{0}=0$. This means that on $\RE \backslash \{\pm \frac{x}{2}\}$, we have $-u^{\prime \prime}_0 = 0$. Any function in $D(h_x)$ that satisfies  such an equation must be equal to zero  on $(-\infty, \frac{x}{2}]$ and $[\frac{x}{2}, +\infty)$. Also, it has to be linear on $[-\frac{x}{2},\frac{x}{2}]$, i.e., $u_0(x)  = A x + B $ for some  constants $A$ and $B$. By imposing the continuity and the boundary conditions required in $D(h_x)$ there follows that it must be $A = B = 0$.
\end{remark}

By Lemma \ref{l:lambda0}, there follows that $-\lambda_{0}$ reaches its minimum for $x=0$, 
\[
\min_{x\in\RE}(-\lambda_{0}(x))= -\lambda_{0}(0)= -\alpha^2\,,
\]
  and $-\lambda_1$ reaches its infimum for $|x| \to +\infty $, 
\[
\lim_{|x| \to +\infty }-\lambda_1(x) = - \alpha^2/4\,. 
\]

We compare the eigenvalue $-\lambda_{0}(x)$ and $-\lambda_1(x)$ for the values $\alpha=-1,-2$ in Figure \ref{1D_plot_comparison}.
\begin {figure}[htbp]
\centering
\includegraphics[width=12cm]{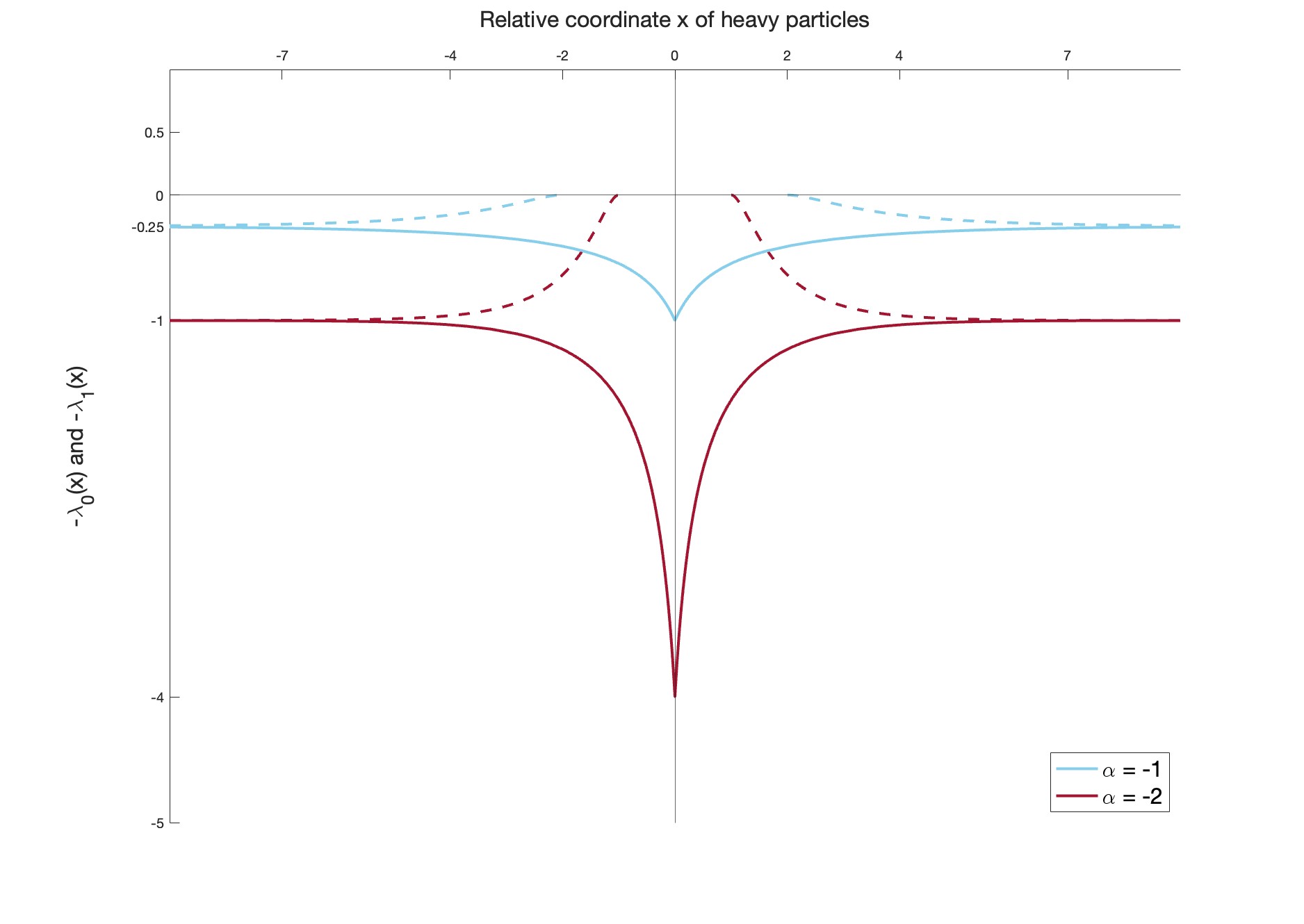}
\caption{\footnotesize {Plot of the eigenvalues of the light-particle Hamiltonian as functions of $x$: $-\lambda_{0}$ is represented by the solid line;  $-\lambda_1$ by the dashed line.}}
\label{1D_plot_comparison}
\end{figure}

By Lemmata \ref{l:spectrumhx} and  \ref{l:lambda0} there follows the following theorem
\begin{theorem}\label{th:lowerboundbH} 
For all $\alpha\geq 0$ and $x\in\RE$, the quadratic form $b_x$ and the associated self-adjoint operator $h_x$ are non-negative. For all $\alpha<0$ and $x\in\RE$, $-\alpha^2$ is a lower bound for the quadratic form $b_x$ and for the associated self-adjoint operator $h_x$. 
\end{theorem}
For later use, we introduce the  normalized eigenfunction corresponding to the lowest eigenvalue of $h_x$ (see, e.g.,  \cite{AGHKH05} and \cite{AT1}):  
\begin{equation}\label{psiBOx}
 \psi^{BO}_{x}(y) :=    N(x) 
\left(e^{-\sqrt{\lambda_{0}(x)}\,|x/2-y|} + e^{-\sqrt{\lambda_{0}(x)}\,|x/2+y|}\right) ,
\end{equation}
where $N(x)$ is the normalization constant  given by 
 \begin{equation}\label{Normalization_constant}
        N(x)  := \left(\frac{\sqrt{\lambda_{0}(x)}}{2\left(1+ e^{-\sqrt{\lambda_{0}(x)}\,|x|}\big(1+ \sqrt{\lambda_{0}(x)}\,|x|\big)\right)} \right)^{1/2} .
    \end{equation}
    Further,
    \begin{equation}\label{Px}
   {P}_x\phi:= \langle\psi^{BO}_x, \phi\rangle\,\psi^{BO}_x
    \end{equation}
    denotes the corresponding spectral projection. 
We point out that $N(x)$ and $\psi^{BO}_x$ are even in the parameter $x$.

\section{Reduction to an effective Hamiltonian for the heavy particles subsystem\label{s:reduction}}

We provide a dimensional reduction of $H^\bosfer_\ve$ following the abstract scheme in \cite{Krejcirik-etal-m2018}. At first, we need to introduce the convenient notion of section of a function: given any $\phi\in L^{2}(\RE^{2})$,  its $x$-section is defined by $\phi_x(y):= \phi(x,y)$. One has that, for a.e. $x\in \RE$, $\phi_x$ is measurable and square integrable;  furthermore,  the function $x\mapsto \|\phi_x\|_{L^2(\RE)}$ is square integrable and
\[
\|\phi\|^{2}_{L^2(\RE^2)} =\int_{\RE} \|\phi_{x}\|_{L^2(\RE)}^2 \d x\,.
\]
Note that, by abuse of notation, we used the same symbol to denote an element in $L^{2}(\RE)$ and any of its representatives.\par
Similarly but less obviously (see, e.g., \cite[Section 5.6]{KJF}),  for a.e. $x\in \RE$, the $x$-section $\phi_x$ of  any $\phi\in H^1(\RE^2)$ is an absolutely continuous function having a square integrable derivative, the function  $x\mapsto \|\phi_x\|_{H^1(\RE)}$ is square integrable and 
\[
\|\phi\|^{2}_{H^1(\RE^2)} =\int_{\RE^2} |\partial_x \phi(x,y)|^2 \d\x + \int_\RE \|\phi_{x}\|_{H^1(\RE)}^2 \d x\,.
\]
Then, we consider the closed, lower bounded  sesquilinear   form $b_x$ defined in \eqref{sesbx1}-\eqref{sesbx2} and the associated  lower bounded self-adjoint operator $h_x$ in  $L^2(\RE)$ defined in \eqref{hx-dom}-\eqref{hx}. 

Recalling the definition of $\BB^\bosfer_{\ve}$ in \eqref{BBve1-b}-\eqref{BBve2-b}, by the discussion above, $b_x(\varphi_x,\psi_x)$ is well defined  for any pair of functions  $(\varphi,\psi) \in H_{\bosfer}^1(\RE^2)\times H_{\bosfer}^1(\RE^2)$
and 
\begin{equation}\label{directintegralB}
\BB^\bosfer_{\ve}(\varphi, \psi)= \int_{\RE^2} \ve^2 {\partial_x \overline\varphi(x,y)}\partial_x \psi(x,y) \d\x  +\int_\RE  b_x(\varphi_x,\psi_x)\d x .    
\end{equation}
This identification, together with Theorem \ref{th:lowerboundbH} gives the following result which completes Theorem \ref{th:spectrumHve}.
\begin{proposition}\label{p:spectrumHveLB}
For any $\ve>0$ and $\alpha<0$,  $\sigma(H^{\bof}_\ve)\subset [-\alpha^2,+\infty)$.
\end{proposition}

To proceed, we define the orthogonal projection $\mathcal{P}$ in $L^2(\RE^2)$ associated to the projection ${P}_x$ defined in Eq. \eqref{Px}:  
\begin{equation*}
 \mathcal{P} :L^2(\RE^2)\to  L^2(\RE^2) \qquad   \left(\mathcal{P} \phi\right)_x  := P_x \phi_x .
\end{equation*} 
Additionally we set
\[ 
\mathcal{P}^\perp  :=1 - \mathcal P. 
\]
For later convenience, sometimes we regard $\psi^{BO}_x$  as a function of two variables, denoted by  $\psi^{BO}$, with the obvious identification $\psi^{BO}(x,y)\equiv \psi^{BO}_x(y)$. With this notation, one has 
\begin{equation*}
\mathcal{P}:L^{2}(\RE^{2})\to L^{2}(\RE^{2})\,,\qquad
\mathcal{P} \phi(x,y) := 
\psi^{BO}(x,y) f_\phi(x) \, ,
\end{equation*}
 where   
 \[
    f_\phi(x) := \int_{\RE} {\psi^{BO}(x,y)} \phi(x,y)\, \d y =\langle\psi_{x}^{BO},\phi_{x}\rangle  \,.
    \]
Let us point out the inequality
    \begin{equation*}
\|  f_\phi\|^{2}_{L^2(\RE)} \leq \int_{\RE}\left(\int_{\RE} {|\psi^{BO}(x,y)}|\, |\phi(x,y)|\, \d y\right)^{\!2}\!\d x\le \int_{\RE}\|\psi^{BO}_{x}\|^{2}_{L^2(\RE)} \|\phi_{x}\|^{2}_{L^2(\RE)}\,\d x=\|\phi\|^{2}_{L^2(\RE^2)}. 
    \end{equation*}
\begin{remark}\label{remL2} Notice that, since $\psi^{BO}$ is an even function with respect to the $x$ variable, $f_{\phi}$ is an even function whenever $\phi\in L^{2}_{\bos}(\RE^{2})$ and is an odd function whenever $\phi\in L^{2}_{\fer}(\RE^{2})$. Hence, $\PP$ is an orthogonal projector in $L^{2}_{\bosfer}(\RE^{2})$ as well.
\end{remark}     
\begin{lemma}\label{l:PH}  i) $\mathcal P$ is a bounded operator in $H^1(\RE^2)$; 
ii) the commutator $[\partial_{x},\PP]:H^{1}(\RE^{2})\subset L^2(\RE^2)\to L^2(\RE^2)$ extends to a bounded operator in $L^{2}(\RE^{2})$ and 
\begin{equation*}
\|\,[\partial_{x}, \PP]\,\| \leq \delta\, ,
\end{equation*}
with 
\begin{equation*}
\delta := 2\left( \sup_{x\in\RE} \int |\partial_x \psi^{BO}(x,y)|^2 \d y \right)^{1/2}
\end{equation*}
and the bound $\delta \leq 4 |\alpha|$.
\end{lemma}
\begin{proof}
In the course of the proof, for the sake of brevity, we will often omit the dependence on $x$ in $\NN(x)$ and $\lambda_0(x)$. From the formula \eqref{psiBOx}, one gets
\[
\partial_y \psi^{BO}(x,y) =     {\NN}\sqrt{\lambda_0 } 
\left(\sgn(x/2-y)e^{-\sqrt{\lambda_0 }\,|x/2-y|} - \sgn(x/2+y)e^{-\sqrt{\lambda_0 }\,|x/2+y|}\right) ,
\]
and so
\[\begin{aligned}
\int_\RE \left|\partial_y\psi^{BO}(x,y)\right|^2 \d y  
= & 2 \NN^2 \sqrt{\lambda_0 }\left(1+e^{-\sqrt{\lambda_0 }\,|x|} -\sqrt{\lambda_0 }\,|x| e^{-\sqrt{\lambda_0 }\,|x|} \right) \\ =&  \lambda_0 \,\frac{1+e^{-\sqrt{\lambda_0 }\,|x|} -\sqrt{\lambda_0 }\,|x| e^{-\sqrt{\lambda_0 }\,|x|} }{1+e^{-\sqrt{\lambda_0 }\,|x|} +\sqrt{\lambda_0 }\,|x| e^{-\sqrt{\lambda_0 }\,|x|} } \\
\leq &\lambda_0  \leq\alpha^2\,.
\end{aligned}\]
Hence,
\begin{align*}
\int_{\RE^2}\left|\partial_y\mathcal{P} \phi(x,y)\right|^2 \d\x  = &\int_{\RE^2} \left|\partial_y\psi^{BO}(x,y)\right|^2 |f_\phi(x)|^2 \d x\,\d y  \\
  \leq & \left(\sup_{x \in \RE} \int_\RE \left|\partial_y\psi^{BO}(x,y)
\right|^2 dy \right)\|f_{\phi}\|^2_{L^2(\RE)}\\
\leq &\alpha^{2}\| \phi \|^2_{L^2(\RE^2)}\, .
\end{align*} 
Defining
\begin{equation*}
    \tilde f_\phi(x) :=  \int_{\RE} \big(\partial_x\psi^{BO}(x,y)\big) \phi(x,y) \d y \,,
    \end{equation*}
one has
     \[
\|  \tilde f_\phi\| \leq \frac{\delta}{2}\,\|\phi\|\,. 
    \]
    Since
    \begin{equation*}
    \partial_x \PP\phi  = \big(\partial_x\psi^{BO}\big) f_\phi  +\psi^{BO}\tilde f_\phi +  \psi^{BO}f_{\partial_x\phi},
    \end{equation*}
    there holds 
  \[\begin{aligned}
    \|\partial_x \PP \phi\|  \leq & \|\big(\partial_x\psi^{BO}\big) f_\phi\|  +\|\psi^{BO}\tilde f_\phi \| +\|\psi^{BO} f_{\partial_x\phi} \|  \\ 
     \leq &
     \frac\delta2\,\|f_\phi\|  +\|\tilde f_\phi \| + \|\partial_x \phi \| \leq (1+\delta) \|\phi\|_{H^1(\RE^2)} . 
   \end{aligned}
   \]
    In a similar way, one obtains 
 \[
    [\partial_x, \PP]\phi  = \big(\partial_x\psi^{BO}\big) f_\phi  +\psi^{BO}\tilde f_\phi ,
    \]
    and   
    \[
    \|\,[\partial_x, \PP]\phi\|  
     \leq \delta\, \|\phi\|. 
    \]   
We are left to prove the bound on $\delta$.
    From the formula \eqref{psiBOx}, one can explicitly compute $\partial_x\psi^{BO}$:
\[
\begin{aligned}
\partial_x \psi^{BO}(x,y) =&
\underbrace{\frac{\NN'}{\NN} \psi^{BO}}_{\varphi^{BO}_1}
-
\underbrace{\frac{{\NN}\lambda_0'}{2\sqrt{\lambda_0}}  
\left(|x/2-y|e^{-\sqrt{\lambda_0}|x/2-y|} +|x/2+y|e^{-\sqrt{\lambda_0}|x/2+y|}\right)}_{\varphi^{BO}_2} 
\\
&-\underbrace{\frac{{\NN} \sqrt{\lambda_0}}{2} 
\left(\sgn(x/2-y)e^{-\sqrt{\lambda_0}|x/2-y|} + \sgn(x/2+y)e^{-\sqrt{\lambda_0}|x/2+y|}\right)}_{\varphi^{BO}_3}. 
\end{aligned}
\]
Computing the norm of $\varphi^{BO}_1$ one obtains
\[
\|\varphi^{BO}_1\|^{2} =\left( \frac{\NN'}{\NN}\right)^2.
\]
From Eq. \eqref{Normalization_constant}, by a straightforward calculation there follows 
    \begin{equation*}
            \frac{\NN'}{\NN} = \frac{(\sqrt{\lambda_0})^{\prime}}{2\sqrt{\lambda_0}} + \frac{\frac{(\sqrt{\lambda_0})^{\prime}}{\sqrt{\lambda_0}} \left(\lambda_0 x^2 + \frac{\lambda_0^{3/2}}{(\sqrt{\lambda_0})^{\prime}}x\right)}{2 \left(  \sqrt{\lambda_0}|x|+e^{\sqrt{\lambda_0}|x|}+1\right)} \, .
    \end{equation*}
We notice that 
\[
\lambda_0(x) = \frac{\alpha^2}{4} \nu^2(|\alpha| x/ 2) \qquad x\geq 0
\]
with 
\[
\nu(x) = \frac{W(xe^{-x})}{x}+1.
\]
Hence, 
\begin{equation*}
    \left\lvert\frac{(\sqrt{\lambda_0})^{\prime}}{2\sqrt{\lambda_0}} \right\rvert = \frac{|\alpha|}{4}  \left|\frac{\nu'(ax/2)}{\nu(ax/2)}\right|.  
\end{equation*}
From the identity $W(x)e^{W(x)} = x$, $x\geq 0$, we infer 
\[
\nu(x) = e^{-x\nu(x)}+1.
\]
Hence, 
\[
\frac{\nu'(x)}{\nu(x)} = - \frac{1}{e^{x\nu(x)}+x}
\]
and $|\nu'(x)/\nu(x)| \leq \frac{1}{e^{x}+x} \leq 1$. We deduce the bound:  
\begin{equation} \label{sqrtlambdaprime}
    \left| \frac{(\sqrt{\lambda_0})^{\prime}}{2\sqrt{\lambda_0}} \right| \leq \frac{|\alpha|}{4}.
\end{equation}
Recalling also that $\sqrt{\lambda_0} \leq |\alpha|$, we infer
    \begin{equation*}
 \left|\frac{\frac{(\sqrt{\lambda_0})^{\prime}}{\sqrt{\lambda_0}} \left(\lambda_0 x^2 + \frac{\lambda_0^{3/2}}{(\sqrt{\lambda_0})^{\prime}}x\right)}{2 \left(  \sqrt{\lambda_0}|x|+e^{\sqrt{\lambda_0}|x|}+1\right)} 
    \right| 
    \leq \frac{\frac{|\alpha|}{2}\lambda_0 x^2+|\alpha|\sqrt{\lambda_0}|x|
    }{2 \left(  \sqrt{\lambda_0}|x|+e^{\sqrt{\lambda_0}|x|}+1\right)}\leq \frac{|\alpha|}{4}, 
        \end{equation*}
        where we used the  inequality $\frac{\frac{1}{2}s^2+s}{2 (s +e^{s}+1)} \leq \frac14$, $s\geq 0$. We deduce that   \[ \left| \frac{\NN'}{\NN}\right|  \leq \frac{|\alpha|}{2}\qquad \text{and} \qquad  
        \|\varphi^{BO}_1\|_{L^2(\RE,dy)}^2 =  \left( \frac{\NN'}{\NN}\right)^2  \leq \frac{\alpha^2}{4} .\]
Next we focus attention on $\varphi^{BO}_2$.
\[
\begin{aligned}
\|\varphi^{BO}_2\|^2 =& \frac{({\NN}\lambda_0')^2}{4(\lambda_0)^{5/2}}\left(1+\frac{(\sqrt{\lambda_0}|x|)^3}{3}e^{-\sqrt{\lambda_0}|x|} +(1+\sqrt{\lambda_0}|x|) e^{-\sqrt{\lambda_0}|x|} \right) \\
=&
\frac{(\lambda_0')^2}{2(\lambda_0)^{2}}\frac{1+\frac{(\sqrt{\lambda_0}|x|)^3}{3}e^{-\sqrt{\lambda_0}|x|} +(1+\sqrt{\lambda_0}|x|) e^{-\sqrt{\lambda_0}|x|}}{1+e^{-\sqrt{\lambda_0}|x|} +\sqrt{\lambda_0}|x| e^{-\sqrt{\lambda_0}|x|}} \leq \frac{(\lambda_0')^2}{(\lambda_0)^{2}} \leq\alpha^2,
\end{aligned}
\]
where in the latter step we used the bound  $|\lambda_0'/\lambda_0| \leq |\alpha|$ that can be easily obtained from the bound \eqref{sqrtlambdaprime}. 

We conclude with a bound on $\|\varphi^{BO}_3\|$.
\[
\begin{aligned}
\|\varphi^{BO}_3\|^2  =& \frac{{\NN}^2\sqrt{\lambda_0}}{2}\left(1-e^{-\sqrt{\lambda_0}|x|} +\sqrt{\lambda_0}|x| e^{-\sqrt{\lambda_0}|x|} \right) \\
=&
 \frac{\lambda_0}{4}\frac{1-e^{-\sqrt{\lambda_0}|x|} +\sqrt{\lambda_0}|x| e^{-\sqrt{\lambda_0}|x|} }{1+e^{-\sqrt{\lambda_0}|x|} +\sqrt{\lambda_0}|x| e^{-\sqrt{\lambda_0}|x|}} \leq  \frac{\lambda_0}{4} \leq \frac{\alpha^2}{4}.
\end{aligned}
\]
Summing up we obtain $\|\partial_x \psi^{BO}\|  \leq \sum_{j=1}^3 \|\varphi_j^{BO}\|  \leq 2|\alpha|$ and $\delta \leq 4|\alpha| $. 
\end{proof}
\begin{remark}\label{remH1}  By Remark \ref{remL2} and by point i) in Lemma \ref{l:PH}, $\mathcal P$ is a bounded operator in $H_{\bosfer}^1(\RE^2)$.
\end{remark}

By Proposition \ref{p:spectrumHveLB} there follows that $-\alpha^2$ is a lower bound for the quadratic form $\BB^\bosfer_{\ve}$ and for the associated self-adjoint operator $H^\bosfer_\ve$. Preferring to work with non negative quadratic forms and operators, we define the sesquilinear form 
\begin{equation*}
q_x (u,v) := b_x(u,v) +\alpha^2
\end{equation*}
with corresponding self-adjoint, non-negative operator $h_x +\alpha^2$.
Then, we define the sesquilinear  form $\QQ_\ve^\bosfer$ with $D(\QQ_\ve^\bosfer) = H_{\bosfer}^1(\RE^2)$ by
\be\label{r:Qve}
\QQ^\bosfer_{\ve} (\varphi, \psi):= \BB^\bosfer_{\ve} (\varphi, \psi) +\alpha^2(\varphi, \psi)_{L^2(\RE^2)} = \int_{\RE^2} \ve^2 {\partial_x \overline\varphi(x,y)}\partial_x \psi(x,y) \d\x  +\int_\RE  q_x(\varphi_x,\psi_x)\d x \,.
\ee
The corresponding non-negative self-adjoint operator is \[\Lscr_\ve^\bosfer:= H^\bosfer_\ve+\alpha^2\,.
\]

By Remarks \ref{remL2} and \ref{remH1}, the following sesquilinear form in $L^{2}_{\bosfer}(\RE^{2})$ is well defined: 
\begin{equation*}
 D(\widehat{\QQ}^{\bosfer}_\ve) =   H_{\bosfer}^1(\RE^2) \qquad \widehat{\QQ}^{\bosfer}_\ve(\phi, \psi)   : = {\QQ}^{\bosfer}_\ve (\mathcal{P} \phi, \mathcal{P} \psi)+{\QQ}^{\bosfer}_\ve \big(\mathcal{P}^{\perp} \phi, \mathcal{P}^{\perp} \psi\big).
\end{equation*}

\begin{proposition}\label{p:QwidehatQ} For all $\phi,\psi \in H^1_\bosfer(\RE^2)$ there holds 
\begin{equation*}
\big|\QQ^{\bosfer}_\ve(\phi,\psi) - \widehat\QQ^{\bosfer}_\ve(\phi,\psi) \big|  \\ 
\leq  
2 \ve \delta \Big( 
\sqrt{\QQ_\ve^\bosfer(\phi)} \, \| \psi \|+\|\phi\| \, \sqrt{\widehat \QQ_\ve^\bosfer( \psi )}\ \Big).
\end{equation*}
\end{proposition}
\begin{proof} By the definition of $\widehat\QQ_\ve^\bosfer$ there follows 
\[
\QQ_\ve^\bosfer(\phi,\psi) - \widehat\QQ_\ve^\bosfer(\phi,\psi) = 
\QQ_\ve^\bosfer(\PP \phi,\PP^\perp \psi) + \QQ_\ve^\bosfer(\PP^\perp\phi,\PP\psi). 
\]
Defining $S_\ve \phi := -i\ve \partial_x$, since $q_x(\PP_x u, \PP_x^\perp u) = q_x(\PP_x^\perp u, \PP_x u)= 0$, by Definition \ref{r:Qve},  there follows 
\[
\begin{aligned}
\QQ_\ve^\bosfer(\phi,\psi) - \widehat\QQ_\ve^\bosfer(\phi,\psi) = &
\langle S_\ve \PP \phi,S_\ve \PP^\perp \psi\rangle  + \langle S_\ve \PP^\perp\phi,S_\ve \PP\psi\rangle \\
= &
\langle S_\ve \PP \phi,S_\ve (1-\PP) \psi\rangle + \langle S_\ve  (1-\PP)\phi,S_\ve \PP\psi\rangle \\
= &
\langle S_\ve \PP \phi,S_\ve  \psi\rangle  +\langle S_\ve  \phi,S_\ve \PP\psi\rangle  
-2\langle S_\ve \PP \phi,S_\ve \PP \psi\rangle  
\\
=&\langle [S_\ve, \PP] \phi,S_\ve  \psi\rangle  + \langle  \PP S_\ve \phi,S_\ve  \psi\rangle  \\
&+\langle S_\ve  \phi,S_\ve \PP\psi\rangle -2\langle [S_\ve, \PP] \phi,S_\ve \PP \psi\rangle -2\langle  \PP S_\ve \phi,S_\ve \PP \psi\rangle .
\end{aligned}
\]
Note that 
\[
\begin{aligned}
-2\langle \PP S_\ve \phi,S_\ve \PP \psi\rangle  =&
-2\langle   S_\ve \phi, \PP S_\ve \PP \psi\rangle  \\
=&  -2\langle   S_\ve \phi, [\PP, S_\ve] \PP \psi\rangle  - 2\langle   S_\ve \phi, S_\ve\PP^2 \psi\rangle  \\
=&  2\langle   S_\ve \phi, [S_\ve,\PP] \PP \psi\rangle - 2\langle   S_\ve \phi, S_\ve\PP \psi\rangle  ,
\end{aligned}
\]
and
\[
\langle \PP S_\ve \phi,S_\ve  \psi\rangle  = \langle   S_\ve \phi,\PP S_\ve  \psi\rangle 
= -  \langle   S_\ve \phi,[ S_\ve,\PP]  \psi\rangle  + \langle   S_\ve \phi, S_\ve \PP   \psi\rangle .
\]
Hence, 
\[
\begin{aligned}
\QQ_\ve^\bosfer(\phi,\psi) - \widehat\QQ_\ve^\bosfer(\phi,\psi) =&\langle[S_\ve, \PP] \phi,S_\ve  \psi\rangle +\big(-  \langle  S_\ve \phi,[ S_\ve,\PP]  \psi\rangle + \langle S_\ve \phi, S_\ve \PP   \psi\rangle\big) \\
&+\langle S_\ve  \phi,S_\ve \PP\psi)-2\langle[S_\ve, \PP] \phi,S_\ve \PP \psi\rangle+ \big( 2\langle  S_\ve \phi, [S_\ve,\PP] \PP \psi\rangle - 2\langle  S_\ve \phi, S_\ve\PP \psi\rangle\big)\\
= &
\langle[S_\ve, \PP] \phi,S_\ve  \psi\rangle - \langle  S_\ve \phi,[ S_\ve,\PP]  \psi\rangle -2\langle[S_\ve, \PP] \phi,S_\ve \PP \psi\rangle+ 2\langle  S_\ve \phi, [S_\ve,\PP] \PP \psi\rangle\\
= &
\langle[S_\ve, \PP] \phi,S_\ve \PP^\perp \psi\rangle - \langle S_\ve \phi,[ S_\ve,\PP]  \PP^\perp\psi\rangle -\langle[S_\ve, \PP] \phi,S_\ve \PP \psi\rangle+ \langle  S_\ve \phi, [S_\ve,\PP] \PP \psi\rangle\,.
\end{aligned}
\]
By Lemma \ref{l:PH}, this gives
\[
\begin{aligned}
&\big|\QQ_\ve^\bosfer(\phi,\psi) - \widehat\QQ_\ve^\bosfer(\phi,\psi) \big|  \\ 
\leq &  
\ve \delta \Big( 
\|\ve \partial_x \phi\| \big(  \|\PP \psi \|  +  \|\PP^\perp \psi \|\big)
+
\|\phi\| \big(  \| \ve \partial_x\PP \psi \| +  \|\ve \partial_x \PP^\perp \psi \|\big) \Big).
\end{aligned}
\]
By the obvious inequality $a+b \leq 2 \sqrt{a^2+b^2}$, there follows 
\[
\begin{aligned}
&\big|\QQ_\ve^\bosfer(\phi,\psi) - \widehat\QQ_\ve^\bosfer(\phi,\psi) \big|  \\ 
\leq &  
2 \ve \delta \Big( 
\|\ve \partial_x \phi\| \sqrt{\|\PP \psi \|^2  +  \|\PP^\perp \psi \|^2}
+
\|\phi\| \sqrt{\| \ve \partial_x\PP \psi \|^2 +  \|\ve \partial_x \PP^\perp \psi \|^2} \ \Big) \\ 
\leq &  
2 \ve \delta \Big( 
\sqrt{\QQ_\ve^\bosfer(\phi)} \, \| \psi \|+\|\phi\| \, \sqrt{\widehat \QQ_\ve^\bosfer( \psi )}\ \Big),
\end{aligned}
\]
where in the latter inequality we used the fact that the quadratic form $q_x$ is non-negative. 
\end{proof} 
In the next Lemma we recall several results  from \cite{Krejcirik-etal-m2018}.
\begin{lemma}[{\cite[Lemmata 2.2 and 2.3]{Krejcirik-etal-m2018}}]\label{l:spectrum-widehatL}
The quadratic form $\widehat\QQ_\ve^\bosfer$ is closed, non-negative, and determines a unique self-adjoint, non-negative operator $\widehat\Lscr_\ve^\bosfer$ in $L^2_\bosfer(\RE^2)$.  Furthermore, 
\begin{enumerate}[(i)]
\item \[\widehat \Lscr_{\ve}^{\bosfer} = \widehat\Lscr_{\ve,\PP}^{\bosfer}\oplus \widehat
\Lscr_{\ve,\PP^\perp}^{\bosfer}\,,\qquad
\sigma(\widehat \Lscr_{\ve}^{\bosfer}) = \sigma(\widehat\Lscr^{\bosfer}_{\ve,\PP}) \cup \sigma(\widehat\Lscr^{\bosfer}_{\ve,\PP^\perp})\,,
\] 
where 
\[\widehat\Lscr^{\bosfer}_{\ve,\PP} : D(\widehat\Lscr^{\bosfer}_{\ve,\PP}) \subset \ran(\PP| L^2_{\bosfer}(\RE^2)) \to \ran(\PP| L^2_\bosfer (\RE^2))\,,
\]
\[
\widehat\Lscr^{\bosfer}_{\ve,\PP^\perp} : D(\widehat\Lscr^{\bosfer}_{\ve,\PP^\perp}) \subset \ran(\PP^\perp| L^2_\bosfer(\RE^2)) \to \ran(\PP^\perp |L^2_\bosfer(\RE^2))\,;
\]
 \item 
 \[
 \langle\psi, \widehat\Lscr^{\bosfer}_{\ve,\PP}\psi\rangle = \QQ_\ve^\bosfer(\PP\psi,\PP\psi) \geq 0\,, \qquad \forall \psi \in D(\widehat\Lscr^{\bosfer}_{\ve,\PP}), 
 \]
 \[
 \langle\psi, \widehat\Lscr^{\bosfer}_{\ve,\PP^\perp}\psi\rangle = \QQ_\ve^\bosfer(\PP^\perp\psi,\PP^\perp\psi) \geq \inf_{x\in \RE} (-\lambda_1(x)) +\alpha^2 = \frac34\alpha^2 \,,\qquad \forall \psi \in D(\widehat\Lscr^{\bosfer}_{\ve,\PP^{\perp}})\,. 
 \]
\end{enumerate}
\end{lemma}

\begin{proposition}\label{p:diff} For all $\lambda \in \RE $, $\phi \in D(\Lscr_\ve^\bosfer)$, and $\psi \in D(\widehat\Lscr_\ve^\bosfer)$, there holds 
\[\begin{aligned}
&\left|
      \big\langle\phi , (\widehat \Lscr_\ve^\bosfer - \lambda  ) \psi\big\rangle
    -
    \big\langle( \Lscr_\ve^\bosfer - \lambda   ) \phi, \psi  \big\rangle\right| \\ 
\leq & 2 \ve \delta \left( \|\psi\| \sqrt{\left( \|(\Lscr_\ve^\bosfer -\lambda )\phi\| +\lambda_+ \,  \|\phi\|\right) \|\phi\|} + \|\phi\| \sqrt{\left( \|(\widehat\Lscr_\ve^\bosfer -\lambda )\psi\| + \lambda _+ \,  \|\psi\|\right) \|\psi\|}\,\right)\,.
\end{aligned}\]
\end{proposition}
\begin{proof}
    For all $\phi \in D(\Lscr_\ve^\bosfer)$, there holds 
    \[\QQ_\ve^\bosfer(\phi) = \langle\phi,(\Lscr_\ve^\bosfer -\lambda )\phi\rangle + \lambda \,  \|\phi\|^2.\]
    Hence, 
    \[\QQ_\ve^\bosfer(\phi) \leq \left( \|(\Lscr_\ve^\bosfer -\lambda )\phi\| +\lambda_+ \,  \|\phi\|\right) \|\phi\|\,.
    \]
    Similarly, for all $\psi \in D(\widehat \Lscr_\ve^\bosfer)$,
    \[\widehat\QQ_\ve^\bosfer(\psi) \leq \left( \|(\widehat\Lscr_\ve^\bosfer -\lambda )\psi\| + \lambda _+ \,  \|\psi\|\right) \|\psi\|.\]
    For all  $\phi \in D(\Lscr_\ve^\bosfer)$ and $\psi \in D(\widehat \Lscr_\ve^\bosfer)$, there holds 
    \[
      \big\langle\phi , (\widehat \Lscr_\ve^\bosfer - \lambda  ) \psi\big\rangle
    -
    \big\langle \Lscr_\ve^\bosfer - \lambda   ) \phi, \psi  \big\rangle
        = 
    \big\langle\phi , \widehat \Lscr_\ve^\bosfer \psi\big\rangle
    -
    \big\langle \Lscr_\ve^\bosfer   \phi, \psi  \big\rangle
        = 
       \widehat \QQ_\ve^\bosfer (\phi , \psi)
    -
 \QQ_\ve^\bosfer (    \phi, \psi  ).
\]
Hence, 
\[\begin{aligned}
&\left|
      \big\langle\phi , (\widehat \Lscr_\ve^\bosfer - \lambda  ) \psi\big\rangle
    -
    \big\langle( \Lscr_\ve^\bosfer - \lambda   ) \phi, \psi  \big\rangle\right| \\ 
\leq & 2 \ve \delta \left( \|\psi\| \sqrt{\left( \|(\Lscr_\ve^\bosfer -\lambda )\phi\| +\lambda_+ \,  \|\phi\|\right) \|\phi\|} + \|\phi\| \sqrt{\left( \|(\widehat\Lscr_\ve^\bosfer -\lambda )\psi\| + \lambda _+ \,  \|\psi\|\right) \|\psi\|}\,\right).
\end{aligned}\]
\end{proof}

\begin{lemma}\label{l:resolvent-difference} (i) Defining 
\[
d_\ve \equiv d_\ve(\lambda) := \dist(\lambda ,\sigma(\Lscr_\ve^\bosfer))\,,\qquad \text{and} \qquad \widehat d_\ve \equiv \widehat d_\ve(\lambda) := \dist(\lambda ,\sigma(\widehat\Lscr_\ve^\bosfer))\,,
\]
one has 
\[\left\{\lambda \in [0,+\infty) \cap \rho(\widehat \Lscr_\ve^\bosfer):  
\frac{4\ve\delta}{\sqrt{\widehat d_\ve(\lambda)}} \,\sqrt{ 1 +\frac{\lambda}{\widehat d_\ve(\lambda)}} < \frac12\right\} \subseteq
\left\{\lambda  \in \rho(\Lscr_\ve^\bosfer):  
d_\ve(\lambda) \geq \frac{\widehat d_\ve(\lambda)}{32}\,\right\}.
\]
(ii) For any $\lambda\in\rho(\Lscr_\ve^\bosfer)\cap\rho(\widehat \Lscr_\ve^\bosfer)$, one has 
\[
\big\|(\Lscr_\ve^\bosfer - \lambda )^{-1} - (\widehat \Lscr_\ve^\bosfer - \lambda )^{-1} \big\|\leq   
2 \ve \delta \left( \frac{1}{\widehat d_\ve}\sqrt{\frac{1}{d_\ve}\left( 1 +\frac{\lambda_+}{d_\ve}\right)} + \frac{1}{d_\ve} \sqrt{ \frac{1}{\widehat d_\ve} \left( 1+ \frac{\lambda _+}{\widehat d_\ve}\right)}\ \right). 
\]
\end{lemma}
\begin{proof}
In Proposition \ref{p:diff} set $\psi = (\widehat\Lscr_\ve^\bosfer -\lambda )^{-1}\phi$, then 
\[\begin{aligned}
&\left| \|\phi\|^2
    -
    \big\langle( \Lscr_\ve^\bosfer - \lambda   ) \phi, (\widehat\Lscr_\ve^\bosfer -\lambda )^{-1}\phi  \big\rangle\right| \\ 
&\leq 
2 \ve \delta \Bigg( \|(\widehat\Lscr_\ve^\bosfer -\lambda )^{-1}\phi\| \sqrt{\left( \|(\Lscr_\ve^\bosfer -\lambda )\phi\| +\lambda_+ \,  \|\phi\|\right) \|\phi\|} \\ & \qquad+ \|\phi\| \sqrt{\left( \|\phi\| + \lambda _+ \,  \|(\widehat\Lscr_\ve^\bosfer -\lambda )^{-1}\phi\|\right) \|(\widehat\Lscr_\ve^\bosfer -\lambda )^{-1}\phi\|}\ \Bigg)\\
&\leq 
2 \ve \delta \|\phi\| \left( \frac{1}{\widehat d_\ve}\sqrt{\left( \|(\Lscr_\ve^\bosfer -\lambda )\phi\| +\lambda_+ \,  \|\phi\|\right) \|\phi\|} +\|\phi\| \sqrt{\frac{1}{\widehat d_\ve}\left( 1 +\frac{\lambda_+}{\widehat d_\ve} \right)}\ \right)\\
&\leq 
2 \ve \delta \|\phi\| \left( \frac{1}{\widehat d_\ve}\sqrt{\|(\Lscr_\ve^\bosfer -\lambda )\phi\|\|\phi\|}
+ \frac{ \|\phi\|}{\widehat d_\ve} \sqrt{\lambda_+ }
+\|\phi\| \sqrt{\frac{1}{\widehat d_\ve}\left( 1 +\frac{\lambda_+}{\widehat d_\ve} \right)}\ \right)\\
\\
&\leq 
2 \ve \delta \|\phi\| \left( \frac{1}{\widehat d_\ve}\sqrt{\|(\Lscr_\ve^\bosfer -\lambda )\phi\|\|\phi\|}
+2\|\phi\| \sqrt{\frac{1}{\widehat d_\ve}\left( 1 +\frac{\lambda_+}{\widehat d_\ve} \right)}\ \right).
\end{aligned}\]
The latter inequality gives 
\[
    \|\phi\|^2 \leq  \big|
    \big\langle( \Lscr_\ve^\bosfer - \lambda   ) \phi, (\widehat\Lscr_\ve^\bosfer -\lambda )^{-1}\phi  \big\rangle\big|+  2 \ve \delta \|\phi\| \left( \frac{1}{\widehat d_\ve}\sqrt{\|(\Lscr_\ve^\bosfer -\lambda )\phi\|\|\phi\|}
+2\|\phi\| \sqrt{\frac{1}{\widehat d_\ve}\left( 1 +\frac{\lambda_+}{\widehat d_\ve} \right)}\ \right).
\]
From which there follows, 
\[
    \|\phi\| \leq \frac{1}{\widehat d_\ve} \|( \Lscr_\ve^\bosfer -  \lambda   ) \phi\|+  2 \ve \delta  \left( \frac{1}{\widehat d_\ve}\sqrt{\|(\Lscr_\ve^\bosfer -\lambda )\phi\|\|\phi\|}
+2 \|\phi\|\sqrt{\frac{1}{\widehat d_\ve}\left( 1 +\frac{\lambda_+}{\widehat d_\ve} \right)}\ \right).
\]
Defining
\[
s : =\sqrt{ \frac{\|( \Lscr_\ve^\bosfer -  \lambda   ) \phi\|}{\widehat d_\ve \|\phi\|}},
\]
one has 
\[
\left(s+\frac{\ve\delta}{\sqrt{\widehat d_\ve}}\right)^2 - \frac{\ve^2\delta^2}{\widehat d_\ve} \geq 1 -  \frac{4\ve\delta}{\sqrt{\widehat d_\ve}} \sqrt{ 1 +\frac{\lambda_+}{\widehat d_\ve}} \, .
\]
If 
\[
\frac{4\ve\delta}{\sqrt{\widehat d_\ve}} \sqrt{ 1 +\frac{\lambda_+}{\widehat d_\ve}} < \frac12,
\]
then the following inequality must be satisfied
\[
s \geq -   \frac{\ve\delta}{\sqrt{\widehat d_\ve}} + \sqrt{ \frac{\ve^2\delta^2}{\widehat d_\ve}+ \frac12 } \geq \frac14 \frac{1}{\sqrt{ 1 + \frac{\ve^2\delta^2}{\widehat d_\ve}}}.
\]
Hence, 
\[
\|( \Lscr_\ve^\bosfer -  \lambda   ) \phi\| \geq \frac{\widehat d_\ve}{16} \frac{1}{ 1 + \frac{\ve^2\delta^2}{\widehat d_\ve}}\|\phi\|\geq \frac{\widehat d_\ve}{32}\|\phi\|>0
\]
which implies $\lambda \in \rho(\Lscr_\ve^\bosfer)$ and the bound on $d_\ve$. We remark that if $\lambda <0$, the statement is trivial since  $(-\infty,0)\subset  \rho(\Lscr_\ve^\bosfer)$.

To prove the bound on the resolvents difference, let  $\tilde\chi, \chi \in L^2_{\bosfer}(\RE^2)$,   
    \[
    \phi = (\Lscr_\ve^\bosfer -\lambda )^{-1}\tilde \chi \qquad \text{and}\qquad  
       \psi = (\widehat \Lscr_\ve^\bosfer - \lambda )^{-1}\chi,
    \]
    and apply Proposition \ref{p:diff} to obtain 
\[\begin{aligned}
&\left|\left\langle\tilde\chi, \big((\Lscr_\ve^\bosfer - \lambda )^{-1} - (\widehat \Lscr_\ve^\bosfer - \lambda )^{-1} \big) \chi \right\rangle\right| 
\\
\leq & 2 \ve \delta \Bigg( \|(\widehat \Lscr_\ve^\bosfer - \lambda )^{-1}\chi\| \sqrt{\left( \|\tilde \chi\| +\lambda_+ \,  \|\left(\Lscr_\ve^\bosfer -\lambda \right)^{-1}\tilde \chi\|\right) \|(\Lscr_\ve^\bosfer -\lambda )^{-1}\tilde \chi\|}  \\ & \qquad + \|(\Lscr_\ve^\bosfer -\lambda )^{-1}\tilde \chi\| \sqrt{\left( \|\chi\| + \lambda _+ \,  \|(\widehat \Lscr_\ve^\bosfer - \lambda )^{-1}\chi\|\right) \|(\widehat \Lscr_\ve^\bosfer - \lambda )^{-1}\chi\|}\ \Bigg)\\
\leq & 2 \ve \delta \left( \frac{1}{\widehat d_\ve}\sqrt{\frac{1}{d_\ve}\left( 1 +\frac{\lambda_+}{d_\ve}\right)} + \frac{1}{d_\ve} \sqrt{ \frac{1}{\widehat d_\ve} \left( 1+ \frac{\lambda _+}{\widehat d_\ve}\right)}\ \right)\|\chi\|\|\tilde \chi\| .
\end{aligned}
\]
To conclude, take 
\[\tilde\chi = \big((\Lscr_\ve^\bosfer - \lambda )^{-1} - (\widehat \Lscr_\ve^\bosfer - \lambda )^{-1} \big) \chi\]
so that 
\[
\Big\|\big((\Lscr_\ve^\bosfer - \lambda )^{-1} - (\widehat \Lscr_\ve^\bosfer - \lambda )^{-1} \big) \chi \Big\|
\leq  2 \ve \delta \left( \frac{1}{\widehat d_\ve}\sqrt{\frac{1}{d_\ve}\left( 1 +\frac{\lambda_+}{d_\ve}\right)} + \frac{1}{d_\ve} \sqrt{ \frac{1}{\widehat d_\ve} \left( 1+ \frac{\lambda _+}{\widehat d_\ve}\right)}\ \right)   \|\chi\|\,.
\]
\end{proof}
\section{The effective Hamiltonian\label{s:5}}
We begin by a rewriting of the quadratic form $\QQ^\bosfer_{\ve} (\PP\varphi,\PP \psi)$ associated to the operator $\widehat\Lscr^{\bosfer}_{\ve,\PP}$ given in Lemma \ref{l:spectrum-widehatL}, i.e., 
\[\QQ^\bosfer_{\ve} (\PP\varphi,\PP \psi) = \int_{\RE^2} \ve^2 {\partial_x (\overline{\psi^{BO} f_\varphi})}\partial_x (\psi^{BO} f_\psi) \d\x  +\int_\RE  q_x(\psi_{x}^{BO} ,\psi_{x}^{BO} )\overline{f_{\varphi}}f_{\psi}\,\d x .
\]
We point out the identity  
\[
 \int_{\RE^2}  {\partial_x (\overline{\psi^{BO} f_\varphi})}\partial_x (\psi^{BO} f_\psi) \d\x   = 
  \int_{\RE}  \overline{ f_\varphi'} f_\psi' \, \d x  
  +
   \int_{\RE}   \left(\int_{\RE} |\partial_x \psi^{BO}|^2  dy\right) \overline{ f_\varphi} f_\psi \, \d x,
\]
where we used $\int_{\RE} \psi^{BO}\partial_x \psi^{BO}   \d y =0$, which is a consequence of the normalization condition $\|\psi_{x}^{BO}\| = 1$. Also we notice that
\[
q_x(\psi_{x}^{BO} ,\psi_{x}^{BO} ) = -\lambda_{0}(x)+\alpha^2  \,.
\]
Hence, 
\[\QQ^\bosfer_{\ve} (\PP\varphi,\PP \psi) =
\int_{\RE} \big(\ve^2 \overline{f_\varphi'}  f_\psi'  +  \big( V + \ve^2 \,R \big) \overline{f_{\varphi}} f_{\psi} \big) \, \d x  
\]
with  
    \begin{equation*}
        V(x) :=
          -\lambda_0(x)+\alpha^2 = -   \left(\frac{ W\left(\frac{|\alpha| |x|}{2} e^{\frac{-|\alpha|  |x|}{2}} \right)}{|x|} + \frac{|\alpha|}{2}\right)^{\!\!2} +\alpha^2  \, ,
    \end{equation*}
    and
    \begin{equation*}
        R(x) := \int_\rr |\partial_x \psi^{BO}(x,y)|^2 \d y \, .
    \end{equation*}
\begin{remark} \label{r:UBO} Let us define
\[
L^{2}_{\bosfer}(\RE):=\{f\in L^{2}(\RE): f(x)=\bosferpm_{\bosfer} f(-x)\}\,,\qquad H^{1}_{\bosfer}(\RE):=
H^{1}(\RE)\cap L^{2}_{\bosfer}(\RE)
\] 
and
\[
U_{BO}: \ran(\mathcal P|L_{\bosfer}^{2}(\RE^{2}))\to L^{2}_{\bosfer}(\RE)\,,\qquad U_{BO}(\psi^{BO}f_{\phi}):=f_{\phi}\,.
\] 
By $\|\psi^{BO}_{x}\|=1$, 
\[
\langle \psi^{BO}f_{\phi},\psi^{BO}f_{\psi}\rangle=\int_{\RE}\langle 
\psi^{BO}_{x},\psi^{BO}_{x}\rangle \overline f_{\phi}(x)f_{\psi}(x)\,dx=\langle f_{\phi},f_{\psi}\rangle\,,
\]
i.e., $U_{BO}$ preserves the scalar product. It is a bijection with inverse $U_{BO}^{-1}f=\psi^{BO}f$; hence, $U_{BO}$ is unitary and 
\[
\ran(\mathcal P|L_{\bosfer}^{2}(\RE^{2}))\simeq L^{2}_{\bosfer}(\RE)\,.
\]
\end{remark}
By the previous Remark, $\QQ^\bosfer_{\ve} (\PP\,\cdot, \PP\,\cdot)$ identifies with the sesquilinear  form in $L^{2}_{\bosfer}(\RE)$ defined by
\[
D(\Qeff^{\bosfer}) := H_{\bosfer}^1(\RE) \times H_{\bosfer}^1(\RE) \,,\qquad 
\Qeff ^{\bosfer} (f,g ) :=
\int_{\RE} \ve^2 \overline{f'}  g'  +  \big( V + \ve^2 \,R \big) \overline{f} g  \, \d x  .
\]
Since $ V + \ve^2 \,R$ is a bounded potential, the associated self-adjoint Hamiltonian in $L_{\bosfer}^2(\RE)$   is given by 
    \begin{equation*}
      D(\Lscr^{\eff\bosfer}_\ve) =
       H_{\bosfer}^2(\RE):=H^{2}(\RE)\cap L^{2}_{\bosfer}(\RE)\qquad  
       \Lscr^{\eff\bosfer}_\ve  
       = -\ve^2 \frac{d^2}{dx^2}+ V + \ve^2 \,R.
    \end{equation*}
Now, we analyze the spectrum of 
\[
\Lscr^{\eff\bosfer}_\ve - \ve^2 R =  -\ve^2 \frac{d^2}{dx^2}+ V . 
\] 
We note that:
\begin{enumerate}[$(i)$]
\item $V$ is even on $\RE$;
\item $0 \leq V(x)<\frac{3 }{4}\alpha^2 $;
\item $V(0)= 0$, and $\lim_{x\to\pm\infty} V(x) = \frac{3 }{4}\alpha^2$;
\item $V$ is strictly increasing on $(0,+\infty)$;
\end{enumerate}
Furthermore, by 
\[
(W(y e^{-y})+y)^2 - 4 y^2 = - 8 y^3  + O(y^4)\,,\qquad y\ll1\,, 
\]
there follows 
\begin{equation}\label{seriesV}
V(x) =   |\alpha|^3|x|+ O(x^2)\,,\qquad |x|\ll1\,.   
\end{equation}
To proceed, we set $\Lambda := \ve^{-1}$ and define
    \begin{equation}\label{Ham}
        \Lrm_\Lambda^\bosfer:=\Lambda^2 \Lscr^{\eff\bosfer}_{1/\Lambda } - R\equiv -\frac{d^2}{dx^2} + \Lambda^2 V  \,.
    \end{equation}
Denoting by $\Lrm_\Lambda$ the self-adjoint operator in $L^{2}(\RE)$ defined by the same differential expression as $\Lrm_\Lambda^\bosfer$,  since $(V-\frac34\alpha^2)$ is bounded and vanishes as $|x|\to+\infty$ one has (see, e.g., \cite[Theorem 3.8.2]{Schechter}) 
\[
\sigma_{ess}(  \Lrm_\Lambda^\bosfer )\subseteq\sigma_{ess}(  \Lrm_\Lambda ) = 
\sigma_{ess}\!\left(-\frac{d^2}{dx^2}+\frac34\,\alpha^2\Lambda^{2}\right)= 
\left[\frac34\,\alpha^2\Lambda^2,+\infty\right)\,.
\]
Since  $\Lrm_\Lambda$ is in the limit point case at both ends $\pm\infty$, its eigenvalues are simple  (see, e.g., \cite[Section 10]{Weidmann}); hence the eigenvalues of $\Lrm_\Lambda^\bosfer$ are simple as well. \par
We denote by $\ell_{\Lambda,n}^\bosfer$, $n=0,1,2,\dots,$ the eigenvalues of $\Lrm_\Lambda^\bosfer$ below the essential spectrum, numbered in increasing order 
\[0<\ell^{\bosfer}_{\Lambda,0}<\ell^{\bosfer}_{\Lambda,1}<\cdots <\ell_{\Lambda,n}^\bosfer< \cdots. \]
\par
The behavior as $\Lambda\gg 1$ of  such eigenvalues   is provided in the following
\begin{theorem}\label{theorem_linearHam}
For any  fixed  integer $n\ge0$ and  sufficiently large $\Lambda$,  $\Lrm_\Lambda^\bosfer$ has at least $n+1$  simple isolated eigenvalues and
\begin{equation*}
        \ell_{\Lambda,n}^\bosfer = e_{n}^{\bosfer} \Lambda^{4/3} + O(\Lambda)\,,\qquad e^{\bos}_{n}:= e_{2n}\,,\quad e_{n}^{\fer}:=e_{2n+1}\,, 
    \end{equation*}
where $e_{k}$ denotes  the $(k+1)$-th eigenvalue of the closure $K^1$ of the essentially self-adjoint operator 
\begin{equation*} 
        D( \mathrm{\dot K^1}) =C^{\infty}_{0}(\RE)\,, \qquad  \mathrm{\dot K^1} := -\frac{d^2}{dx^2}  + |\alpha|^3|x|  .
    \end{equation*}
\end{theorem}
We prove this theorem by Barry Simon's approach presented in \cite{Sim}, where the author considered the  case with a smooth, bounded from below potential. We adapt his strategy to our situation, where the potential is not differentiable at the origin. The proof of  Theorem \ref{theorem_linearHam} follows from the combination of  Lemmata \ref{l:lowerbound} and \ref{l:upperbound}.
\par
For later use, we recall several properties of the eigenvalues and eigenfunctions of  ${\mathrm{K}}^1$. We refer to \cite[Sect. 6.10]{Schwinger} for the details (in particular, see  \cite[Eqs. 6.10.9, 6.10.10,   6.9.1, 6.9.2]{Schwinger}; in the notation there, $\sigma_{2n} = \overline{\overline{\sigma}}_{n+1}$ and $\sigma_{2n+1} = \overline{\sigma}_{n+1}$).
\begin{lemma}\label{r:airy}
The solutions $\phi_k\in L^{2}(\RE)$ of the  eigenvalue equations   
    \begin{equation*}
        {\mathrm{K}}^{1} \phi_k = e_k \phi_k\,, \qquad  k = 0,1,2,\dots\,.
    \end{equation*}
    can be written in terms of the Airy function $\Ai$ as
    \begin{equation}\label{phi-even}
        \phi_{2n}(x) = C_{2n} \Ai(|\alpha||x| + \sigma_{2k})\,,\qquad n = 0,1,2,\dots\,,
    \end{equation}
    \begin{equation}\label{phi-odd}
        \phi_{2n+1}(x) = C_{2n+1} \sgn(x) \Ai(|\alpha||x| + \sigma_{2n+1})\,,\qquad n = 0,1,2,\dots\,,
    \end{equation}
where the $C_k$'s are normalization constants,  
the $\sigma_{2n}$'s and $\sigma_{2n+1}$'s interlace, 
\[\dots<\sigma_{2n+1}<\sigma_{2n}
<\dots<\sigma_{1}<\sigma_{0}<0\,,
\]
and are the  extrema  and the zeros   of the Airy function respectively (i.e., $\Ai'(\sigma_{2n})=0$ and   $\Ai(\sigma_{2n+1}) =0$). The eigenvalues 
\[0<e_0<e_1<e_2<\dots<e_k<\dots. \]
are given by
\begin{equation*}
    e_k = |\sigma_k|\,\alpha^2\,.
\end{equation*}
\end{lemma}
\begin{remark}\label{ev/odd} Note that $L^{2}(\RE)=L^{2}_{\bos}(\RE)\oplus L^{2}_{\fer}(\RE)$, this is equivalent to decomposing any function in $L^2(\RE)$ in the sum of its even and odd parts. Hence,  
one has \[\mathrm{K}^{1}=\mathrm{K}^{1\bos}\oplus\mathrm{K}^{1\fer}\,.
\] 
By Lemma \ref{r:airy}, $\{e^{\bos}_{n},\phi_{n}^{\bos}\}:=\{ e_{2n},\phi_{2n}\}$ and $\{e^{\fer}_{n},\phi_{n}^{\fer}\}:=\{e_{2n+1},\phi_{2n+1}\}$ are the eigensystems of $\mathrm{K}^{1\bos}$ and $\mathrm{K}^{1\fer}$ respectively.
\end{remark}
\begin{remark}
For $x \gg 1$ the Airy function $\Ai$ behaves as  (see, e.g., \cite[Sec. 10.4]{AS})
    \begin{equation}\label{asym_Airy}
    \Ai(x) =\frac{1}{2 \sqrt{\pi} x^{1 / 4}}\ e^{-\frac{2}{3} x^{3 / 2}} \left(1 + O(x^{-3 / 2})\right)\, .
\end{equation}
\end{remark}
\begin{remark} By \cite{Hethcote}, one has the bounds
\[
-\left(\frac{3\pi}{8}\,(4n+3)+\frac32\arctan\frac5{18\pi(4n+3)}\ \right)^{2/3}\le \sigma_{2n+1}\le-\left(\frac{3\pi}{8}\,(4n+3)\right)^{2/3}\,.
\] 
Then, by $\sigma_{2n+1}<\sigma_{2n}<\sigma_{2(n-1)+1}$, one obtains
\[
-\left(\frac{3\pi}{8}\,(4n+3)+\frac32\arctan\frac5{18\pi(4n+3)}\ \right)^{2/3}\le \sigma_{2n}\le
-\left(\frac{3\pi}{8}\,(4n-1)
\right)^{2/3}.
\] 
\end{remark}
As a first step toward the proof of Theorem \ref{theorem_linearHam}, we define an auxiliary self-adjoint operator $\Lrm_\Lambda^{1\bosfer}$  to compare with $\Lrm_\Lambda^\bosfer$; it is defined as
\begin{equation} \label{scaled_eig}
         \Lrm^{1\bosfer}_\Lambda :=\Lambda^{4/3}U_{\Lambda}\mathrm{K}^{1\bosfer}U^{-1}_{\Lambda} \, ,
    \end{equation}
where  the unitary operator $U_{\Lambda}$ in  $L^{2}_{\bosfer}(\RE)$ is defined by $ \left(U_{\Lambda}f \right)(x) := \Lambda^{1/3}\, f(\Lambda^{2/3} x)$. For any $\psi\in C^{\infty}_{0}(\RE)\cap L^{2}_{\bosfer}(\RE)$ one has
    \begin{equation*}
  \Lrm^{1\bosfer}_\Lambda \psi(x)= -\psi''(x) + \Lambda^2  |\alpha|^3|x| \psi(x)\,.
    \end{equation*}
Notice that, by \eqref{scaled_eig}, the spectrum of  $\Lambda^{-4/3}\Lrm^{1\bosfer}_\Lambda$ is  independent of $\Lambda$; by Lemma \ref{r:airy} and Remark \ref{ev/odd}, the eigensystem of $\Lrm^{1\bosfer}_\Lambda$ is given by   $\{\Lambda^{4/3} e^{\bosfer}_{n}, U_{\Lambda}\phi^{\bosfer}_{n}\}$.\par
We introduce a cut-off function on a suitable scale of $\Lambda$. Let $j \in C^\infty_0(\RE)\cap L^{2}(\RE)$ with $0\leq j\leq1$, even,  and $j(x)=1$ for $|x| \leq 1$ and $j(x)=0$ for $|x| \geq 2$. Define:
    \begin{equation} \label{cutoff}
        \mathrm{J_1}(x) := j (\Lambda^{1/2} x) \, ,
    \end{equation}
and let 
    \begin{equation*} 
        \mathrm{J_0}(x) := \sqrt{1 - \mathrm{J^{2}_1}(x)} \qquad  ( \mathrm{J^{2}_0} +  \mathrm{J^{2}_1} = 1)  .
    \end{equation*}
We point out  that both $J_0$ and $J_1$ depend on $\Lambda$, to simplify the notation we omit this dependence. We also  remark that in the definition of $J_1$, see Eq. \eqref{cutoff}, any exponent between $1/3$ and $2/3$ would work. 

We first point out that, by the Taylor expansion of $V$, taking into account the fact that $J_1$ is supported on $|x|\leq 2/\Lambda^{1/2}$, there follows that   $\mathrm{J_1}\left( \Lrm_\Lambda^\bosfer -  \Lrm^{1\bosfer}_\Lambda\right) \mathrm{J_1} = \mathrm{J_1}\left(\Lambda^2 \left(V - |\alpha|^3|x| \right)\right) \mathrm{J_1}$ is a bounded operator in $L^{2}_{\bosfer}(\RE)$ and  its norm satisfies the bound 
    \begin{equation}\label{norm_estimate}
        \lVert \mathrm{J_1}\left( \Lrm_\Lambda^\bosfer -  \Lrm^{1\bosfer}_\Lambda\right) \mathrm{J_1}\rVert = \lVert \mathrm{J_1}\left(\Lambda^2 \left(V - |\alpha|^3|x| \right)\right) \mathrm{J_1}\rVert \leq C \lVert \mathrm{J_1} \left(\Lambda^2 x^2 \right)\mathrm{J_1}\rVert =  O\left(\Lambda\right) \, .
    \end{equation}

\subsection{Lower bound} To establish a lower bound for $\ell_{\Lambda,n}^\bosfer$, we use  the IMS (Ismagilov \cite{ismagilov}, Morgan \cite{morgan}, Morgan-Simon \cite{morgan-simon}) localization technique (see, for instance, \cite[Lemma 3.1]{Sim} or \cite[Chapter 11]{Cycon_etal}):
\begin{lemma}\label{l:lowerbound} Suppose that $\Lrm_\Lambda^\bosfer$ has $n+1$ eigenvalues below the essential spectrum. Then 
    \begin{equation}\label{lower_bound}
        \ell_{\Lambda,n}^\bosfer \geq \Lambda^{4/3} e^{\bosfer}_{n} + O(\Lambda) \qquad \Lambda\gg 1\,.
     \end{equation}
\end{lemma}
    \begin{proof}
     Denote     $\mathrm{P}^{1\bosfer}_{\Lambda}$  the orthogonal projection of the $n$-dimensional subspace spanned by the first $n$ eigenvectors of $\Lrm^{1\bosfer}_\Lambda$.
Our goal is to prove
    \begin{equation} \label{lb}
        \Lrm_\Lambda^\bosfer \geq e^{\bosfer}_{n} \Lambda^{4/3}  + {\mathrm{F}^\bosfer_1} + O(\Lambda) \, ,
    \end{equation}
where ${\mathrm{F}^\bosfer_1}$ is the  symmetric operator
\begin{equation*}
   {\mathrm{F}^\bosfer_1} := \mathrm{J_1}(\Lrm^{1\bosfer}_{\Lambda}- e_n\Lambda^{4/3})\mathrm{P}^{1\bosfer}_{\Lambda}\mathrm{J_1}  .
\end{equation*}
We remark that $  {\mathrm{F}^\bosfer_1}$ has finite rank at most equal to $n$. As usual, inequalities of the form of  \eqref{lb} have to be understood in the sense $(\psi,\Lrm_\Lambda^\bosfer\psi) \geq e^\bosfer_{n} \Lambda^{4/3}\|\psi\|^2  + (\psi,{\mathrm{F_{1}^{\bosfer}}}\psi)  + O(\Lambda)$ for all $\psi\in D(\Lrm_\Lambda^\bosfer)$.

Inequality \eqref{lb} implies \eqref{lower_bound}. This is easily seen by noticing that one can choose a normalized vector $\psi$ in the span of the first $n+1$ eigenvectors of $\Lrm_\Lambda^\bosfer$ such that $\psi \in \ker {\mathrm{F}^\bosfer_1}$. For such $\psi$ we obtain 
\[
\ell^{\bosfer}_{\Lambda, n} \geq \langle\psi, \Lrm_\Lambda^\bosfer \psi\rangle \geq e^\bosfer_{n}  \Lambda^{4/3} + O(\Lambda)
\]
which implies \eqref{lower_bound}.

Let us first recall the IMS localization. Since  $\mathrm{J_0}$ and $\mathrm{J_1}$ are smooth and  $\mathrm{J_0}^2+\mathrm{J_1}^2=1$, one has
\begin{equation*}
    \Lrm_{\Lambda}^{\bosfer} = \sum_{i=0}^{1} \mathrm{J_i} \Lrm_\Lambda^\bosfer  \mathrm{J_i} - \sum_{i=0}^{1}\left( \mathrm{J_i}'\right)^2 \, .
\end{equation*}
Therefore, we can re-write  $\Lrm_\Lambda^\bosfer$ as
    \begin{equation}\label{rewritten}
        \Lrm_\Lambda^\bosfer = \mathrm{J_0} \Lrm_\Lambda^\bosfer  \mathrm{J_0} + \mathrm{J_1}  \Lrm_\Lambda^{1\bosfer} \mathrm{J_1}+ \mathrm{J_1}\left( \Lrm_\Lambda^\bosfer  -\Lrm_\Lambda^{1\bosfer}\right) \mathrm{J_1}-\sum_{i=0}^1 \left( \mathrm{J_i}'\right)^2.
    \end{equation}
From the definition of $J_0$ and $J_1$, there follows the  bound on the operator norm:
 \begin{equation}\label{grad_cutoff}
        \left\|\sum_{i=0}^{1}\left( \mathrm{J_i}'\right)^2\right\| = O(\Lambda) \, .
  \end{equation}
By the definition of the orthogonal projection $P_\Lambda^{1\bosfer}$, there follows 
\begin{equation*}
  \Lrm^{1\bosfer}_\Lambda  =  \Lrm^{1\bosfer}_\Lambda \mathrm{P}^{1\bosfer}_{\Lambda} + \Lrm^{1\bosfer}_\Lambda\left(\mathrm{I}-\mathrm{P}^{1\bosfer}_{\Lambda}\right)
      =  {\mathrm{F}}^\bosfer  + e^\bosfer_{n} \Lambda^{4/3}\mathrm{P}^{1\bosfer}_{\Lambda}   + \Lrm^{1\bosfer}_\Lambda\left(\mathrm{I}-\mathrm{P}^{1\bosfer}_{\Lambda}\right) \geq 
       {\mathrm{F}}^\bosfer  + e^\bosfer_{n} \Lambda^{4/3},
\end{equation*}
    with 
    \begin{equation*}
   {\mathrm{F}^\bosfer} := (\Lrm^{1\bosfer}_{\Lambda}- e^\bosfer_{n}\Lambda^{4/3})\mathrm{P}^{1\bosfer}_{\Lambda} .
\end{equation*}
Hence, 
    \begin{equation} \label{3}
    \mathrm{J_1}  \Lrm^{1\bosfer}_\Lambda \mathrm{J_1}  \geq {\mathrm{F}^\bosfer_1} + e^\bosfer_{n} \Lambda^{4/3} \mathrm{J_1}^2.
    \end{equation}
The latter inequality can be understood as an inequality for $\langle\psi,    \mathrm{J_1}  \Lrm^{1\bosfer}_\Lambda \mathrm{J_1}\psi\rangle$ for all $\psi \in D(\Lrm_\Lambda^\bosfer)$, since $\psi \in D(\Lrm_\Lambda^\bosfer)$ implies $ \mathrm{J_1}\psi \in D(\Lrm_\Lambda^{1\bosfer})$.

Also, we need a control on the support of $\mathrm{J_0}$. Since $V$ is even and increasing for $x>0$, by the expansion \eqref{seriesV}, one gets
\[
V(x) \geq V(1/\Lambda^{1/2}) \geq  \frac{c}{\Lambda^{1/2}} \qquad\text{on $\supp(\mathrm{J_0})$},
\]
for some positive constant $c$ and  $\Lambda$ large enough. Hence, since $-\frac{d^{2} }{dx^{2}\,} $ is a positive definite operator, one has
\begin{equation}\label{4}
\mathrm{J_0} \Lrm^{\bosfer}_\Lambda  \mathrm{J_0} \geq c \Lambda^{3/2} \left(\mathrm{J_0}\right)^2 \geq e^\bosfer_{n} \Lambda^{4/3} \left(\mathrm{J_0}\right)^2 ,
\end{equation}
for $\Lambda$ large enough. Taking into account \eqref{rewritten}, together with  \eqref{norm_estimate},\eqref{grad_cutoff},\eqref{3} and \eqref{4}, the inequality \eqref{lb} follows   
\end{proof}

\subsection{Upper bound}
To obtain  the upper bound, we use the Rayleigh-Ritz variational principle. In this approach, we take the scaled eigenfunctions of  $\mathrm{K}^{1\bosfer}$ as trial wave functions for  $\Lrm^{1\bosfer}_\Lambda$.
\begin{lemma} \label{l:upperbound} For any  fixed  integer $n\ge 0$ and  sufficiently large $\Lambda$, $\Lrm_\Lambda^\bosfer$ has at least $n+1$  simple isolated eigenvalues and 
\begin{equation}\label{upper_bound}
  \ell_{\Lambda,n}^\bosfer \leq \Lambda^{4/3}e^\bosfer_{n} + O(\Lambda) .
\end{equation}
\end{lemma}
\begin{proof}
Recall that the functions $\phi^{\bosfer}_{n}$'s in \eqref{phi-even} and \eqref{phi-odd} are the orthonormal  eigenfunctions of  $\mathrm{K}^{1\bosfer}$ corresponding to the  eigenvalues $e^\bosfer_{n}$. Therefore, the $U_{\Lambda}\phi^{\bosfer}_{n}$'s are the orthonormal  eigenfunctions of  $\Lrm^{1\bosfer}_\Lambda$, that is
\begin{equation*}
\Lrm^{1\bosfer}_\Lambda \left(U_{\Lambda}\phi^{\bosfer}_{n}\right) = \Lambda^{4/3} e^\bosfer_{n} \left(U_{\Lambda}\phi^{\bosfer}_{n}\right) \, .
\end{equation*}
Defining
\begin{equation*}
        \psi^{\bosfer}_n := \mathrm{J_1} U_{\Lambda}\phi^{\bosfer}_{n} 
    \end{equation*}
and taking into account the asymptotic behavior of the Airy functions, see Eq.  \eqref{asym_Airy}, one has
    \begin{equation}\label{5}
        \langle \psi^{\bosfer}_n ,\psi^{\bosfer}_m \rangle = \delta_{mn} + O(e^{-c\Lambda^{1/4}}) ,
    \end{equation}
    for some positive constant $c$. 
Furthermore, by 
\[
\mathrm{J_1}\Lrm^{1\bosfer}_\Lambda\mathrm{J_1} = \frac12 \left( \mathrm{J_1^2}\Lrm^{1\bosfer}_\Lambda + \Lrm^{1\bosfer}_\Lambda \mathrm{J_1^2}\right) + (\mathrm{J_1'})^2\,,
\] 
by \eqref{5} and \eqref{grad_cutoff}, one obtains
\begin{equation}\begin{aligned}\label{6}
        \langle \psi^{\bosfer}_n ,\Lrm^{1\bosfer}_\Lambda \psi^{\bosfer}_m \rangle =& \Lambda^{4/3} \left(\frac{e^\bosfer_{n} + e^{\bosfer}_{m}}{2}\right) \langle\psi^{\bosfer}_n ,\psi^{\bosfer}_m \rangle + \big( U_{\Lambda}\phi^{\bosfer}_{n} ,\left(\mathrm{J_1'}\right)^2U_{\Lambda}\phi^{\bosfer}_{m} \big) \\ =& \Lambda^{4/3} e^\bosfer_{n}\delta_{nm} + O(\Lambda)\, .
        \end{aligned}
    \end{equation}
Therefore, by \eqref{norm_estimate},   
    \begin{equation*}
    \langle \psi^{\bosfer}_n ,\Lrm_\Lambda^\bosfer\psi^{\bosfer}_m \rangle  = \langle \psi^{\bosfer}_n ,\Lrm^{1\bosfer}_\Lambda  \psi^{\bosfer}_m \rangle + O\left(\Lambda\right) 
         =  \Lambda^{4/3} e^\bosfer_{n} \delta_{mn} + O(\Lambda)\, .
    \end{equation*}
Let us fix $n \ge 0$. We want to apply the Rayleigh-Ritz variational method, however, since the vectors $\{\psi^{\bosfer}_i\}_{i=0}^n$ are not orthonormal (but almost orthonormal, see Eq. \eqref{5}), we proceed as follows. Noticing that  the $\psi_i^{\bosfer}$'s are linearly independent, by the Gram-Schmidt algorithm  we  construct a set of orthonormal eigenfunctions $\{\tilde \psi_i^{\bosfer}\}_{i=0}^n$ which, by construction, still satisfy Eq. \eqref{6}, that is   
\begin{equation}\label{tildepsi}
        \langle\tilde \psi^{\bosfer}_n , \Lrm^{1\bosfer}_\Lambda  \tilde\psi^{\bosfer}_m \rangle =\Lambda^{4/3} e^\bosfer_{n}\delta_{nm} + O(\Lambda)\, .
    \end{equation}
 Then, a direct application of the Rayleigh-Ritz variational method (see e.g. \cite{ReSi}, Theorem XIII.3), taking as trial space the span of $\{\tilde \psi^{\bosfer}_i\}_{i=1}^n$, gives the upper bound $\ell_{\Lambda,n}^\bosfer\leq \Lambda^{4/3} e^\bosfer_{n} + O(\Lambda)$. The fact that  $\sigma_{ess}(  \Lrm_\Lambda^\bosfer ) \subseteq [\frac34\alpha^2\Lambda^2,+\infty)$ and Eq. \eqref{tildepsi} guarantee that there are at least $n+1$ eigenvalues.
\end{proof}
Our main result on the eigenvalues of $\Lscr_\ve^{\eff\bosfer}$ is given in the following 
\begin{theorem}\label{th:eigenvalues}
For any  fixed  integer $n\ge 0$ and $\ve>0$ sufficiently small, the $\Lscr^{\eff\bosfer}_{\ve}$ has at least $n+1$  simple isolated eigenvalues. The $(n+1)$-th eigenvalue is given by 
        \begin{equation}\label{Eeff-ae}
            \eigenvalueL^{\eff\bosfer}_{\ve, n} = s^{\bosfer}_{n}\alpha^2 \ve^{2/3} + O (\ve) \, ,
        \end{equation}
    where $s^{\bos}_{n}:=|\sigma_{2n}|$ and $s_{n}^{\fer}:=|\sigma_{2n+1}|$ and the negative numbers  $\sigma_{k}$ are defined in Lemma \ref{r:airy}.
\end{theorem}
\begin{proof}
Denote by $\tilde \ell_{\Lambda,n}^\bosfer$ the eigenvalues of the self-adjoint operator (compare the following definition with Eq. \eqref{Ham}) 
\[
\widetilde \Lrm_\Lambda^\bosfer:=\Lambda^2 \Lscr^{\eff\bosfer}_{1/\Lambda } = \Lrm_\Lambda^\bosfer +R.
\] 
The eigenvalues  $\tilde \ell_{\Lambda,n}^\bosfer$ satisfy the lower bound \eqref{lower_bound}, because $R$ is positive. Moreover, the $\tilde \ell_{\Lambda,n}^\bosfer$'s satisfy the upper bound \eqref{upper_bound}, because in the Rayleigh-Ritz variational approach (see Eq. \eqref{tildepsi})  $R$ gives a contrubution of order $1$. Hence, for any fixed integer $n\ge 0$,  $\widetilde \Lrm_\Lambda^\bosfer$ has at least $n+1$ simple isolated eigenvalues, and $\tilde \ell_{\Lambda,n}^\bosfer  =  \Lambda^{4/3}e^\bosfer_{n} + O(\Lambda)$. Noticing that $ \Lscr^{\eff\bosfer}_{\ve} = \ve^2  \widetilde \Lrm^{\bosfer}_{1/\ve}$ we obtain the expansion \eqref{Eeff-ae}.
\end{proof}

\section{Proof of Theorem \ref{th:main}\label{s:6}}
The result in Eq. \eqref{spectrum1} about the essential spectrum is part of Theorem \ref{th:spectrumHve}. The lower bound on the spectrum in the same equation follows immediately from  Theorem \ref{th:lowerboundbH} and Eq. \eqref{directintegralB}. 

Fix $n\ge 0$. By Theorem \ref{th:eigenvalues}, one can take  $\ve$ so small that $\Lscr_\ve^{\eff\bosfer}$ has $n+1$ eigenvalues, $\eigenvalueL^{\eff\bosfer}_{\ve,0}, \eigenvalueL^{\eff\bosfer}_{\ve,1},\dots, \eigenvalueL^{\eff\bosfer}_{\ve,n}$.  These are also eigenvalues for $\widehat \Lscr^{\bosfer}_{\PP,\ve}$, as a matter of fact they are the lowest eigenvalues (ordered in increasing order) of  $\widehat \Lscr^{\bosfer}_{\PP,\ve}$, 
and  of $\widehat \Lscr^{\bosfer}_{\ve}$, 
 see Lemma  \ref{l:spectrum-widehatL}. Choose $c_n>0$ so that $\eigenvalueL^{\eff\bosfer}_{\ve,n-1}<\mu^{\bosfer}_{\ve,n}< \eigenvalueL^{\eff\bosfer}_{\ve,n}$ and $\widehat d_\ve(\mu^{\bosfer}_{\ve,n}) > c_n \ve^{2/3}$, where $\mu^{\bosfer}_{\ve,n} := \eigenvalueL^{\eff\bosfer}_{\ve,n} - c_n \ve^{2/3} $.  Since $( \mu^{\bosfer}_{\ve,n}-\eigenvalueL^{\eff\bosfer}_{\ve,n} )^{-1}$ is the lowest eigenvalue of  $(\mu^{\bosfer}_{\ve,n} - \widehat \Lscr_\ve^\bosfer)^{-1}$, 
\[
(\mu^{\bosfer}_{\ve,n} - \eigenvalueL^{\eff\bosfer}_{\ve,n})^{-1}=\min_{\psi\in L^{2}_{\bosfer}(\RE^{2}) ,\|\psi\|=1} \big\langle\psi, (\mu^{\bosfer}_{\ve,n}-\widehat \Lscr_\ve^\bosfer)^{-1}\psi\big\rangle \,.
\] 

By Lemma \ref{l:resolvent-difference} (i), possibly for smaller $\ve$, there holds $\mu^{\bosfer}_{\ve,n} \in \rho(\Lscr_\ve^\bosfer)$ and $d_\ve(\mu^{\bosfer}_{\ve,n})> c'_n \ve^{2/3}$. 
Let us define $\eigenvalueL^{\bosfer}_{\ve,n}$ through the relation
\[
(\mu^{\bosfer}_{\ve,n} - \eigenvalueL^{\bosfer}_{\ve,n} )^{-1}=\inf_{\psi\in  L^{2}_{\bosfer}(\RE^{2}),\|\psi\|=1}  \big\langle\psi, (\mu^{\bosfer}_{\ve,n}- \Lscr_\ve^\bosfer)^{-1}\psi\big\rangle\,.
\] 
By the spectral mapping theorem, either $\eigenvalueL^{\bosfer}_{\ve,n}$ is an eigenvalue
of $\Lscr_\ve^\bosfer$ or $\eigenvalueL^{\bosfer}_{\ve,n} =\inf\sigma_{ess}(\Lscr_\ve^\bosfer)$, where, by Theorem \ref{th:spectrumHve}, $\inf\sigma_{ess}(\Lscr_\ve^\bosfer)=\frac{3+\ve^{2}}{4+\ve^{2}}\,\alpha^{2}$.

By the bound on the resolvent difference in Lemma \ref{l:resolvent-difference} (ii), there exists $c''_{n}>0$ such that 
\[
\left|(\eigenvalueL^{\bosfer}_{\ve,n}-\mu^{\bosfer}_{\ve,n})^{-1}- (\eigenvalueL^{\eff\bosfer}_{\ve,n} - \mu^{\bosfer}_{\ve,n})^{-1}\right|\le c''_{n}\,.
\]
Therefore, 
\begin{align*}
|\eigenvalueL^{\bosfer}_{\ve,n}-\eigenvalueL^{\eff\bosfer}_{\ve,n}|\le& c''_{n}\,|\eigenvalueL^{\bosfer}_{\ve,n}-\mu^{\bosfer}_{\ve,n}|\,|\eigenvalueL^{\eff\bosfer}_{\ve,n} - \mu^{\bosfer}_{\ve,n}|\le c''_{n}c_{n}\ve^{2/3}|\eigenvalueL^{\bosfer}_{\ve,n}-\eigenvalueL^{\eff\bosfer}_{\ve,n}+c_{n}\ve^{2/3}|\\
\le &c_{n}''c_{n}\ve^{2/3}|\eigenvalueL^{\bosfer}_{\ve,n}-\eigenvalueL^{\eff\bosfer}_{\ve,n}|+c_{n}''c^{2}_{n}\ve^{4/3}\,.
\end{align*}
This inequality and Theorem \ref{th:eigenvalues} give 
\[
\eigenvalueL^{\bosfer}_{\ve,n}=\eigenvalueL^{\eff\bosfer}_{\ve,n}+O(\ve^{4/3})= s_{n}^{\bosfer}\alpha^2\ve^{2/3} +O(\ve)\,.
\]
This shows that $\eigenvalueL^{\bosfer}_{\ve,n}<\inf\sigma_{ess}(\Lscr_\ve^\bosfer)$ whenever $\ve$ is sufficiently small; hence, $\eigenvalueL^{\bosfer}_{\ve,n}$ is an eigenvalue below the essential spectrum of $\Lscr_\ve^\bosfer$. One can repeat the argument above for any $k=0,\dots,n$, and, by the inequalities  
\[
\eigenvalueL^{\eff\bosfer}_{\ve,k} -\tilde c_k \ve^{4/3}  \leq 
\eigenvalueL^{\bosfer}_{\ve,k}\leq \eigenvalueL^{\eff\bosfer}_{\ve,k} +\tilde c_k \ve^{4/3}\qquad k=0,\dots,n
\]
\[
 \eigenvalueL^{\eff\bosfer}_{\ve,k-1}+\tilde c_{k-1} \ve^{4/3}  \leq 
\eigenvalueL^{\eff\bosfer}_{\ve,k} -\tilde c_k \ve^{4/3} \qquad k=1,\dots,n
\]
which hold true for $\ve$ small enough, there follows that $\eigenvalueL^{\bosfer}_{\ve,n}$ is the $(n+1)$-th eigenvalue.   The proof is then concluded by $\sigma_{d}(H^\bosfer_\ve)=\sigma_{d}(\Lscr_\ve^\bosfer)-\alpha^{2}$.

\vskip10pt\noindent
{\bf Acknowledgements.}
The authors express their sincere gratitude to their  friend and collaborator Rodolfo Figari; the enlightening and stimulating discussions we had during his visit to our Department greatly contributed to the development of this work.

We are grateful to Nicholas Raymond who kindly drew our attention  to Ref. \cite{DucheneRaymond14}. 

The authors acknowledge funding from the Next Generation EU-Prin 2022 project 2022CHELC7 ``Singular Interactions and Effective Models in Mathematical Physics'' and the support of the National Group of Mathematical Physics (GNFM-INdAM). 

\vskip10pt\noindent
{\bf Conflict of Interest.} The authors have no conflicts to disclose.

\vskip10pt\noindent
{\bf Data Availability Statement.} Data sharing is not applicable to this article as no new data were created or analyzed in this study.


\begin{thebibliography}{10}

\bibitem{AS}
M.~Abramowitz and I.~A. Stegun.
\newblock {\em Handbook of Mathematical Functions with Formulas, Graphs, and
  Mathematical Tables}, volume No. 55 of {\em National Bureau of Standards
  Applied Mathematics Series}.
\newblock U. S. Government Printing Office, Washington, DC, 1964.
\newblock For sale by the Superintendent of Documents.

\bibitem{AT1}
H.~Akbas, O.~T. Turgut.
\newblock {Born-Oppenheimer approximation for a singular system}.
\newblock {\em J. Math. Phys.}, 59, 012107, 2018. \href{https://doi.org/10.1063/1.5021364}{10.1063/1.5021364.}

\bibitem{AT2}
H.~Akbas, O.~T. Turgut.
\newblock {Born-Oppenheimer approximation for a simple renormalizable system}. In: A. Michelangeli (ed.), {\it Mathematical Challenges of Zero-Range Physics},
Springer INdAM Series 42, 2021.

\bibitem{AGHKH05}
S.~Albeverio, F.~Gesztesy, R.~H{\o}egh-Krohn, and H.~Holden.
\newblock {\em Solvable Models in Quantum Mechanics}.
\newblock AMS Chelsea Publishing, Providence, RI, second edition, 2005.
\newblock With an appendix by Pavel Exner.

\bibitem{BCFTjmp18}
G.~Basti, C.~Cacciapuoti, D.~Finco, and A.~Teta.
\newblock {The three-body problem in dimension one: From short-range to contact
  interactions}.
\newblock {\em J. Math. Phys.}, 59:072104, 2018. \href{https://doi.org/10.1063/1.5030170}{10.1063/1.5030170.}

\bibitem{BCFTahp23}
G.~Basti, C.~Cacciapuoti, D.~Finco, and A.~Teta.
\newblock Three-body Hamiltonian with regularized zero-range interactions in dimension three. 
\newblock {\em Ann. Henri Poincar\'e}, 24, 223-276, 2023. \href{https://doi.org/10.1007/s00023-022-01214-9}{10.1007/s00023-022-01214-9.}

\bibitem{Becker2018}
S.~Becker, A.~Michelangeli, and A.~Ottolini.
\newblock Spectral analysis of the 2 + 1 fermionic trimer with contact interactions.
\newblock {\em Math. Phys. Anal. Geom.}, 21(4):35, 2018. \href{https://doi.org/10.1007/s11040-018-9294-0}{10.1007/s11040-018-9294-0.}

\bibitem{BO27}
M.~Born and R.~Oppenheimer.
\newblock {Zur Quantentheorie der Molekeln}.
\newblock {\em Ann. Phys.}, 389: 457-484, 1927. \href{https://doi.org/10.1002/andp.19273892002}{10.1002/andp.19273892002.}

\bibitem{BGP}
J. Br\"uning,  V. Geyler, K. Pankrashkin.
\newblock{Spectra of self-adjoint extensions and applications to solvable Schr\"odinger operators.} 
{\em Rev. Math. Phys.}, 20 (2008), 1-70. \href{https://doi.org/10.1142/s0129055x08003249}{10.1142/s0129055x08003249.}

\bibitem{CFPrma18}
C.~Cacciapuoti, D.~Fermi, and A.~Posilicano.
\newblock {On inverses of Kre{\v \i}n's $\mathscr{Q}$-functions}.
\newblock {\em Rend. Mat. Appl. (7)}, 39(2): 229--240, 2018.

\bibitem{CDS} J. M. Combes, P. Duclos, R. Seiler.  
\newblock The Born-Oppenheimer approximation.
\newblock In: Wightman, Velo (eds.), {\em Rigorous Atomic and Molecular Physics Proceedings}, vol. 1980, pp. 185. Plenum, New York 1981.

\bibitem{CDFMT12}
M.~Correggi, G.~Dell'Antonio, D.~Finco, A.~Michelangeli, and A.~Teta.
\newblock Stability for a system of $N$ fermions plus a different particle with zero-range interactions.
\newblock {\em Rev. Math. Phys.}, 24(07):1250017, 2012. \href{https://doi.org/10.1142/s0129055x12500171}{10.1142/s0129055x12500171.}

\bibitem{CDFMT15}
M.~Correggi, G.~Dell'Antonio, D.~Finco, A.~Michelangeli, and A.~Teta.
\newblock A class of Hamiltonians for a three-particle fermionic system at unitarity.
\newblock {\em Math. Phys. Anal. Geom.}, 18(1):32, 2015. \href{https://doi.org/10.1007/s11040-015-9195-4}{10.1007/s11040-015-9195-4.}

\bibitem{CFT15}
M.~Correggi, D.~Finco, and A.~Teta.
\newblock Energy lower bound for the unitary $N+1$ fermionic model.
\newblock {\em Europhysics Letters}, 111(1):10003, 2015. \href{https://doi.org/10.1209/0295-5075/111/10003}{10.1209/0295-5075/111/10003.}

\bibitem{Cycon_etal}
H.~L. Cycon, R.~G. Froese, W.~Kirsch, and B.~Simon.
\newblock {\em Schr\"odinger Operators with Application to Quantum Mechanics
  and global geometry}.
\newblock Texts and Monographs in Physics. Springer-Verlag, Berlin, study
  edition, 1987.
  
\bibitem{DucheneRaymond14}
V.~Duch\^{e}ne and N.~Raymond.
\newblock Spectral asymptotics of a broken $\delta$-interaction.
\newblock {\em  J. Phys. A: Math. Theor.}, 47(15):155203, 2014. \href{https://doi.org/10.1088/1751-8113/47/15/155203}{10.1088/1751-8113/47/15/155203.}

\bibitem{ExnerIchinose01}
P.~Exner and T.~Ichinose.
\newblock Geometrically induced spectrum in curved leaky wires.
\newblock {\em J. Phys. A: Math. Gen.}, 34(7):1439, 2001. \href{https://doi.org/10.1088/0305-4470/34/7/315}{10.1088/0305-4470/34/7/315.}

\bibitem{ExnerNemcova03}
P.~Exner and K.~N\v{e}mcov\'a.
\newblock Leaky quantum graphs: approximations by point-interaction Hamiltonians.
\newblock {\em J. Phys. A: Math. Theor.}, 36(40): 10173, 2003. \href{https://doi.org/10.1088/0305-4470/36/40/004}{10.1088/0305-4470/36/40/004.}

\bibitem{FT23}
D.~Ferretti and A.~Teta.
\newblock Some remarks on the regularized Hamiltonian for three bosons with contact interactions.
\newblock In: {\em M. Correggi and M. Falconi (eds) ``Quantum Mathematics I''.} INdAM 2022. Springer INdAM Series, vol 57. Springer, Singapore, 2023.

\bibitem{FT24b}
D.~Ferretti and A.~Teta.
\newblock{Hamiltonian for a Bose gas with contact interactions.}
\newblock {\em Rev. Math. Phys.}, 38(2): 2550020, 2026.  \href{https://doi.org/10.1142/S0129055X25500205}{10.1142/S0129055X25500205.}

\bibitem{FT24a}
D.~Ferretti and A.~Teta.
\newblock{Hamiltonians for quantum systems with contact interactions.}
\newblock In: {\em D. Alpay, I. Sabadini, F. Colombo (eds) ``Operator Theory''.} Springer, Basel, 2025. \href{https://doi.org/10.1007/978-3-0348-0692-3_92-1}{10.1007/978-3-0348-0692-3\_92-1.}

\bibitem{FST24}
R.~Figari, H,~Saberbaghi, and A.~Teta.
\newblock {On a family of finitely many point interaction Hamiltonians free of ultraviolet pathologies.} 
\newblock {J. Phys. A: Math. Theor.}, 57(5):055303, 2024. \href{https://doi.org/10.1088/1751-8121/ad1ac9}{10.1088/1751-8121/ad1ac9.}

\bibitem{GHL20}
M.~Griesemer, M.~Hofacker, and U. Linden.
\newblock From short-range to contact interactions in the 1d Bose gas. 
\newblock {\em Math. Phys. Anal. Geom.}, 23:19, 2020. \href{https://doi.org/10.1007/s11040-020-09344-4}{10.1007/s11040-020-09344-4.}

\bibitem{HJ07}
G.~Hagedorn and A.~Joye.
\newblock Mathematical analysis of Born-Oppenheimer approximations.
\newblock In {\em Spectral Theory and Mathematical Physics: a Festschrift in Honor of Barry Simon's 60th Birthday}, Proc. Sympos. Pure Math. 76:203-226, AMS, Providence, RI,  2007.

\bibitem{Hethcote}
H.~W. Hethcote.
\newblock Bounds for zeros of some special functions.
\newblock {\em Proc. Amer. Math. Soc.}, 25:72--74, 1970. \href{https://doi.org/10.1090/s0002-9939-1970-0255909-x}{10.1090/s0002-9939-1970-0255909-x.}

\bibitem{ismagilov}
R.~S. Ismagilov.
\newblock Conditions for the semiboundedness and discreteness of the spectrum
  in the case of one-dimensional differential operators.
\newblock {\em Dokl. Akad. Nauk SSSR}, 140:33--36, 1961.

\bibitem{jecko14}
T.~Jecko.
\newblock On the mathematical treatment of the Born-Oppenheimer approximation.
\newblock {\em J. Math. Phys.}, 55(5):053504, 2014. \href{https://doi.org/10.1063/1.4870855}{10.1063/1.4870855.}

\bibitem{Krejcirik-etal-m2018}
D.~Krej\v{c}i\v{r}\'ik, N.~Raymond, J.~Royer, and P.~Siegl.
\newblock Reduction of dimension as a consequence of norm-resolvent convergence
  and applications.
\newblock {\em Mathematika}, 64(2):406--429, 2018. \href{https://doi.org/10.1112/s0025579318000013}{10.1112/s0025579318000013.}

\bibitem{KJF} A. Kufner, O. John, S. Fucik.  
\newblock {\em  Function Spaces.}
\newblock Noordhoff Int. Publ. 1977.

\bibitem{MF62a}
R.~A. Minlos and L.~Faddeev.
\newblock On the point interaction for a three-particle system in quantum mechanics.
\newblock {\em Soviet Phys. Dokl.}, 6(12):1072-1074, 1962.

\bibitem{MF62b}
R.~A. Minlos and L.~Faddeev.
\newblock Comment on the problem of three particles with point interactions.
\newblock {\em Soviet Phys. JETP}, 14(6):1315-1316, 1962.

\bibitem{morgan}
J.~D. Morgan.
\newblock Schr\"odinger operators whose potentials have separated
  singularities.
\newblock {\em J. Oper. Theory}, 1(1):109--115, 1979. 

\bibitem{morgan-simon}
J.~D. Morgan and B.~Simon.
\newblock Behavior of molecular potential energy curves for large nuclear separations.
\newblock {\em Int. J. Quantum Chem.}, 17:1143--1166, 1980. \href{https://doi.org/10.1002/qua.560170609}{10.1002/qua.560170609.}

\bibitem{PankrashkinVogel22}
K.~Pankrashkin and M.~Vogel.
\newblock On Schr\"odinger operators with $\delta'$-potentials supported on star graphs.
\newblock {\em J. Phys. A: Math. Theor.}, 55(29):295201, 2022. \href{https://doi.org/10.1088/1751-8121/ac775a}{10.1088/1751-8121/ac775a.}

\bibitem{Pjfa01}
A.~Posilicano.
\newblock A {K}re\u{\i}n-like formula for singular perturbations of
  self-adjoint operators and applications.
\newblock {\em J. Funct. Anal.}, 183(1):109--147, 2001. \href{https://doi.org/10.1006/jfan.2000.3730}{10.1006/jfan.2000.3730.}

\bibitem{ReSi}
M.~Reed and B.~Simon.
\newblock {\em Analysis of Operators}, volume~4 of {\em Methods of Modern Mathematical Physics}.
\newblock Academic Press, New York, 1978.

\bibitem{Saberbaghi2025}
H.~Saberbaghi.
\newblock{Efimov effect in the Born-Oppenheimer picture.}
\newblock{\href{https://arxiv.org/abs/2511.15973}{ArXiv:2511.15973 [math-ph]}}, 36 pp,  2025.

\bibitem{Schechter}
M.~Schechter.
\newblock {\em Operator Methods in Quantum Mechanics}.
\newblock Dover Publications, Inc., Mineola, NY, 2002.
\newblock Reprint of the 1981 original.

\bibitem{Schwinger}
J.~Schwinger.
\newblock {\em Quantum Mechanics. Symbolism of Atomic Measurements}.
\newblock Springer-Verlag, Berlin, 2001.

\bibitem{Sim}
B.~Simon.
\newblock Semiclassical analysis of low lying eigenvalues. {I}. {N}ondegenerate
  minima: asymptotic expansions.
\newblock {\em Ann. Inst. H. Poincar\'e{} Sect. A (N.S.)}, 38(3):295--308,
  1983. Errata: 40(2):224, 1984. 
  
\bibitem{Stone}
M. H. Stone.
\newblock{\it Linear Transformations in Hilbert Space.} 
\newblock American Mathematical Society. New York, 1932.

\bibitem{Teta}
A. Teta. 
\newblock{\em A Mathematical Primer on Quantum Mechanics.}
\newblock Unitext for Physics. Springer, Cham, 2018

\bibitem{thomas}
L.~E. Thomas.
\newblock Multiparticle Schr\"odinger Hamiltonians with point interactions.
\newblock {\em Phys. Rev. D}, 30, 1233,1984. \href{https://doi.org/10.1103/PhysRevD.30.1233}{10.1103/PhysRevD.30.1233.}

\bibitem{Weidmann}
J. Wiedmann. 
\newblock{\em  Spectral Theory of Ordinary Differential Operators}, Lecture Notes in Mathematics 1258.
\newblock Springer-Verlag, Berlin, 1987.

\bibitem{Yoshitomi17}
K.~Yoshitomi.
\newblock Finiteness of the discrete spectrum in a three-body system with point interaction.
\newblock {\em Mathematica Slovaca}, 67(4):1031--1042, 2017. \href{https://doi.org/10.1515/ms-2017-0030}{10.1515/ms-2017-0030.}

\end{thebibliography}
\end{document}